\pdfoutput=1
\documentclass[journal]{IEEEtran}

\ifCLASSINFOpdf
\else
\fi
\usepackage[mathscr]{eucal} 
\usepackage{amsmath,amsthm,amscd,amssymb,amsbsy}
\usepackage{graphicx}
\usepackage[caption=false,font=footnotesize]{subfig}
\usepackage{epstopdf}
\usepackage{algorithm}
\usepackage{algorithmic}
\usepackage{verbatim}
\usepackage{array}
\usepackage{cite}
\usepackage{color}

\def\a{{\mathbf a}}

\def\d{{\mathbf d}}
\def\n{{\mathbf n}}

\def\s{{\mathbf s}}
\def\u{{\mathbf u}}
\def\x{{\mathbf x}}
\def\y{{\mathbf y}}

\def\h{{\mathbf h}}
\def\bv{{\mathbf v}}

\def\hx{{\hat{\mathbf x}}}
\def\ha{{\hat{\mathbf a}}}
\def\hs{{\hat{\mathbf s}}}

\def\hT{{\hat T}}
\def\hK{{\hat K}}

\def\tn{{\tilde{\mathbf n}}}

\def\A{{\mathbf A}}
\def\B{{\mathbf B}}
\def\C{{\mathbf C}}
\def\D{\mathbf{D}}

\def\H{{\mathbf H}}
\def\I{{\mathbf I}}
\def\L{{\mathbf L}}

\def\Q{{\mathcal Q}} 

\def\U{{\mathbf U}}
\def\W{{\mathbf W}}
\def\X{{\mathbf X}}

\def\cO{{\mathcal O}}
\def\bS{{\mathbf S}}

\def\Ct{\widetilde{\mathbf C}}
\def\ct{\widetilde{c}}

\def\bLambda{{\boldsymbol \Lambda}}

\def\balpha{{\boldsymbol \alpha}}

\def\cX{{\cal X}}

\def\cG{{\cal G}}
\def\cV{{\cal V}}
\def\cE{{\cal E}}
\def\cN{{\cal N}}
\def\cS{{\cal S}}
\def\cU{{\cal U}}

\def\cF{{\cal F}}
\def\cC{{\cal C}}

\def\sR{{\mathbb R}}
\def\sC{{\mathbb C}}

\def\ie{\emph{i.e.}}

\newtheorem{theorem}{Theorem}
\newtheorem{proposition}{Proposition}

\newtheorem{definition}{Definition}

\DeclareMathOperator*{\diag}{diag} 
\DeclareMathOperator*{\argmin}{argmin} 

\newcolumntype{C}[1]{>{\centering}p{#1}}

\graphicspath{{./figures/}}
\begin{document}
%
\title{Fast Graph Sampling Set Selection Using Gershgorin Disc Alignment}
%
%
%

\author{Yuanchao~Bai,~\IEEEmembership{Student Member,~IEEE,}
        Fen~Wang,
        Gene~Cheung,~\IEEEmembership{Senior Member,~IEEE,}\\
        Yuji~Nakatsukasa,
        Wen~Gao,~\IEEEmembership{Fellow,~IEEE}
\thanks{Yuanchao~Bai and Wen~Gao are with the School of Electronics Engineering and Computer Science,
Peking University, Beijing, 100871, China.
(e-mail: \{yuanchao.bai, wgao\}@pku.edu.cn).}
\thanks{Gene~Cheung is the \textit{corresponding author} and is with the Department of Electrical Engineering \& Computer Science, York University, Toronto, Canada. (e-mail: genec@yorku.ca).}
\thanks{Fen~Wang is with State Key Laboratory of ISN, Xidian University, Xi'an 710071, Shaanxi, China (e-mail: fenwang@stu.xidian.edu.cn)}
\thanks{Yuji Nakatsukasa is with the University of Oxford, UK, and the National Institute of Informatics, Japan. (e-mail: nakatsukasa@maths.ox.ac.uk)}
}

\maketitle

\begin{abstract}
Graph sampling set selection, where a subset of nodes are chosen to collect samples to reconstruct a smooth graph signal, is a fundamental problem in graph signal processing (GSP).
Previous works employ an unbiased least-squares (LS) signal reconstruction scheme and select samples via expensive extreme eigenvector computation.
Instead, we assume a biased graph Laplacian regularization (GLR) based scheme that solves a system of linear equations for reconstruction.
We then choose samples to minimize the condition number of the coefficient matrix---specifically, maximize the smallest eigenvalue $\lambda_{\min}$.
Circumventing explicit eigenvalue computation, we maximize instead the lower bound of $\lambda_{\min}$, designated by the smallest left-end of all Gershgorin discs of the matrix.
To achieve this efficiently, we first convert the optimization to a dual problem, where we minimize the number of samples needed to align all Gershgorin disc left-ends at a chosen lower-bound target $T$.
Algebraically, the dual problem amounts to optimizing two disc operations: i) shifting of disc centers due to sampling, and ii) scaling of disc radii due to a similarity transformation of the matrix.
We further reinterpret the dual as an intuitive disc coverage problem bearing strong resemblance to the famous NP-hard set cover (SC) problem.
The reinterpretation enables us to derive a fast approximation scheme from a known SC error-bounded approximation algorithm.
We find an appropriate target $T$ efficiently via binary search.
Extensive simulation experiments show that our disc-based sampling algorithm runs substantially faster than existing sampling schemes and outperforms other eigen-decomposition-free sampling schemes in reconstruction error.
\end{abstract}


%
\IEEEpeerreviewmaketitle

\section{Introduction}
\label{sec:intro}
Sampling of smooth signals---and subsequent signal reconstruction from collected samples---are fundamental problems in signal processing.
Graph sampling extends the well-known studies of Nyquist sampling on regular data kernels to irregular data kernels described by graphs.
Graph sampling can be studied under different contexts: aggregation sampling \cite{aggressive2016TSP,aggressive2018TSIPN}, local measurement \cite{wang2016local} and subset sampling \cite{chen2015sampling}.
In this paper, we focus on \textit{graph sampling set selection}: a subset of nodes are chosen to collect samples to reconstruct a graph signal with a well-defined notion of smoothness.

For the noiseless case, graph sampling set selection is well investigated \cite{sampling2008TAMS,chen2015signal,cornell}; in particular, \cite{cornell} shows that random selection of $K$ samples on a graph for a target $K$-bandlimited signal can enable perfect reconstruction with high probability.
The noisy case where acquired samples are corrupted by additive noise is more challenging, and there exists many proposals.
Assuming that the target signal is bandlimited, a class of schemes compute one or more extreme eigenvectors of a graph Laplacian matrix (or sub-matrix) and examine the per-node magnitudes to greedily determine a sample set.
While these methods show good sampling performance using a least-squares (LS) signal reconstruction scheme, their complexity is high.
On the other hand, there are faster sampling methods based on the notion of local operators \cite{akie2017accelerated,akie2018eigenFREE}, but these methods are ad-hoc in nature, and there are no global notions of optimality associated with the selected sampling sets.

In practice, real-world graph signals, such as temperature in sensor network \cite{TSIPN2016Chen}, category labels in graph-based classifier \cite{active_sampling2014KDD} and ratings in recommendation system \cite{TSP2019Ortiz}, are typically not strictly bandlimited with known bandwidths.
Instead, generally ``smooth" signals with respect to known graph structures are relatively more common.
Different from standard least-squares reconstruction only considering spectrum energy under a given cutoff frequency, we assume a different signal smoothness notion defined using \textit{graph Laplacian regularization} (GLR), and then propose a corresponding sampling scheme with roughly linear complexity.
Further, our scheme can be proven to minimize the global worst case mean squared error (MSE);
thus our sampling scheme can gracefully scale to very large graphs and has excellent signal reconstruction performance.
Specifically, we first assume a biased GLR  based scheme that solves a system of linear equations for reconstruction given collected samples \cite{GDenoising2017Pang}.
We then choose samples to minimize the condition number of the coefficient matrix; \textit{i.e.}, maximize the smallest eigenvalue $\lambda_{\min}$.
Circumventing computation of extreme eigenvalues, we maximize instead the lower bound of $\lambda_{\min}$, designated by the smallest left-end of all Gershgorin discs of the matrix \cite{varga04}.

To achieve this efficiently, we first convert the optimization to a dual problem, where we minimize the number of samples needed to align all Gershgorin disc left-ends at a pre-selected lower-bound target $T$.
Algebraically, the dual problem amounts to two disc operations: i) shifting of disc centers due to sampling, and ii) scaling of disc radii due to a similarity transformation of the matrix.
We further reinterpret the dual as an intuitive \textit{disc coverage} (DC) problem, which bears strong resemblance to the famous \textit{set cover} (SC) problem \cite{cormen2009introduction}.
Though DC remains NP-hard (we provide a proof by reduction), the reinterpretation enables us to derive a fast approximation scheme from a known SC error-bounded approximation algorithm \cite{cormen2009introduction}.
An appropriate target $T$ can be located quickly via simple binary search.
Extensive simulation experiments show that our disc-based sampling algorithm runs substantially faster than existing eigen-decomposition-based schemes (hundreds to thousands times faster) for medium-sized graphs and outperforms other eigen-decomposition-free schemes in reconstruction error for large-sized graphs.

The outline of the paper is as follows.
We first overview related works in Section\;\ref{sec:related_work}.
We introduce notations and a biased reconstruction method in Section\;\ref{sec:preliminary}.
In Section\;\ref{sec:formulation}, we first propose our sampling formulation to minimize the upper bound of MSE and then derive its dual problem based on Gershgorin circle theorem.
We develop a fast sampling set selection algorithm to tackle our sampling formulation with roughly linear complexity in Section\;\ref{sec:sampling}.
Experimental results are presented in Section\;\ref{sec:experiments}.
We conclude in Section\;\ref{sec:conclusion}.

\section{Related Works}
\label{sec:related_work}
The problem of graph sampling can be defined from different contexts.
From an \textit{aggregation sampling} context, the authors in \cite{aggressive2016TSP,aggressive2018TSIPN} proposed to aggregate the sampled signal at one node (or a subset) but with the graph signal shifted successively by a known shift operator.
From a \emph{local weighted sampling} context, \cite{wang2016local} proposed to collect samples via a weighted sum of signals from different local sets.
In contrast, we focus on graph sampling in the context of set selection in this paper: how to select a subset of nodes to collect samples to reconstruct a smooth signal in high fidelity.
We next focus on detailed reviews of  related works under the umbrella of set selection-based graph sampling.

\subsection{Deterministic Graph Sampling Set Selection}
Sampling theory for noise-free graph signals was first studied in \cite{sampling2008TAMS}, where a \textit{uniqueness set} is defined for perfect recovery of \emph{bandlimited} graph signals.
Starting from this definition, \cite{sp_proxy2016,anis2018sampling,gadde2015probabilistic} proved a sufficient and necessary condition for a uniqueness set, and proposed a lightweight sampling set selection strategy based on the notion of graph spectral proxies.

In the presence of noise, assuming an unbiased \textit{least-squares} (LS) signal reconstruction with independent and identically distributed (i.i.d) noise model, the minimum mean square error (MMSE) function for sampling is actually the known A-optimal criterion \cite{boyd2004convex,MIA2018Fen,experiments}. In response, the authors in \cite{MFN2016TSP,2018greedy} proposed to directly optimize the A-optimal function with greedy scheme. Alternatively, \cite{e_optimal2015} used the E-optimality criterion for graph sampling, which can be interpreted as minimizing the worst case of MSE.
However, the criteria in \cite{MFN2016TSP,2018greedy,e_optimal2015} all required extreme eigen-pair computation of the graph variation operator to choose samples, which is not scalable for large graphs.
Though \cite{MIA2018Fen} avoided eigen-pair computation by developing a sampling method based on truncated Neumann series,  it must compute many matrix-series  multiplications for good approximation, which is expensive.

To tackle the complexity issue in large scale graphs, very recently, \cite{akie2017accelerated,akie2018eigenFREE} proposed an eigen-decomposition-free graph sampling method by successively maximizing the coverage of the localization operators, which was implemented via Chebyshev polynomial approximation \cite{gw2011Hammond}.
Nonetheless, this sampling idea utilizes local message heuristically and has no global performance guarantee.
Orthogonally, our early work  \cite{yuanchao2019ICASSP} proposed a fast graph sampling algorithm via Gershgorin disc alignment without any eigen-pair computation, but its performance was sub-optimal due to the node sampling strategy simply based on breadth first search (BFS).
It is known that random sampling \cite{r_sampling2018ACHA,puy2018structured} can lead to very low computational complexity, but it typically requires more samples for the same signal reconstruction quality compared to its deterministic counterparts.

In this paper, we propose a roughly linear-time  deterministic graph sampling method without any eigen-pair computation, while its reconstruction performance remains superior.
Unlike previous fast sampling works \cite{sp_proxy2016,akie2018eigenFREE,yuanchao2019ICASSP}, our method can be proven to optimize the robustness of the linear reconstruction system resulting from GLR and reduce the upper-bound of the global reconstruction MSE, while each sampling step is performed locally for speed.

\subsection{Sampling of GMRF Graph Signals}

Besides the bandlimitedness assumption, graph signal (label information) in classical semi-supervised learning is interpreted as a Gaussian Markov random field (GMRF), where the signal is smooth with respect to the graph Laplacian operator \cite{microsoft}.
With this property, active semi-supervised learning  can be interpreted as graph sampling on a nearest neighbor graph constructed by node features \cite{active_sampling2014KDD}.
An adaptive graph sampling method was proposed \cite{dasarathy2015s2} to collect samples based on both successively modified graph and selected samples.
To improve efficiency, authors in \cite{gad2016active} proposed a hybrid adaptive and non-adaptive sampling method to combine the graph sampling and active learning methods.
Under the GMRF model, \cite{pinyu} showed that the MSE function for graph sampling is the trace of the posterior covariance matrix, based on which they developed a greedy sampling method. However, this method is very inefficient in large graphs because of the matrix inverse computation in each greedy step.
In this paper, we derive a fast sampling strategy for  GMRF graph signals with smoothness prior, but it is also applicable for  bandlimited graph signals.

\section{Preliminary}
\label{sec:preliminary}
\subsection{Background and Notations}
A graph $\cG$ is defined by a triplet $\cG(\cV,\cE,\W)$, where $\cV$ and $\cE$ represent sets of $N$ nodes and $M$ edges in the graph, respectively.
Associated with each edge $(i, j)\in\cE$ is a weight $w_{ij}$, which reflects the correlation or similarity between two connected nodes $i$ and $j$.
In this paper, we assume a connected, undirected graph;  \textit{i.e.}, $w_{ij}=w_{ji}, \forall i, j \in \cV$.
$\W$ is an \textit{adjacency} matrix with $w_{ij}$ as the $(i,j)$-th entry of the matrix.
We assume a non-negative weight matrix with elements $w_{ij}>0$ if $\forall (i, j)\in\cE$, and $w_{ij}=0$ otherwise.

Given $\W$, the \textit{combinatorial graph Laplacian} matrix $\L$ is computed as \cite{shuman13,ortega2018ieee}:
\begin{equation}
    \L\triangleq \D-\W
    \label{eq:graph_laplacian}
\end{equation}
where $\D=\diag(\W\mathbf{1})$ is a diagonal \textit{degree} matrix. $\mathbf{1}$ is a vector of all 1's and $\diag(\cdot)$ is an operator that returns a diagonal matrix with the elements of an input vector on the main diagonal.
As $\L$ is real and symmetric, it can be eigen-decomposed via the \textit{Spectral Theorem} \cite{horn1990matrix},
\begin{equation}
    \L=\U\bLambda\U^\top
    \label{eq:graph_spectrum}
\end{equation}
where $\bLambda$ is a diagonal matrix containing real eigenvalues $\theta_k$ along its diagonal, and $\U$ is an orthonormal matrix composed of real eigenvectors $\u_k$'s as columns.
Since $w_{ij}\ge0$ for $\forall (i,j)\in \cE$, $\x^\top \L\x=\sum_i\sum_j w_{ij}(x_i-x_j)^2\ge0$ for $\forall \x\in\sR^N$.
Thus, $\L$ is a \textit{positive semi-definite} (PSD) matrix.
The non-negative eigenvalues $\theta_k$'s are interpreted as \textit{graph frequencies}.
By sorting $\theta_k$'s in a non-decreasing order, \emph{i.e.}, $0=\theta_1\le\theta_2\le \ldots \le\theta_N$, $\theta_k$'s are ordered from lowest to highest frequencies.
$\theta_1$ represents the lowest DC graph frequency and $\theta_N$ represents the highest AC graph frequency.
The corresponding eigenvectors in $\U$ are interpreted as Fourier-like graph frequency basis; together they form the \textit{graph spectrum} for a graph $\cG$.

Given a graph frequency basis $\U$, a graph signal $\x$ can be transformed to its graph frequency domain representation $\balpha$ via \textit{graph Fourier transform} (GFT) \cite{shuman13,ortega2018ieee}:
\begin{equation}
    \balpha=\U^\top\x=[\u_1, ..., \u_N]^\top\x
    \label{eq:gft}
\end{equation}
A graph signal $\x$ is called \textit{$\omega$-bandlimited} if $\alpha_i=0$, for $\forall i \in \{i ~|~ \theta_i>\omega \}$.

A graph signal $\x\in\sR^N$ is defined on the nodes of a graph,
\ie, one scalar value $x_i$ is assigned to the $i$-th node.
Given a graph signal $\x$ and subset  $\cS\subset\{1, ..., N\}$ with cardinality $|\cS|=K<N$, we define the sampling process as \cite{sampling2008TAMS,sampling2014icassp,e_optimal2015,sp_proxy2016,r_sampling2018ACHA,MIA2018Fen,MFN2016TSP,2018greedy,akie2018eigenFREE,cornell}:
\begin{equation}
    \y=\x_\cS+\n=\H\x+\n
    \label{eq:sampling}
\end{equation}
where $\y\in \sR^K$ is a sampled observation of length $K$.
$\x_\cS\in \sR^K$ denotes a sub-vector of $\x$ consisting of elements indexed by $\cS$ and
$\n$ is an additive noise term.
The \textit{sampling matrix} $\H\in\{0, 1 \}^{K\times N}$ is defined as \cite{MIA2018Fen}:
\begin{equation}
    \H_{ij}=\begin{cases}
    1,~~~~~j=\cS_i;\\
    0,~~~~~\mbox{otherwise}.
    \end{cases}
    \label{eq:sampling_matrix}
\end{equation}
where $\cS_i$ denotes the $i$-th element of set $\cS$.
The complement set of $\cS$ on $\cV$ is denoted by $\cS^c=\cV\setminus\cS$.
We define $L_2(\cS)$ as the space of all graph signals that are zero except on the subset of nodes indexed by $\cS$ \cite{sp_proxy2016}:
\begin{equation}\label{LSet}
   L_2(\cS)=\{\x\in\sR^N|~\x_{\cS^c}=\mathbf{0}\}
\end{equation}

\subsection{Signal Reconstruction with Graph Laplacian Regularizer}
Instead of assuming an unbiased \textit{least-squares} (LS) signal reconstruction scheme from sparse samples \cite{e_optimal2015,MFN2016TSP,2018greedy,MIA2018Fen}, we adopt a popular biased graph signal reconstruction scheme based on \textit{graph Laplacian regularization} (GLR) \cite{GDenoising2017Pang}, which has demonstrated its effectiveness in numerous practical applications, such as image processing \cite{GDenoising2017Pang,GS2017Xian,gene2018ieee,deblur2017YBai,liu2018graph,dinesh20183d} , computer graphics \cite{Laplacian_mesh2005eurographics} and semi-supervised learning \cite{gl_semi_learning2004lt,su2018}.
Biased signal reconstruction using GLR leads to an unconstrained $l_2$-$l_2$-norm minimization, whose solution can be obtained by solving a system of linear equations, efficiently computed using state-of-the-art numerical linear algebra algorithms such as \textit{conjugate gradient} (CG) \cite{hestenes1952methods}.
For large noise variance, it is shown that the biased estimation has smaller reconstructed errors on average than the unbiased counterpart \cite{MIA2018Fen,chen2017spl}.

Specifically, given a noise-corrupted sampled observation $\y\in\sR^K$ on a graph $\cG$, one can formulate an optimization for the target signal $\hat{\x} \in \mathbb{R}^N$ using GLR as follows:
\begin{equation}
    \hx=\argmin_\x \|\H\x-\y\|^2_2 +
    \mu \; \x^{\top} \L \, \x
    \label{eq:reconstruction}
\end{equation}
where $\H$ is a sampling matrix defined in (\ref{eq:sampling_matrix}).
$\mu > 0$ is a tradeoff parameter to balance GLR against the $l_2$-norm data fidelity term.
The law of selecting an optimal $\mu$ with respect to both noise variance and graph spectrum is studied in \cite{chen2017spl}.
The focus of this paper is on graph sampling rather than graph signal reconstruction, and interested readers can refer to \cite{chen2017spl} for a discussion on the appropriate selection of $\mu$.

Since objective (\ref{eq:reconstruction}) is quadratic, the optimal solution $\hx$ can be obtained by solving a system of linear equations:
\begin{equation}
    (\H^{\top}\H+\mu\L)\hx=\H^{\top} \y.
    \label{eq:linear_equation}
\end{equation}

For $\L$ constructed on a connected graph $\cG$, DC graph frequency $\u_1=\mathbf{1}/{\sqrt{N}}$ leads to $\u^\top_1\L\u_1=0$ and $\|\H\u_1\|_2>0$ for $\forall\cS$ with $K\ge1$; for any nonzero graph signal $\x\in \sR^N \ne \u_1$, $\x^\top\L\x>0$ and $\|\H\x\|_2\ge 0$.
Therefore, $\forall \x\ne \mathbf{0}$, $\x^\top(\H^\top\H+\mu\L)\x=\|\H\x\|^2_2+\mu\x^\top\L\x>0$.
Hence, coefficient matrix $\H^{\top}\H+\mu\L$ must be \textit{positive definite} (PD), and (\ref{eq:linear_equation}) has a unique solution $\hx=(\H^\top\H+\mu\L)^{-1}\H^\top\y$.

\section{Problem Formulation}
\label{sec:formulation}
\subsection{Reconstruction Stability and Sampling Set}
Although $\H^{\top}\H+\mu\L$ is a PD matrix, (\ref{eq:linear_equation}) can potentially be poorly conditioned, because
the minimum eigenvalue $\lambda_{\min}$ of $\H^{\top}\H+\mu\L$ can be very close to zero.
For instance, if $\H$ is chosen such that there exists a low frequency graph signal $\x\in L_2(\cS^c)$ and its bandlimit $\omega$ is very small, then
\begin{align}
    \lambda_{\min}&\le\frac{\x^\top(\H^\top\H+\mu\L)\x}{\|\x\|^2_2}
    \mathop =\limits^{(1)}
    \mu\frac{\x^\top\L\x}{\|\x\|^2_2} \notag \\
    &=\mu\frac{\alpha_1^2\theta_1+\ldots+\alpha_k^2\theta_k}{\|\balpha\|^2_2}\le \mu \, \omega. 
\end{align}
where $\mathop =\limits^{(1)}$ holds because of the definition in \eqref{LSet}.
For small $\mu$, $\lambda_{\min}$ is then upper-bounded by a small number.

This means that condition number $\rho = \lambda_{\max} / \lambda_{\min}$ for coefficient matrix $\H^{\top} \H + \mu \L$
can be very large, resulting in instability when $\hat{\x}$ is numerically computed in \eqref{eq:linear_equation}, and more importantly, $\hat{\x}$ is sensitive to the noise in the observation $\y$.
Our goal then is to maximize $\lambda_{\min}$ by appropriately choosing sampling matrix $\H$, so that $\rho$ is minimized\footnote{Largest eigenvalue $\lambda_{\max}$ is upper-bounded by $1+2 \mu d_{\max} \ll \infty$, where $d_{\max}$ is the maximum node degree of the graph $\mathcal{G}$. Thus to minimize condition number $\rho = \lambda_{\max}/\lambda_{\min}$, it is sufficient to maximize $\lambda_{\min}$.}.

Denote a diagonal matrix by  $\A=\H^{\top}\H=\diag(\a)=\diag(a_1, ..., a_N)\in\sR^{N\times N}$, where \textit{sampling vector} $\a$ satisfies
\begin{equation}
    a_{i}=\begin{cases}
    1,~~~~~i\in \cS,\\
    0,~~~~~\mbox{otherwise}.
    \label{eq:ata}
    \end{cases}
\end{equation}
In other words, $\a$ is $0$-$1$ bit array encoded sampling set $\cS$.
Denote by $\B=\A+\mu\L$.

Given matrix $\L$, parameter $\mu$ and a sampling budget $K$ where $0< K< N$, our objective is to find a sampling set $\cS$ with $|\cS| \leq K$ (or, equivalently, a sampling vector $\a$ with $\sum_{i=1}^N a_i \leq K$), to maximize $\lambda_{\min} $ of $\B$, \ie,
\begin{align}
    \max_\a&~~\lambda_{\min}(\B) \label{eq:max_lambda_min} \\
    \mbox{s.t.}~&\B=\A+\mu\L, \notag\\
    &\A=\diag(\a), ~~ a_{i}\in\{0,1\}, ~~ \sum_{i=1}^N a_i \leq K. \notag
\end{align}

We next show that maximizing $\lambda_{\min}(\B)$ in (\ref{eq:max_lambda_min}) is equivalent to minimizing the worst case of MMSE.

\begin{proposition}
    Maximizing $\lambda_{\min}(\B)$ in (\ref{eq:max_lambda_min}) minimizes the upper bound of MSE between original signal $\x$ and reconstructed signal $\hx$.
\end{proposition}
\begin{proof}
    Reconstructed signal $\hx$ from a noisy sampled signal $\y$ using (\ref{eq:linear_equation}) is
    \begin{align}
        \hx&=(\H^\top\H+\mu\L)^{-1}\H^\top\y \notag\\
        &=(\H^\top\H+\mu\L)^{-1}\H^\top(\H\x+\n) \\
        &=(\H^\top\H+\mu\L)^{-1}\H^\top\H(\x+\tn)\notag
    \end{align}
    where $\tn\in\sR^N$ is an $N$-dimensional additive noise.
    Denote by $\B=\A+\mu\L=\H^\top\H+\mu\L$. We have
    \begin{equation}
        \hx=\B^{-1}(\B-\mu\L)(\x+\tn)
           =(\I-\mu\B^{-1}\L)(\x+\tn) \label{eq:mse_proof2}
    \end{equation}

    To minimize MSE between original signal $\x$ and reconstructed signal $\hx$ is equivalent to minimizing $\|\hx-\x\|_2$.
    \begin{align}
       \notag \|\hx-\x\|_2&=\|(\I-\mu\B^{-1}\L)(\x+\tn)-\x\|_2 \\\notag
        &=\|-\mu\B^{-1}\L(\x+\tn)+\tn\|_2 \\
        &\le\|-\mu\B^{-1}\L(\x+\tn)\|_2+\|\tn\|_2 \\\notag
        &\le\mu\|\B^{-1}\|_2\|\L(\x+\tn)\|_2+\|\tn\|_2
    \end{align}

    Since $\B$ is a symmetric PD matrix, $\|\B^{-1}\|_2=\frac{1}{\lambda_{\min}(\B)}$.
    Thus, maximizing $\lambda_{\min}(\B)$ minimizes the upper bound of MSE, given fixed $\mu$, $\L$ and $\x$.
\end{proof}

\subsection{Maximizing the lower bound of $\lambda_{\min}(\B)$}
\label{subsec:max_lower_bound}
Solving (\ref{eq:max_lambda_min}) directly is a challenging combinatorial optimization problem.
To begin with, for large graphs, computing $\lambda_{\min}(\B)$ given $\a$ using a state-of-the-art numerical linear algebra algorithm like the Lanczos algorithm \cite{parlett1998symmetric} or the LOBPCG algorithm \cite{lobpcg} is still computation-expensive.
Moreover, we need to compute $\lambda_{\min}(\B)$ for exponential number of candidate sampling vectors $\a$ where $\sum_{i=1}^N a_i = K$.

In order to study $\lambda_{\min}(\B)$ \textit{without} performing expensive eigen-pair computation, we design an efficient algorithm based on \textit{Gershgorin Circle Theorem} (GCT) \cite{horn1990matrix}.
We first describe GCT below:
\begin{theorem}[Gershgorin Circle Theorem]
    Let $\X$ be a complex $n\times n$ matrix with entries $x_{ij}$. For $i\in\{1, ..., n\}$, let $R_i=\sum_{j\ne i}|x_{ij}|$ be the sum of the absolute values of the off-diagonal entries in the $i$-th row. Consider the $n$ Gershgorin discs
    \begin{equation}
        \Psi_i(x_{ii}, R_i)=\{z\in\sC~|~|z-x_{ii}|\le R_i\},
        ~~ i \in \{1, \ldots, n\}.
        \label{eq:gct}
    \end{equation}
    where $x_{ii}$ and  $R_i$ are the center and radius of disc $i$, respectively.
    Each eigenvalue of $\X$ lies within at least one Gershgorin disc $\Psi_i(x_{ii}, R_i)$.
    \label{th:gct}
\end{theorem}

For our specific real-valued coefficient matrix $\B$, each real eigenvalue $\lambda$ lies within at least one \textit{Gershgorin disc} $\Psi_i(b_{ii},R_i)$ with disc center $b_{ii}$ and radius $R_i$, \textit{i.e.},
\begin{equation}
    b_{ii} - R_i \le \lambda\le b_{ii} + R_i,
    \label{eq:disk}
\end{equation}
where $R_i=\sum _{j\ne i} |b_{ij}|=\mu\sum _j w_{ij}=\mu d_{i}$, and $d_i$ is the degree of node $i$.
The third equality is true since there are no self-loops in $\cG$.
Center of disc $i$ is $b_{ii} = \mu d_i+a_{i}$.

Using GCT, we can compute the lower bound of $\lambda_{\min}(\B)$ by examining left-ends of Gershgorin discs.
Specifically, the lower bound of $\lambda_{\min}(\B)$ is computed as the smallest left-end of all Gershgorin discs, \ie,
\begin{equation}
    \min_i \{b_{ii} - R_i\} = \min_i a_{i}
    \leq \lambda_{\min}(\B).
    \label{eq:lower_bound1}
\end{equation}

For intuition, we examine the two extreme sampling cases.
One extreme case is when no nodes are sampled, \ie, $\cS=\varnothing$, and both the lower bound and $\lambda_{\min}$ of $\B =\mu\L$ are $0$.
The other extreme case is when all nodes are sampled, \ie, $\cS=\cV$, and both the lower bound and $\lambda_{\min}$ of $\B=\I+\mu\L$ are $1$.
Note that in both extreme cases, the smallest Gershgorin disc left-end (lower bound) is also exactly $\lambda_{\min}$.
The observation inspires us to maximize the lower bound of $\lambda_{\min}(\B)$ instead of maximizing $\lambda_{\min}(\B)$ directly.

In general, the tightness of the Gershgorin lower bound of $\lambda_{\min}$ does not hold for any sampling sets. In order to tighten and maximize the lower bound of $\lambda_{\min}(\B)$, we introduce two basic operations to manipulate Gershgorin discs:

\vspace{0.1in}
\noindent
\textbf{Disc Shifting}:
The first operation is \textit{disc shifting} via sampling.
In (\ref{eq:disk}), the left-end of the $i$-th Gershgorin disc $\Psi_i$ in matrix $\B$ is $b_{ii} - R_i = a_{i}$.
When node $i$ is sampled, we have
\begin{equation}
    a_i: 0\rightarrow 1,~~~i\in\cS
\end{equation}
The corresponding left-end of disc $\Psi_i$ shifts from $0$ to $1$.

\vspace{0.1in}
\noindent
\textbf{Disc Scaling}:
The second operation is \textit{disc scaling} via similarity transformation.
For $0<K<N$, there is at least one $a_i = 0$, $i \in \mathcal{S}^c$, and the lower bound of $\lambda_{\min}(\B)$ in \eqref{eq:lower_bound1} is always $0$.
We introduce disc scaling to pull left-ends of unsampled discs away from $0$, resulting in a tighter lower bound for $\lambda_{\min}(\B)$.
Specifically, disc scaling is defined as
\begin{equation}
    \C=\bS\B\bS^{-1},
    \label{eq:similar_trans}
\end{equation}
where $\bS$ is a chosen invertible square matrix, named \textit{scaling matrix}. Matrix $\C$ is the similarity transformation of the square matrix $\B$ and shares the same eigenvalues with $\B$ (but with different eigenvectors) \cite{varga04}.
Thus, we can examine Gershgorin disc left-ends of transformed matrix $\C$ instead of original $\B$.

\begin{figure}[!t]
\centering
\vspace{20pt}
\includegraphics[width=0.9\linewidth]{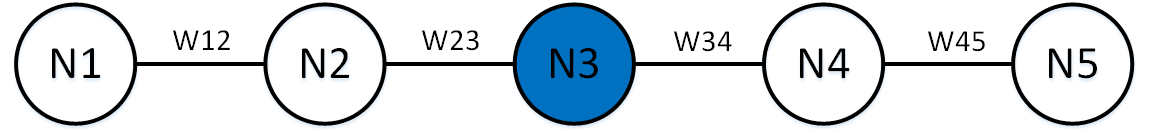}
\caption{An illustrative example of a $5$-node path  graph.}
\label{fig:nodes_example}
\end{figure}

To simplify the design of $\bS$, we employ a PD diagonal matrix $\bS=\diag(\s)$, $s_i > 0$, $\forall i \in \{1, \ldots, N\}$, to scale Gershgorin discs of $\B$, so that left-ends of Gershgorin discs of resulting transformed $\C$ can be moved right.
Specifically, the $i$-th diagonal term $s_i$ of $\bS$ in (\ref{eq:similar_trans}) and its corresponding element $s_i^{-1}$ in $\bS^{-1}$ are used to scale the radius $R_i$ of $\Psi_i$ and the radii of its neighbors' discs $\Psi_j$ respectively, where $j \in \mathcal{N}_i = \{j ~|~ w_{ij} > 0\}$.
For example, if a scalar $s_i>1$ is used to \textit{expand} $R_i$, then its neighbors' discs are \textit{shrunk} with $s_i^{-1}<1$.
Since $s_i$ is always offset by $s_i^{-1}$ on the main diagonal, the center $b_{ii}$ of disc $\Psi_i$ is unchanged.

\vspace{0.1in}
\noindent \textit{5-node Path Graph Example:} To illustrate the two basic operations (shifting and scaling) on Gershgorin discs, as an example we consider a $5$-node path graph and start by sampling node $3$, as shown in Fig.~\ref{fig:nodes_example}. Assuming $\mu=1$, the graph's coefficient matrix $\B$ with $(3,3)$-th entry updated is shown in Fig.\;\ref{fig:adj_matrix_s0}.
The left-end of disc $\Psi_3$ is shifted from $0$ to $1$, as shown in Fig.\;\ref{fig:disk_e1}.
After disc $\Psi_3$ is shifted, we apply a scalar $s_3>1$ on the third row of $\B$ to expand radius and center of disc $\Psi_3$ in Fig.~\ref{fig:adj_matrix_s1} and \ref{fig:disk_e2}.
Afterwords, its neighbors' discs $\Psi_2$ and $\Psi_4$ are shrunk when applying $s_3^{-1}<1$ on the third column of $\B$, as shown in Fig.~\ref{fig:adj_matrix_s2} and \ref{fig:disk_e3}. Note that $s_3$ is offset by $s_3^{-1}$ on the main diagonal of  disc $\Psi_3$, so its  center $b_{33}=1+d_3$ is unchanged.

\vspace{0.1in}
With the above two operations, we can select a sampling set and \textit{shift} left-ends of the corresponding Gershgorin discs from $0$ to $1$.
At the same time, we \textit{scale} radii of all Gershgorin discs to align their left-ends, resulting in a tighter lower bound of $\lambda_{\min}(\B)$.
Given a graph $\cG$ and a sampling budget $K$, maximizing the lower bound using the two disc operations (encoded in vectors $\a$ and $\s$) can be formulated as:
\begin{align}
    \max_{\a,\s}& ~ \min_{i \in \{1, \ldots, N\}} ~~
    c_{ii}-\sum_{j\ne i}|c_{ij}|
    \label{eq:max_lower_bound} \\
    \mbox{s.t.}~&\C=\bS \left( \A+\mu\L \right) \bS^{-1} \notag\\
    &\A=\diag(\a), ~~~ a_{i}\in\{0,1\}, ~~~ \sum_{i=1}^N a_i \leq K, \notag \\
    &\bS=\diag(\s), ~~~ s_i > 0. \notag
\end{align}
where the objective is to maximize the smallest left-end $\min_{i \in \{1, \ldots, N\}} ~ c_{ii}-\sum_{j\ne i}|c_{ij}|$ of all Gershgorin discs of similarity-transformed matrix $\C$ of original $\A + \mu \L$.
We call this formulated optimization  (\ref{eq:max_lower_bound}) the \textit{primal problem}.

\begin{figure}[!t]
\centering
\subfloat[]{
\label{fig:adj_matrix_s0}
\includegraphics[width=0.31\linewidth]{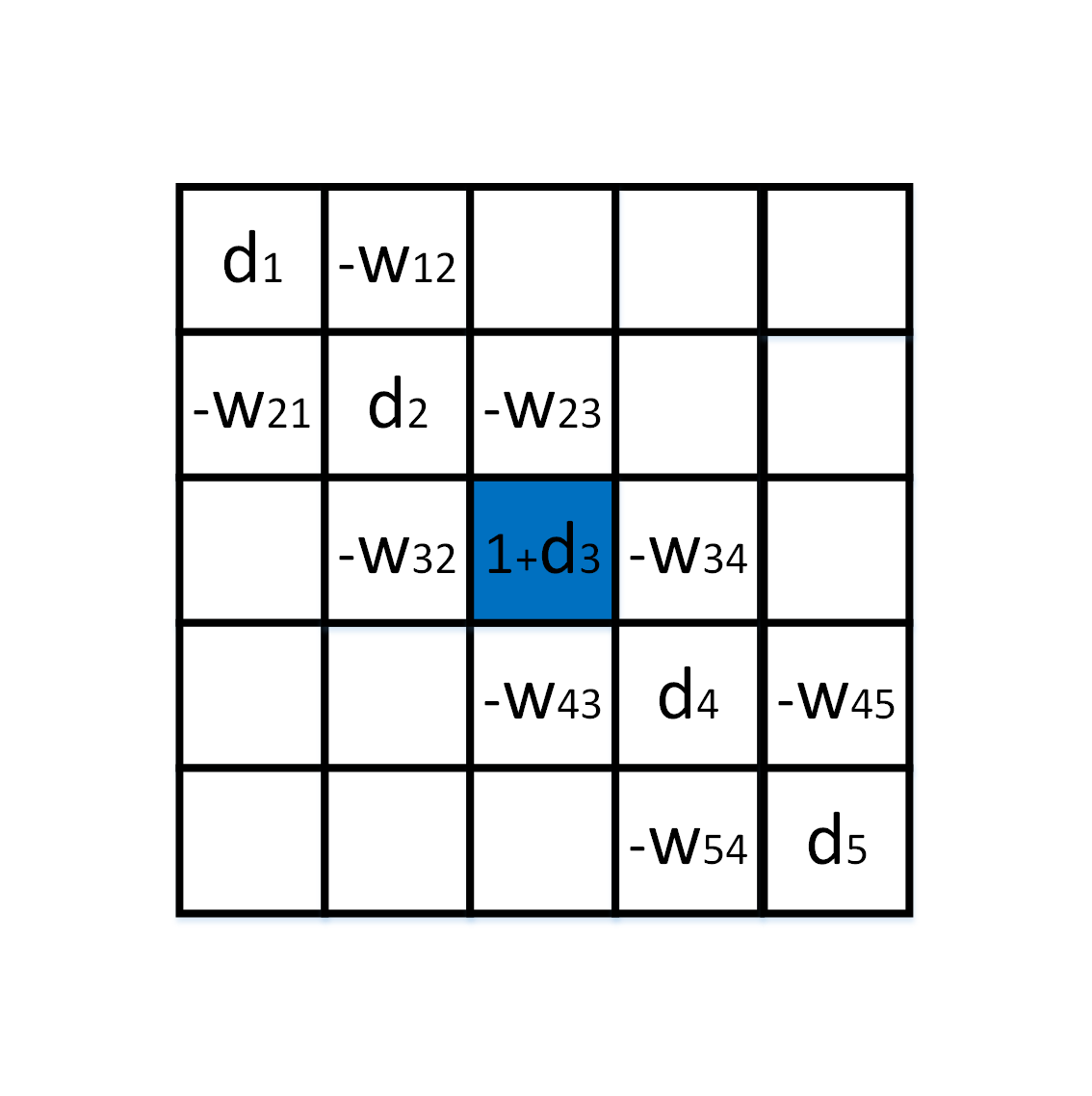}}
\subfloat[]{
\label{fig:adj_matrix_s1}
\includegraphics[width=0.32\linewidth]{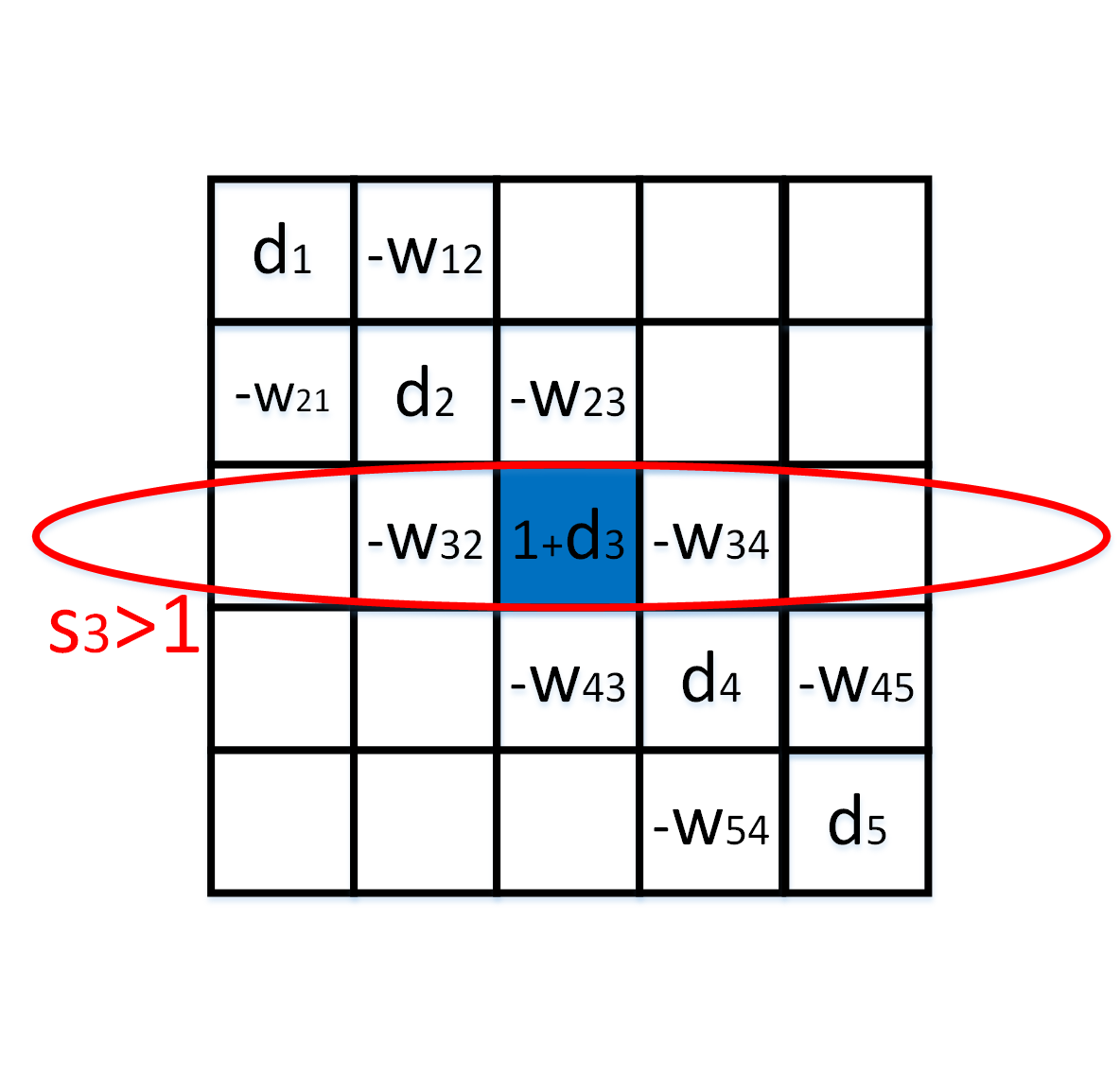}}
\subfloat[]{
\label{fig:adj_matrix_s2}
\includegraphics[width=0.32\linewidth]{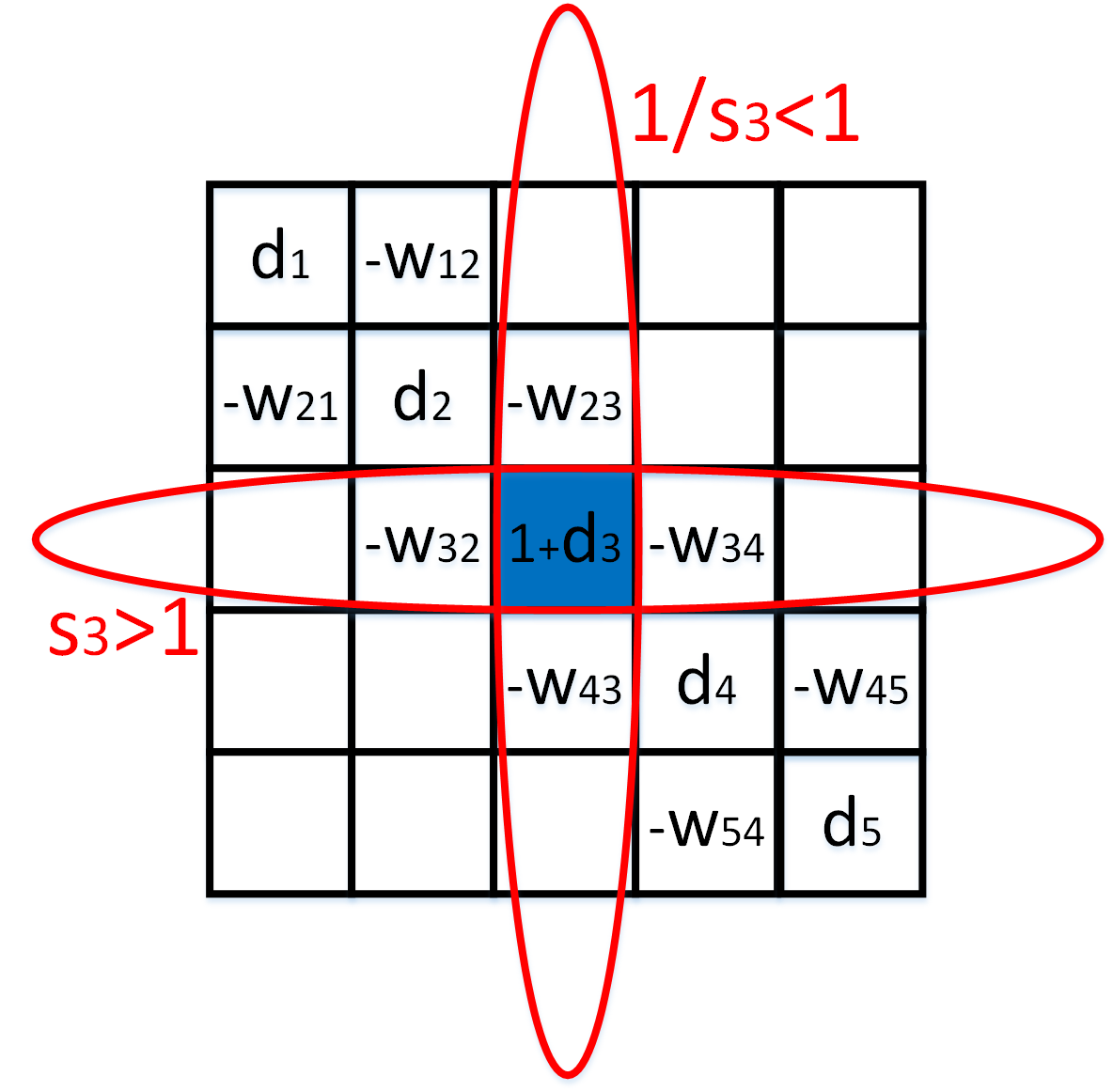}} \\
\subfloat[]{
\label{fig:disk_e1}
\includegraphics[width=0.32\linewidth]{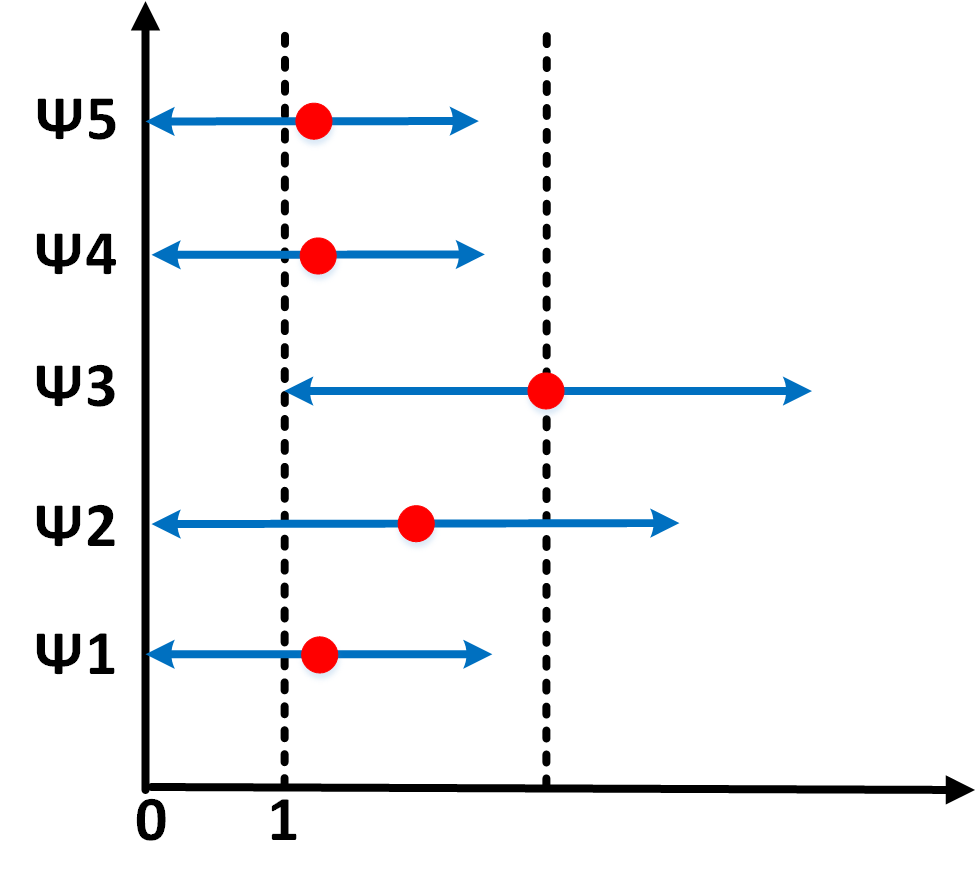}}
\subfloat[]{
\label{fig:disk_e2}
\includegraphics[width=0.32\linewidth]{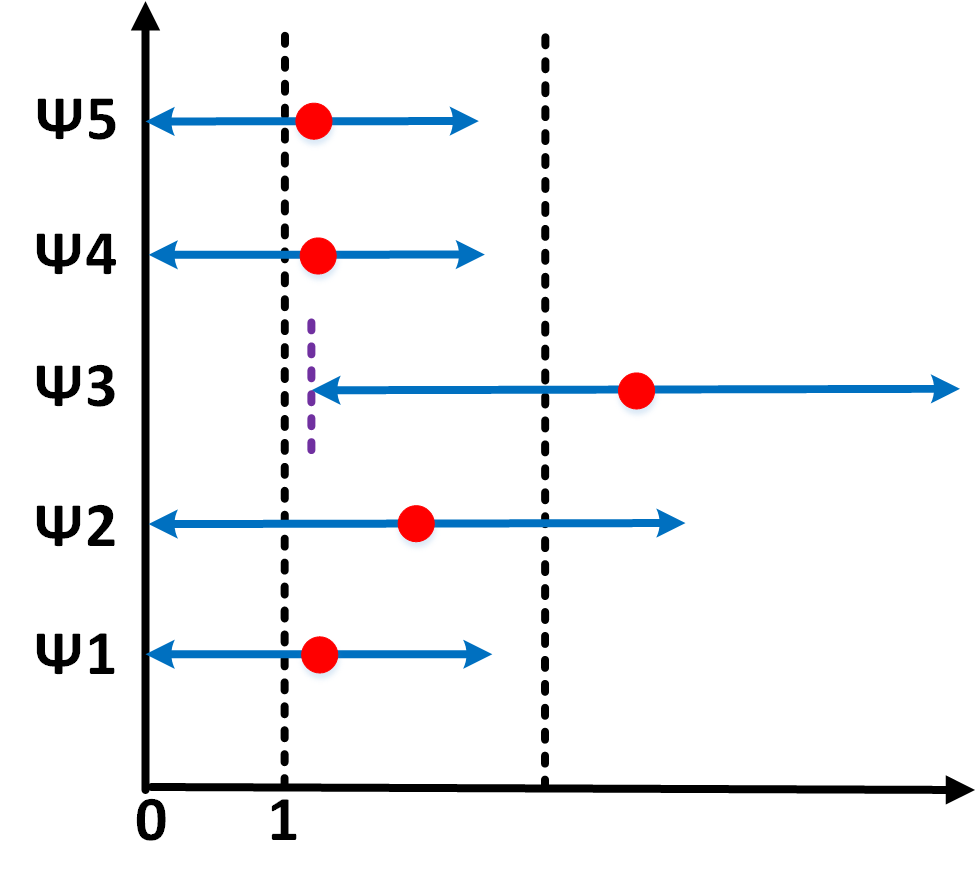}}
\subfloat[]{
\label{fig:disk_e3}
\includegraphics[width=0.32\linewidth]{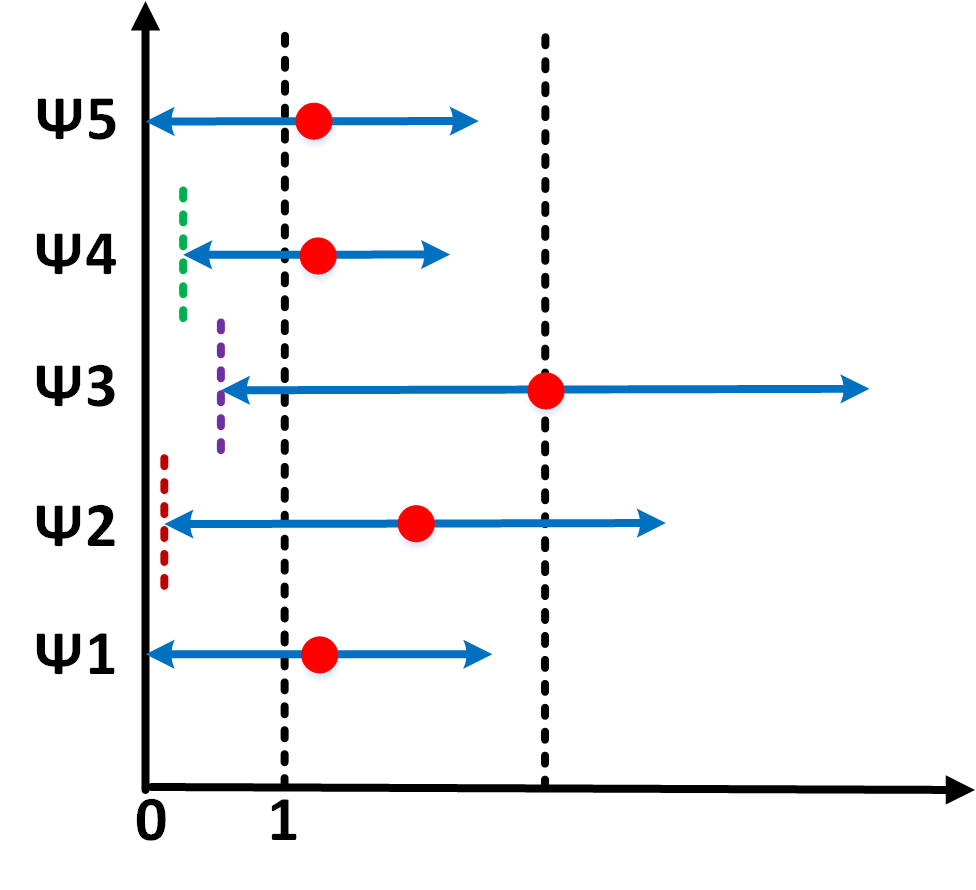}}
\caption{An illustration of disc operations. (a) $\&$ (d) Shifting disc $\Psi_3$. (b) $\&$ (e) Shifting and expanding disc $\Psi_3$ with $s_3>1$. (c) $\&$ (f) Shifting and expanding disc $\Psi_3$ with $s_3>1$ and $s_3^{-1}<1$. The center of $\Psi_3$ is unchanged. The radii of disc $\Psi_2$ and $\Psi_4$ are shrunk because $s_3^{-1}<1$. In (a), (b) and (c), we show disc operations on coefficient matrices $\B$. In (d), (e) and (f), red dots and blue arrows represent disc centers and radii respectively.
}
\label{fig:scale_op_example}
\end{figure}

\subsection{Formulating the Dual Problem}
The primal problem \eqref{eq:max_lower_bound} is difficult to solve because of the combinatorial nature of $a_i$'s and the max-min objective.
Instead, we formulate a \textit{dual problem}\footnote{Though the dual problem is not exactly the same definition as classical optimization setups in linear programming \cite{murty1983linear} and convex programming \cite{convex_optimization}, the essence of exchanging the roles of objective function and constraints to obtain a related optimization problem is identical.} \eqref{eq:disc_alignment} that is essentially a rewriting of \eqref{eq:max_lower_bound} with the objective and the constraint on $\sum_{i=1}^N a_i$ swapping roles.
Specifically, given a graph $\cG$ and a lower bound target $T$, \eqref{eq:disc_alignment} minimizes the number of sampled nodes subject to a constraint to lower-bound all left-ends:
\begin{align}
    \min_{\a,\s}&~~\sum_{i=1}^N a_i \label{eq:disc_alignment} \\
    \mbox{s.t.}~&\C=\bS \left( \A + \mu \L \right) \bS^{-1}, ~~~ c_{ii}-\sum_{j\ne i}|c_{ij}|\ge T, ~~ \forall i \notag\\
    &\A=\diag(\a), ~~~~~ a_{i}\in\{0,1\}, \notag \\
    &\bS\;=\diag(\s), ~~~~~ s_i > 0. \notag
\end{align}
Instead of a max-min criterion, the objective in \eqref{eq:disc_alignment} is now linear in $a_i$'s.

Optimal solutions to the primal problem (\ref{eq:max_lower_bound}) and the dual problem (\ref{eq:disc_alignment})
are related via the following proposition:

\begin{proposition}
If $T$ is chosen such that there exists at least one optimal solution $(\ha, \hs)$ to the dual problem (\ref{eq:disc_alignment}) satisfying $\sum_i \hat{a}_i = K$, then there exists one optimal solution to the dual problem (\ref{eq:disc_alignment}) at lower bound target $T$ that is also an optimal solution to the primal problem (\ref{eq:max_lower_bound}).
\label{prop:primal_dual}
\end{proposition}
\begin{proof}
For feasible solutions $(\a, \s)$ to the dual problem (\ref{eq:disc_alignment}) with $\sum_i a_i > K$, they are not feasible to the primal problem (\ref{eq:max_lower_bound}), and thus can be ignored. For each solution $(\a, \s)$ to the dual problem (\ref{eq:disc_alignment}) with $\sum_i a_i < K$ satisfying all constraints except $\min_i c_{ii} - \sum_{j\neq i} |c_{ij}| \geq T$, they must have $\min_i c_{ii} - \sum_{j\neq i} |c_{i,j}| < T$, otherwise $(\ha, \hs)$ would not be the optimal solution with $\sum_i \hat{a}_i = K$. Thus this solution $(\a, \s)$, while feasible in the primal, has a worse objective than $(\ha, \hs)$ in the primal problem. So the optimal solution to the primal problem (\ref{eq:max_lower_bound}) must be a solution to the dual problem (\ref{eq:disc_alignment}) with $\sum_i a_i = K$ (though not necessarily the one $(\ha, \hs)$ if multiple optimal solutions exist).
\end{proof}
Denote by $(\ha,\hs)$ an optimal solution to the dual problem (\ref{eq:disc_alignment}) given lower bound target $T$, and $\hK_T$ the corresponding objective function value. We can further prove that $\hK_T$ is non-decreasing with respect to $T$.
\begin{proposition}
    Let $\hK_{T}$ be the optimal objective function value to (\ref{eq:disc_alignment}) given lower bound $T$.
    Then, we have $\hK_{T_1}\ge \hK_{T_2}$, $\forall T_1\ge T_2$.
    \label{pr:non_decreasing}
\end{proposition}
\begin{proof}
    Since $T_1\ge T_2$, $\hK_{T_1}$ is a feasible objective function value to (\ref{eq:disc_alignment}) given lower bound $T_2$. Hence, $\hK_{T_1}\ge \hK_{T_2}$ must be satisfied, since $\hK_{T_2}$ is the optimal (minimum) objective function value to (\ref{eq:disc_alignment}) given $T_2$.
\end{proof}

Proposition~\ref{prop:primal_dual} and Proposition~\ref{pr:non_decreasing} provide us with an important insight:

\begin{quote}
\textit{
Instead of solving the primal problem  (\ref{eq:max_lower_bound}) directly,  we can solve the dual problem (\ref{eq:disc_alignment}) iteratively for different $T\in(0,1)$, where we employ a binary search to seek the largest $\hT$ such that the objective value of (\ref{eq:disc_alignment}) is $K$ exactly.}
\end{quote}
When this largest $\hT$ is found, we achieve an optimal solution $(\ha,\hs)$ to (\ref{eq:max_lower_bound}), and $\hT$ is the corresponding optimal objective value.
When solving the dual problem (\ref{eq:disc_alignment}), we select the smallest sampling set and scale radii of all Gershgorin discs, so that the discs' left-ends (the lower bound of $\lambda_{\min}(\B)$) is no smaller than target $T$.
Hence, we refer to the dual problem (\ref{eq:disc_alignment}) also as a \textit{disc alignment} problem.

When left-ends of all discs are aligned to $T$ as closely as possible, not only is the smallest left-end of Gershgorin discs a lower bound for $\lambda_{\min}(\B)$, {as stated in GCT}, but we can also prove that the largest left-end gives an upper bound:

\begin{proposition}
    Let $\Gamma_i$ be the Gershgorin disc from the $i$-th row of matrix $\C=\bS\B\bS^{-1}$, where $\bS$ is a non-singular diagonal matrix satisfying $\bS=\diag(\s)$ where $\s > \mathbf{0}$.
    Denote by $\ell_i$ the left-end of $\Gamma_i$, we have
    \begin{equation}
        \min_i \ell_i \le\lambda_{\min}(\B)\le \max_i \ell_i
        \label{eq:eigenvalue_bound}
    \end{equation}
    In words, while the smallest left-end $\ell_i$ gives a lower bound for $\lambda_{\min}(\B)$, the largest left-end $\ell_i$ gives an upper bound.
    \label{pr:yuji_bound}
\end{proposition}

The proof of Proposition~\ref{pr:yuji_bound} is in the Appendix~\ref{ap:yuji}.
Note that $\lambda_{\min}(\B)$ can be exactly the lower bound $T$, if the left-ends of all Gershgorin discs are aligned at $T$ satisfying $\min_i \ell_i = \max_i \ell_i$.

\section{Sampling Set Selection Algorithm Development}
\label{sec:sampling}
Although, unlike the primal \eqref{eq:max_lower_bound}, the dual \eqref{eq:disc_alignment} has a linear objective, it is nonetheless difficult to solve optimally.
Instead of manipulating all discs simultaneously, we first define a coverage subset $\Omega_i$ for each node $i$---the maximum subset of discs whose left-ends can be aligned at or beyond $T$ by sampling only node $i$. We introduce an approximation algorithm to pre-compute each coverage subset $\Omega_i$ (\textit{Algorithm~\ref{al:1}}) in Section~\ref{subsec:disc_aligned_subset}.

Pre-computed coverage subsets lead to a coarse but intuitive reinterpretation of the dual \eqref{eq:disc_alignment} as the disc coverage problem, which we prove to be NP-hard via reduction from the famous \textit{set cover} (SC) problem in Section~\ref{subsec:disc_coverage}.
The reinterpretation enables us to derive a fast algorithm for the dual \eqref{eq:disc_alignment} based on a known SC approximation algorithm (\textit{Algorithm~\ref{al:2}}) in the combinatorial optimization literature in Section~\ref{subsec:disc_alignment}.

We further employ a binary search to seek the largest $\hT$ such that the objective of the dual \eqref{eq:disc_alignment} is $K$ exactly. The complete sampling algorithm (\textit{Algorithm~\ref{al:3}}) is proposed in Section~\ref{subset:binary_search}, which is summarized in Fig.~\ref{fig:bs_gda}.
Finally, we discuss the computational complexities of different graph sampling approaches in Section~\ref{subsec:complexity}.

\begin{figure}[!t]
\centering
\includegraphics[width=0.92\linewidth]{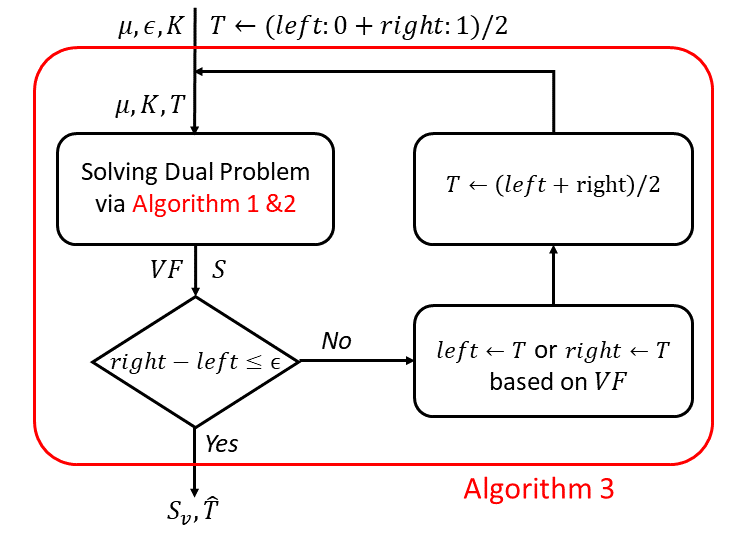}
\caption{The workflow of the complete sampling algorithm (Algorithm~\ref{al:3}). Algorithm~\ref{al:1} and \ref{al:2} solve the dual \eqref{eq:disc_alignment}. Binary search seeks the largest $\hT$ such that the objective of the dual \eqref{eq:disc_alignment} is $K$ exactly. }
\label{fig:bs_gda}
\end{figure}

\subsection{Coverage Subset of Each Sampled Node}
\label{subsec:disc_aligned_subset}
Before we approximately solve \eqref{eq:disc_alignment}---align all Gershgorin discs' left-ends at or beyond target $T\in(0,1)$ using as few sampled nodes as possible---we first define a \textit{coverage subset} $\Omega_i$ for each node $i$.
\begin{definition}[Coverage subset]
    A coverage subset $\Omega_i$ for each node $i$ is the optimal subset $\cX$ satisfying
    \begin{align}
        \max_{\s}&~~|\cX| \\
        \mbox{s.t.}~&\C=\bS \left( \A+\mu\L \right) \bS^{-1} \notag\\
        &\A=\diag(\a), ~~~ a_{i}=1, ~~~a_{j\ne i}=0 \notag \\
        &\bS=\diag(\s), ~~~ s_j > 0 \notag \\
        &\cX=\{j|~c_{jj}-\sum_{k\ne j} |c_{jk}|\ge T\}, ~~~\forall j,k\in\{1,\dots,N\}\notag
    \end{align}
    In words, a coverage subset $\Omega_i$ is the maximum subset of discs whose left-ends can be aligned at or beyond $T$ by sampling only node $i$.
    \label{def:coverage_subset}
\end{definition}

We propose an approximation algorithm to estimate each coverage subset $\Omega_i$.
As described in Section~\ref{subsec:max_lower_bound}, disc scaling operation can balance left-ends of discs in the local neighbourhood of a sampled node.
Starting from the sampled node $i$, we progressively employ disc scaling operations to align left-ends of discs at $T$ with increasing hops from the node $i$, as illustrated in Fig.~\ref{fig:align_op_general}. The algorithm stops when farther nodes cannot be aligned.

\begin{figure}[!t]
\centering
\includegraphics[width=0.7\linewidth]{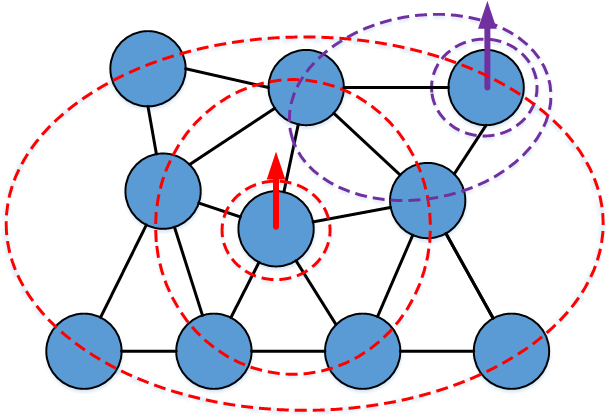}
\caption{Coverage subset estimation. Starting from each sampled node $i$ (marked by vertical arrows), we progressively employ disc scaling operations to align left-ends of discs at $T$ with increasing hops from node $i$. The algorithm stops when farther nodes cannot be aligned.}
\label{fig:align_op_general}
\end{figure}

Specifically, we first sample node $i$ (thus moving the corresponding disc $\Psi_i$'s center $b_{ii}$ from $\mu d_i$ to $1 + \mu d_i$, and $\Psi_i$'s left-end $a_{i}$ from $0$ to $1$).
We next apply scalar $s_i$ to \textit{expand} $\Psi_i$'s radius $R_i$ and align its left-end at exactly $T$.
Scalar $s_i$ must hence satisfy
\begin{equation}
    a_{i}+\mu \, \left(d_i-s_i\cdot\sum_{j\in \mathcal{N}_i}\frac{w_{ij}}{s_j}\right)=T,
    \label{eq:scale_function}
\end{equation}
where initially $s_j=1, \forall j$.
$\cN_i$ denotes the neighbors of node $i$ satisfying $w_{ij}>0$. Solving for $s_i$ in (\ref{eq:scale_function}), we obtain
\begin{equation}
    s_i=\frac{a_{i}+\mu \, d_i-T}{\mu \, \sum_{j\in \cN_i} \frac{w_{ij}}{s_j}}.
    \label{eq:scale_factor}
\end{equation}

Computing $s_i$ for sample node $i$ using \eqref{eq:scale_factor} implies that $s_i > 1$ (\textit{i.e.}, expansion of $\Psi_i$'s radius),
which means that node $i$'s neighbors' discs $\Psi_j$'s radii will \textit{shrink} due to $s_i^{-1}$.
Specifically, disc left-end $b_{jj} - R_j$ of an unsampled neighbor $j$'s disc (thus $a_{j}=0$) is now:
\begin{equation}
b_{jj} - R_j = a_{j} + \mu \left( d_j - s_j \cdot \sum_{k \in \mathcal{N}_j \setminus \{i\}} \frac{w_{jk}}{s_k}
- s_j \cdot \frac{w_{ji}}{s_i} \right).
\end{equation}

If a neighboring disc $\Psi_j$'s left-end remains smaller than $T$, then nothing further needs to be done.
However, if $\Psi_j$'s left-end is larger than $T$, then $j$ belongs to $\Omega_i$.
We then subsequently expand its radius to align its left-end at $T$ using \eqref{eq:scale_factor}.
This shrinks the disc radii of node $j$'s neighbors and so on; The sequence of nodes explored starting from sample $i$ is determined using \textit{Breadth First Search} (BFS) \cite{cormen2009introduction}.
Expansion factor $s_j$ decreases with hops away from node $i$.
Since we always expand a current disc ($s_j \ge 1$) leading to shrinking of neighboring discs (${s_j}^{-1} \le 1$) in each step, the left-end of each already scaled node remains larger than or equal to $T$.

Note that $1 \leq |\Omega_i| \leq N$ for any value of $T \in (0, 1)$.
When $T\rightarrow 0$, $|\Omega_i|\rightarrow N$, and we need to visit all $N$ nodes using BFS.
The worst-case time complexity for computing a single $\Omega_i$ is $\cO(N+M)$, where $M=|\cE|$.
Typically, sampling budget $K$ is much larger than $1$, and it is not necessary to compute $\Omega_i$ corresponding to a very small $T$.
To enable faster computation, we set a parameter $p$ to estimate $\Omega_i$ only within a $p$-hop neighborhood from node $i$.
Thus, the time complexity of estimating $\Omega_i$ within $p$ hops is $\cO(P+Q)$, where $P$ is the number of nodes within $p$ hops from node $i$, and $Q$ is the number of edges within $p$ hops from node $i$.
Assuming a maximum degree per node $d_{\max}$, we have $P\le Q\le d^p_{\max}$. $d^p_{\max}$ is achieved only if graph $\cG$ is a tree.
Since $p$ is typically set much smaller than the number of hops required to traverse the entire graph with BFS, $P\ll N$ and $Q\ll M$.
The procedure of estimating coverage subset $\Omega_i$ is sketched in Algorithm~\ref{al:1}.

\begin{figure}[!t]
\centering
\subfloat[]{
\label{fig:adj_matrix_s3}
\includegraphics[width=0.32\linewidth]{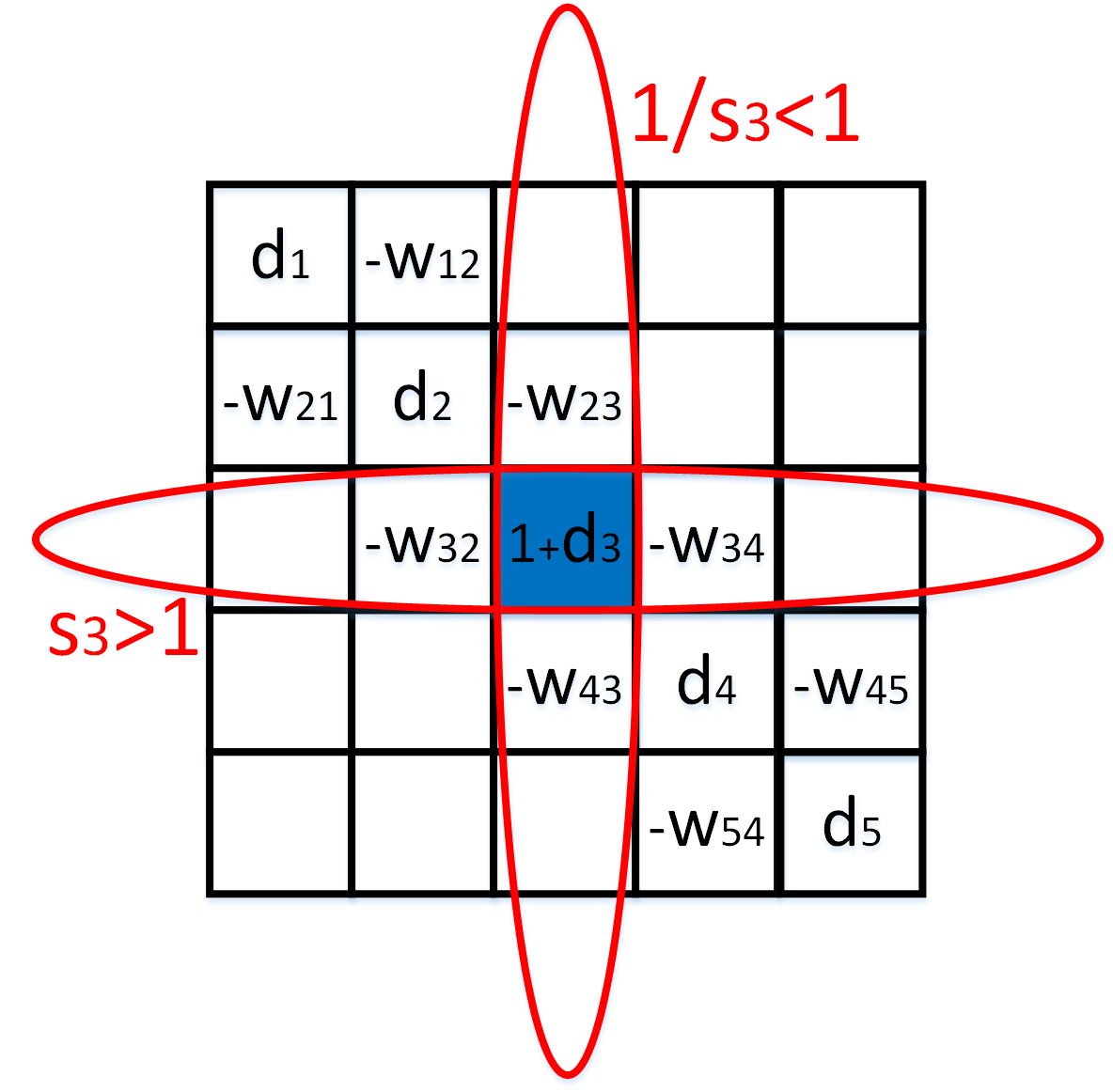}}
\subfloat[]{
\label{fig:adj_matrix_s4}
\includegraphics[width=0.32\linewidth]{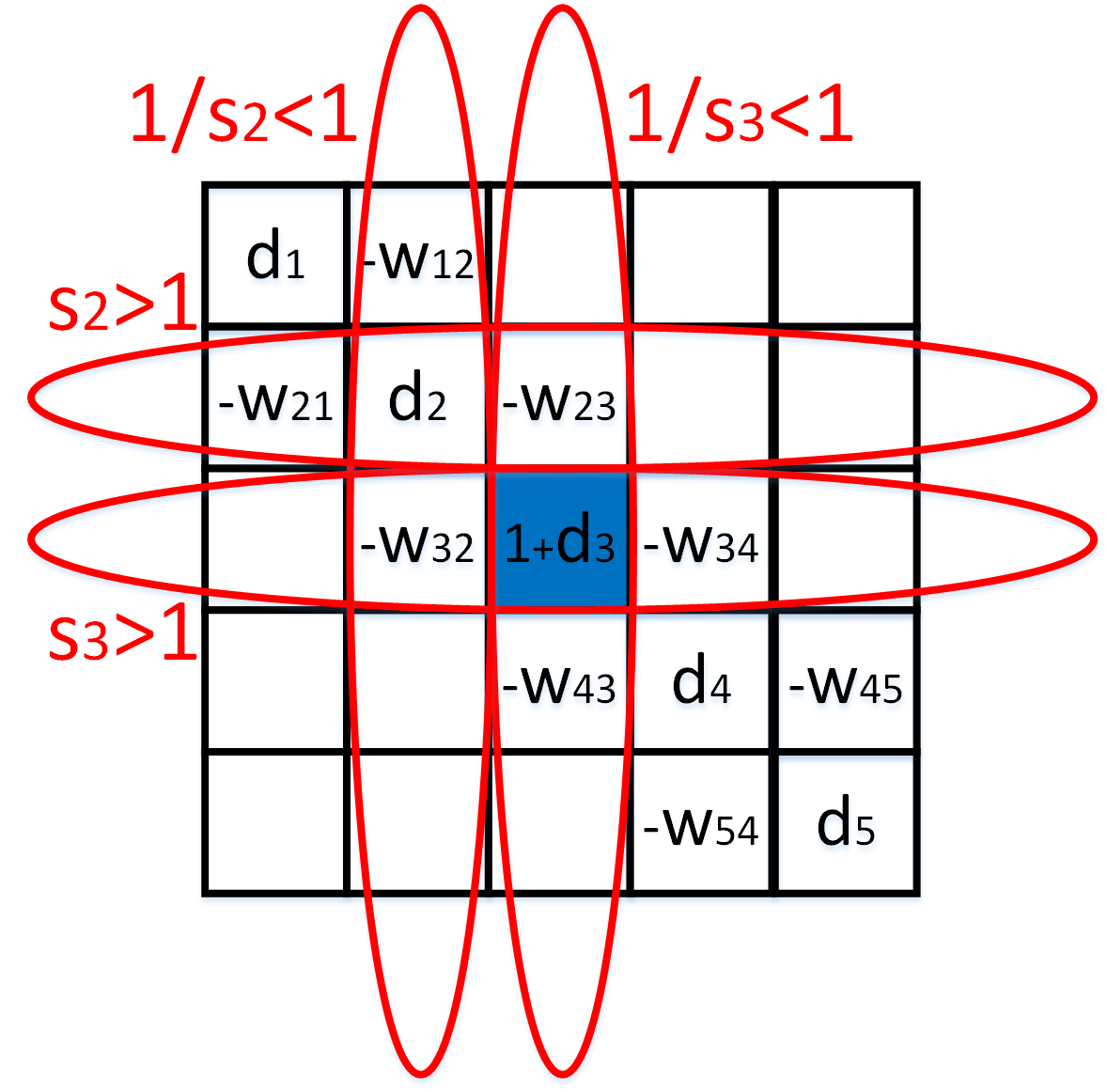}}
\subfloat[]{
\label{fig:adj_matrix_s5}
\includegraphics[width=0.32\linewidth]{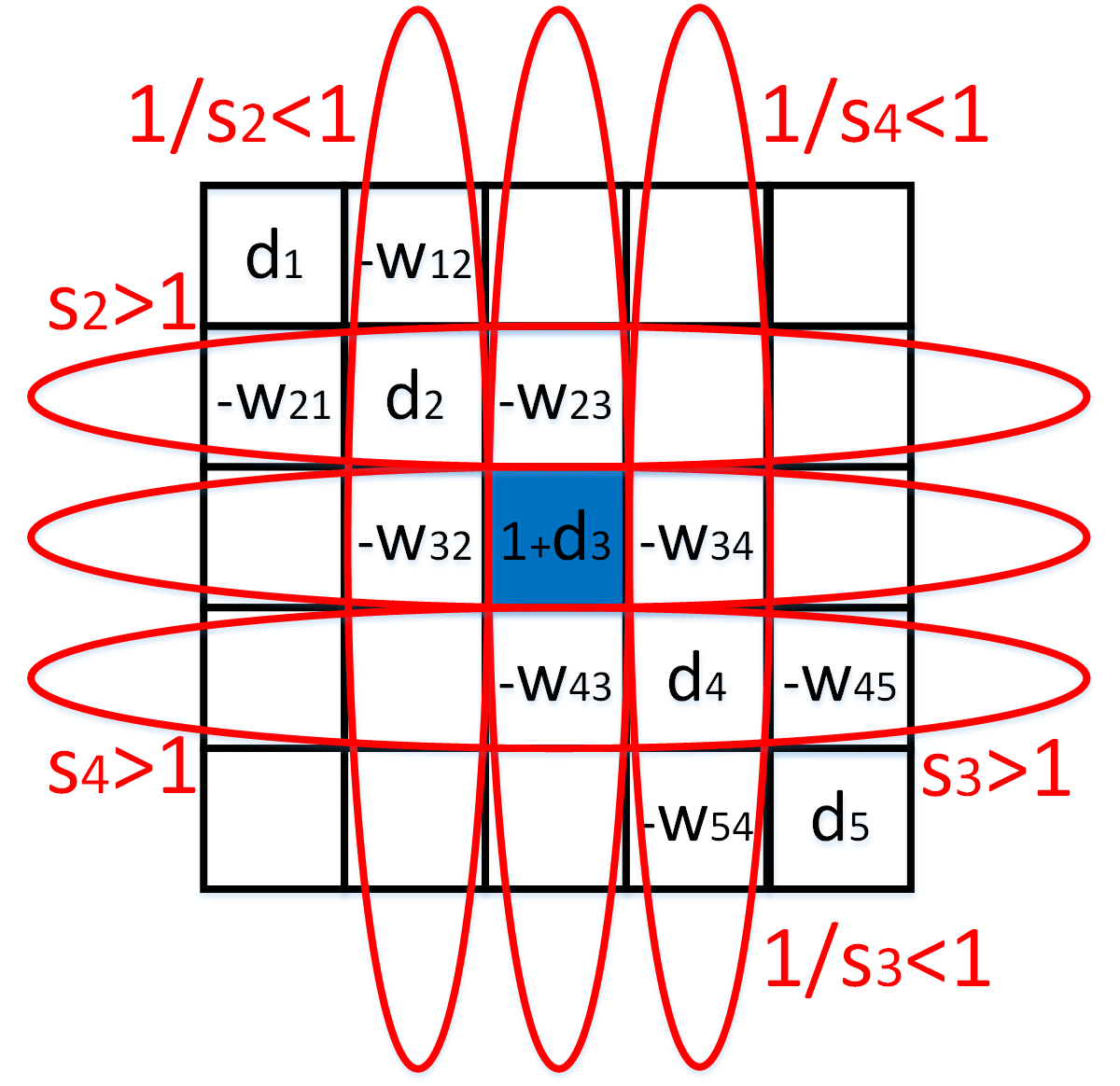}} \\
\subfloat[]{
\label{fig:disk_e4}
\includegraphics[width=0.32\linewidth]{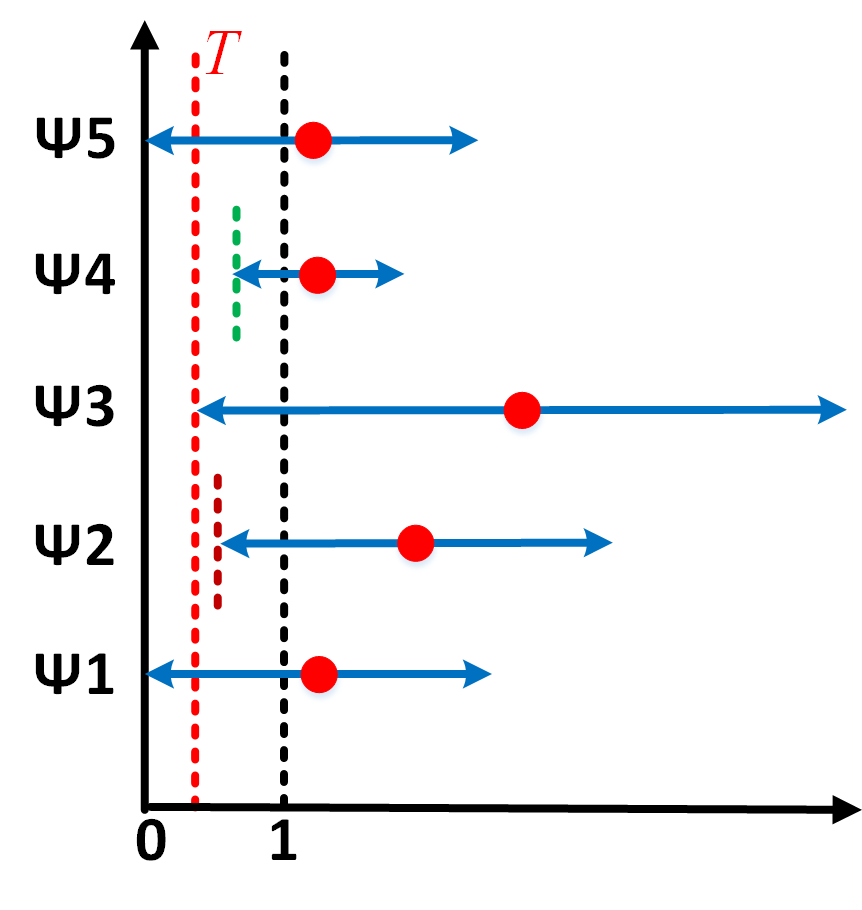}}
\subfloat[]{
\label{fig:disk_e5}
\includegraphics[width=0.32\linewidth]{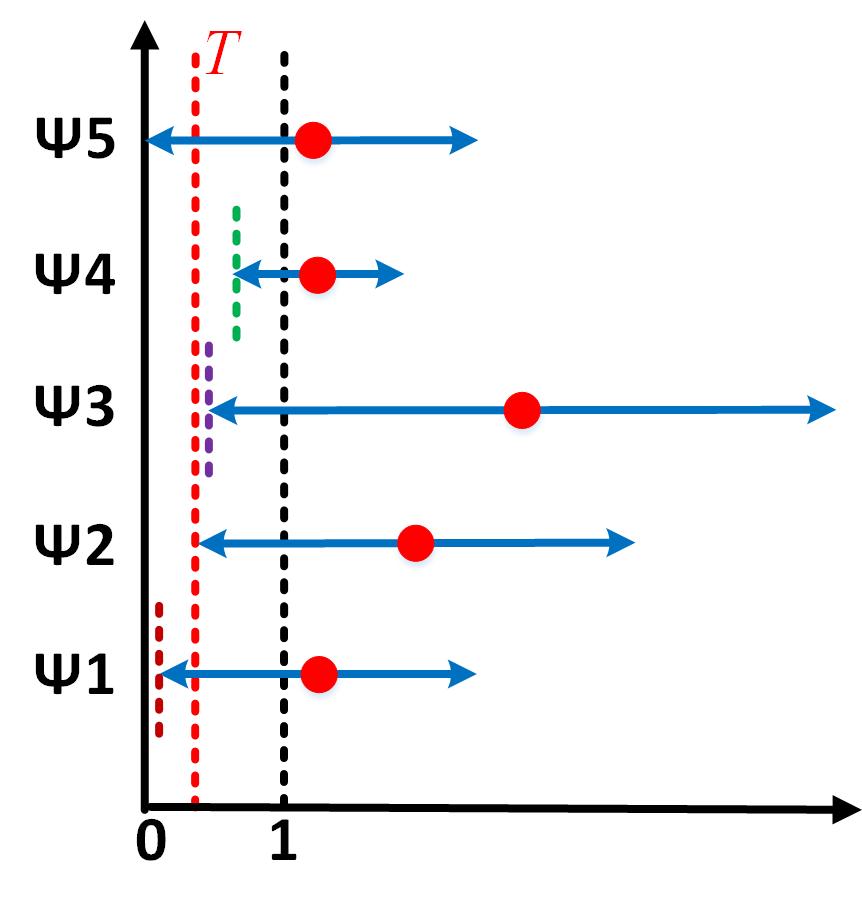}}
\subfloat[]{
\label{fig:disk_e6}
\includegraphics[width=0.32\linewidth]{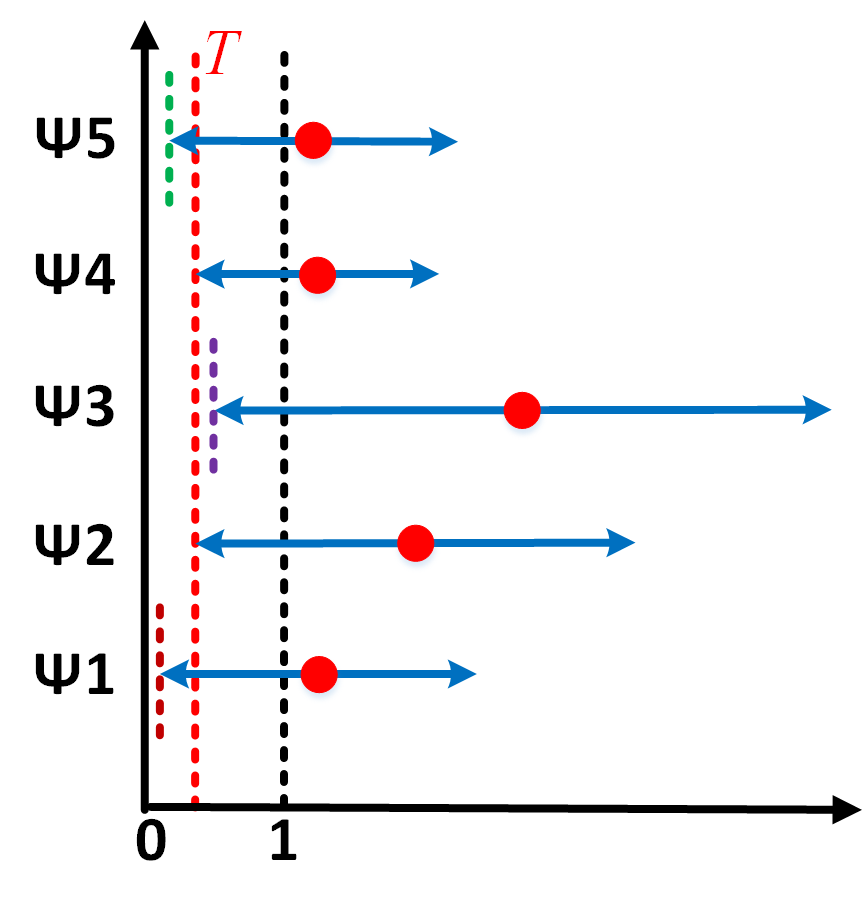}}
\caption{An illustration of estimating $\Omega_3$ on a $5$-node path graph. (a) \& (d) aligning left-end of disc $\Psi_3$ to $T$. (b) \& (e) aligning left-end of disc $\Psi_2$ to $T$. (c) \& (f) aligning left-end of disc $\Psi_4$ to $T$. In (a), (b) and (c), we show disc operations on coefficient matrices $\B$. In (d), (e) and (f), red dots and blue arrows represent disc centers and radii respectively.
}
\label{fig:align_op}
\end{figure}

\begin{algorithm}[ht]
\caption{Estimating Coverage Subset $\Omega_i$}
\label{al:1}
\begin{algorithmic}[1] 
\small
\REQUIRE  
Graph $\cG$, lower bound $T$, node $i$, hop constraint $p$ and $\mu$.\\
\STATE Initialize $\d=\W\mathbf{1}$, $\s=\mathbf{1}$ and $\a=[0,\ldots,0,a_i=1,0,\ldots,0]$. \\
\STATE Initialize $\h=\mathbf{0}$ for hop number.
\STATE Initialize $\Omega_i=\varnothing$. 
\STATE Initialize $\Q=\varnothing$ for enqueued nodes. \\
\STATE Initialize an empty $queue$. \\
\STATE $Enqueue(queue,i)$ and $\Q\leftarrow\Q\cup\{i\}$.\\

\STATE\textbf{while} $queue$ is not empty \textbf{do}\\
\STATE\ \ \ \ \ \ $k\leftarrow$Dequeue($queue$). \\
\STATE\ \ \ \ \ \ Update $s_k$ using (\ref{eq:scale_factor}).\\
\STATE\ \ \ \ \ \ \textbf{if} $s_k\ge1$ and $h_k\le p$ \textbf{do}\\
\STATE\ \ \ \ \ \ \ \ \ \ $\Omega_i\leftarrow\Omega_i\cup\{k\}$.
\STATE\ \ \ \ \ \ \ \ \ \ \textbf{for} $t$ \textbf{in} $k$'s neighbours $\cN_k$ \textbf{do}\\
\STATE\ \ \ \ \ \ \ \ \ \ \ \ \ \ \textbf{if} $t\notin\Q$ \textbf{do}\\
\STATE\ \ \ \ \ \ \ \ \ \ \ \ \ \ \ \ \ $Enqueue(queue,t)$ and $\Q\leftarrow\Q\cup\{t\}$.\\
\STATE\ \ \ \ \ \ \ \ \ \ \ \ \ \ \ \ \ $h_t\leftarrow h_k+1$.\\
\STATE\ \ \ \ \ \ \ \ \ \ \ \ \ \ \textbf{endif}\\
\STATE\ \ \ \ \ \ \ \ \ \ \textbf{endfor}\\
\STATE\ \ \ \ \ \ \textbf{endif}\\

\STATE\textbf{endwhile}\\

\ENSURE  Subset $\Omega_i$.
\end{algorithmic}
\end{algorithm}

\vspace{0.1in}
\noindent
\textit{Illustrative Example on 5-node Path Graph:}
We use an example to illustrate the operations of the above algorithm, shown in Fig.\;\ref{fig:align_op}.
We assume the same $5$-node path graph in Fig.~\ref{fig:nodes_example} and choose to sample node $3$.
Assuming $\mu=1$, the graph's coefficient matrix $\B$ with entry $(3,3)$ updated is shown in Fig.\;\ref{fig:adj_matrix_s0}.
The left-end of node 3's Gershgorin disc shifts from $0$ to $1$, as shown in Fig.\;\ref{fig:disk_e1}.

We next perform disc scaling to align the left-end of disc $\Psi_3$ to $T$.
As shown in Fig.\;\ref{fig:adj_matrix_s3}, scalar $s_3$ is applied to the third row of $\B$, and thus the radius of disc $\Psi_3$ is expanded by $s_3$, where $s_3>1$.
Simultaneously, scalar $s_3^{-1}$ is applied to the third column, and thus the radii of discs $\Psi_2$ and $\Psi_4$ are shrunk due to the scaling of $w_{23}$ and $w_{43}$ respectively by $s_3^{-1}$.
Note that entry $(3,3)$ of $\B$ (and $\Psi_3$'s disc center) is unchanged, since scalar $s_3$ is offset by $s_3^{-1}$.
We see that due to the expansion of disc $\Psi_3$ radius, the disc left-ends of its neighboring nodes (nodes $2$ and $4$) move beyond threshold $T$, as shown in Fig.\;\ref{fig:disk_e4}.

We next apply scalar $s_2$ to disc $\Psi_2$ to expand its radius by $s_2$, where $s_3>s_2>1$, and the radii of discs $\Psi_1$ and $\Psi_3$ are shrunk due to the scaling of $w_{12}$ and $w_{32}$ by $s_2^{-1}$, as shown in Fig.\;\ref{fig:adj_matrix_s4}.
$s_2$ must be smaller than $s_3$ for the left-end of $\Psi_2$ not to move past $0$.
The discs are shown in Fig.~\ref{fig:disk_e5}.
Scaling $\Psi_2$ does not move the left-end of disc $\Psi_1$ beyond threshold $T$.
Thus, we stop aligning the further-away disc $\Psi_1$.

Subsequently, similar disc operations can be performed on $\Psi_4$ as shown in Fig.\;\ref{fig:adj_matrix_s5} and \ref{fig:disk_e6}.
Scaling $\Psi_4$ does not move the left-end of disc $\Psi_5$ beyond $T$. Thus, we stop aligning the further-away disc $\Psi_5$.
Finally, the left-ends of discs $\Psi_2$, $\Psi_3$ and $\Psi_4$ are moved beyond threshold $T$, \ie, $\Omega_3=\{2,3,4\}$.

\vspace{0.1in}
\noindent
\textit{Analogy of Throwing Pebbles in a Pond:}
To impart intuition, we consider the following real-world analogy of computing subset $\Omega_i$ for sample $i$.
Suppose a planar graph $\mathcal{G}$ models the water surface of a pond, where nodes denote water surface locations.
Sampling a node $i$ is analogous to throwing a pebble at a water surface location $i$ in the pond.
The thrown pebble causes a ripple to adjacent locations $j \in \mathcal{N}_i$, where the magnitude of the ripple at location $j$ is smaller than at $i$.
Analogously, when aligning left-ends of Gershgorin discs at $T$, scalar $s_j$ of neighbor $j \in \mathcal{N}_i$ is necessarily smaller than $s_i$.
The ripple in the pond dissipates away from the thrown location $i$.
Analogously, scalars $s_j$ becomes smaller further away in hop count from sample $i$.
The ripple stops increasing in size when the dissipating fringes become unnoticeable.
Analogously, subset $\Omega_i$ stops growing when left-ends of neighboring discs fail to move beyond target $T$.

\subsection{NP-Hardness of Disc Coverage Problem}
\label{subsec:disc_coverage}

After pre-computing coverage subsets, we can reinterpret the dual \eqref{eq:disc_alignment} more coarsely but in a more intuitive manner:
given pre-computed coverage subsets $\Omega_i, \forall i \in \{1, \ldots, N\}$,
how to select minimum-cost sampling vector $\a$ so that every node $j$ belongs to at least one sample's coverage subset $\Omega_i$, \textit{i.e.},
\begin{align}
\min_{\a} & \sum_{i=1}^N a_i
\label{eq:disc_coverage} \\
\mbox{s.t.} & ~~~
j \in \bigcup_{i ~|~ a_i=1} \Omega_i,
~~~ \forall j \in \{1, \ldots, N\} \nonumber \\
& a_i \in \{0, 1\}  \nonumber
\end{align}
We term this the \textit{disc coverage} (DC) problem.
DC \eqref{eq:disc_coverage} is related to the dual \eqref{eq:disc_alignment} in the sense that: i) they have the same optimization objective, and ii) every feasible solution $\a$ to  \eqref{eq:disc_coverage} is also a feasible solution to \eqref{eq:disc_alignment}, where the scalars $\s$ corresponding to $\a$ are the ones calculated when coverage subsets $\Omega_i$'s are pre-computed.
Thus solving DC provides an approximate feasible solution to the dual.

Unfortunately, even this DC problem \eqref{eq:disc_coverage} is NP-hard.
A simple NP-hardness proof via a reduction from the famous NP-hard set cover (SC) problem \cite{cormen2009introduction} is provided in Appendix~\ref{ap:np}.
However, the resemblance of the DC problem to SC enables us to derive an efficient approximation algorithm to the dual \eqref{eq:disc_alignment}, which we describe next.

\subsection{Disc Alignment via Greedy Set Cover}
\label{subsec:disc_alignment}

Given nodes $\cU=\{1,2,\ldots,N\}$ and a collection of pre-computed $N$ coverage subsets $\Omega=\{\Omega_1,\ldots,\Omega_N\}$, we derive a fast algorithm to DC \eqref{eq:disc_coverage} based on an error-bounded approximation algorithm to SC.
We first describe SC formally.
\begin{definition}[Set cover problem]
    An instance $(\cX,\cF)$ of the set cover problem consists of a finite element set $\cX$ and a family $\cF$ of subsets of $\cX$, such that every element $x \in \cX$ belongs to at least one subset $\cS \in \cF$:
    \begin{equation}
        x \in \bigcup_{\cS\in\cF} \cS, ~~~
        \forall x \in \cX
    \end{equation}
    We say that a subset $\cS\in\cF$ covers its elements. The problem is to find a minimum-size subset collection $\cC\subseteq\cF$ whose members cover all of $\cX$:
    \begin{equation}
        x \in \bigcup_{\cS\in\cC} \cS, ~~~
        \forall x \in \cX
    \end{equation}
    \label{def:set_cover}
\end{definition}

Given the resemblance between DC and SC, we borrow a known greedy approximation algorithm for SC with error-bounded performance \cite{cormen2009introduction}:
at each iteration, we select a yet-to-be-chosen subset $\cS$ that covers the largest number of uncovered elements in $\cX$.
Similarly, for our algorithm for DC, we select sample $i$ and coverage subset $\Omega_i$ with the largest uncovered nodes at each iteration.
{When a sampling vector $\a$ satisfying the condition that each node belongs to at least one sample coverage subset $\Omega_i$ is computed, the greedy algorithm can achieve a $H(\max_i |\Omega_i|)$-approximate solution to DC \cite{cormen2009introduction}, \ie,
\begin{equation}
    |\a|\le H(\max_i |\Omega_i|)\cdot|\hat{\a}|,
\end{equation}
where $H(n)=\sum_{i=1}^n1/i$ is the harmonic series and $\hat{\a}$ is the optimal solution to DC. $|\a|$ denotes $\sum_i a_i$.}

Since we are given a sampling budget $K$, we terminate the algorithm when the number of selected samples reaches $K$ for practical implementation.
The resulting sampling vector $\a$ is valid only if $| \a| \le K$ \textit{and} each node belongs to at least one sample coverage subset $\Omega_i$.
The procedure of disc alignment via greedy set-covering is outlined in Algorithm~\ref{al:2}.

\begin{algorithm}[htb]
\caption{Disc Alignment via Greedy Set Cover}
\label{al:2}
\begin{algorithmic}[1] 
\small
\REQUIRE  
Graph $\cG$, lower bound $T$, sampling budget $K$ and $\mu$.\\
\STATE Initialize $\cU=\{1,\ldots,N\}$. \\
\STATE Initialize $\Omega=\varnothing$ and $\cS=\varnothing$. 
\STATE Initialize validity flag $VF=true$ and $n=1$. \\
\STATE \textbf{for} $i=1\rightarrow N$ \textbf{do}\\
\STATE\ \ \ \ \ \ Estimate $\Omega_i$ using Algorithm~1.
\STATE\ \ \ \ \ \ $\Omega\leftarrow \Omega\cup\{\Omega_i\}$.
\STATE \textbf{endfor} \\

\STATE \textbf{while} $|\cU|\ne\varnothing$ and $n\le K$ \textbf{do}\\
\STATE\ \ \ \ \ \ Select $\Omega_i\in\Omega$ that maximizes $|\Omega_i\cap\cU|$.
\STATE\ \ \ \ \ \ $\cU\leftarrow\cU\setminus\{\Omega_i\cap\cU\}$.
\STATE\ \ \ \ \ \ $\cS\leftarrow\cS\cup \{i\}$.
\STATE\ \ \ \ \ \ $n\leftarrow n+1$.
\STATE\textbf{endwhile}\\

\STATE\textbf{if} $|\cU|\ne\varnothing$ \textbf{do}\\
\STATE\ \ \ \ \ \ $VF\leftarrow false$.\\
\STATE\textbf{endif}\\

\ENSURE  validity flag $VF$, sampling set $\cS$.
\end{algorithmic}
\end{algorithm}

The time complexity of $N$ subsets preparation within $p$ hops is $\cO(PN+QN)$.
We represent uncovered elements in $\cU$ using $0$-$1$ bit array.
The uncovered elements in $\Omega_i$ to $\cU$ can be computed through at most $\cO(P)$ times fast bit operations, since the number of ``1'' in the bit array of $\Omega_i$ is at most $P$.
We compute uncovered elements of $N$ subsets at each subset selection stage.
We employ no more than $K$ times subset selections, thus the time complexity of the greedy subset selection algorithm is $\cO(NKP)$.
The total time complexity of Algorithm~\ref{al:2} comprises $\cO\left(PN+QN\right)$ node search plus $\cO\left(NKP\right)$ bit operations.

\subsection{Binary Searching the Maximum Lower Bound}
\label{subset:binary_search}
According to Proposition~\ref{pr:non_decreasing}, the number of sampled nodes $K$ is non-decreasing with respect to lower bound $T$ in the dual problem (\ref{eq:disc_alignment}).
Hence, we can perform binary search to find the maximum lower bound $\hT\in(0, 1)$ corresponding to $K$ sampled nodes.
We name the proposed algorithm \textit{Binary Search with Gershgorin Disc Alignment} (\texttt{BS-GDA}), as outlined in Algorithm~\ref{al:3}.

\begin{algorithm}[htb]
\caption{Binary Search with Gershgorin Disc Alignment}
\label{al:3}
\begin{algorithmic}[1] 
\small
\REQUIRE  
Graph $\cG$, sampling budget $K$, numerical precision $\epsilon$ and $\mu$.\\
\STATE Initialize $left=0$, $right=1$. \\
\STATE Initialize valid sampling set $\cS_v=\varnothing$. \\
\STATE\textbf{while} $right-left>\epsilon$ \textbf{do}\\
\STATE\ \ \ \ \ \ $T\leftarrow(left+right)/2$. \\
\STATE\ \ \ \ \ \ Estimate $VF, \cS$ given $T$ using Algorithm~2.\\
\STATE\ \ \ \ \ \ \textbf{if} $VF$ is $false$ \textbf{do}\\
\STATE\ \ \ \ \ \ \ \ \ \ $right\leftarrow T$ \\
\STATE\ \ \ \ \ \ \textbf{else} \\
\STATE\ \ \ \ \ \ \ \ \ \ $left\leftarrow T$ \\
\STATE\ \ \ \ \ \ \ \ \ \ $\cS_v\leftarrow S$ \\
\STATE\ \ \ \ \ \ \textbf{endif}\\
\STATE\textbf{endwhile}\\
\STATE $\hT\leftarrow left$.

\ENSURE  Sampling set $\cS_v$, maximum lower bound $\hT$.

\end{algorithmic}
\end{algorithm}

At each iteration, we estimate validity flag $VF$ and sampling set $\cS$ given a certain lower bound $T$ using Algorithm~\ref{al:2}. If $VF=false$, then the lower bound $T$ is set too large, and we update $right$ to reduce $T$. Otherwise, $VF=true$, then $T$ may be still too small, leading to $|\cS|<K$ or not the maximum $T$ corresponding to $|\cS|=K$. We update $left$ to increase $T$.
When $right-left\le\epsilon$, \texttt{BS-GDA} converges and we find the maximum lower bound $\hT$ with numerical error lower than $\epsilon$.

In order to achieve numerical precision $\epsilon$ in \texttt{BS-GDA}, we need to employ $\cO(\log_2\frac{1}{\epsilon})$ times binary search. For instance, to achieve numerical precision $\epsilon=10^{-3}$, we employ $10$ times binary search.
The total time complexity of Algorithm~3 comprises $\cO\left(N(P+Q)\log_2\frac{1}{\epsilon}\right)$ node search plus $\cO\left(NKP\log_2\frac{1}{\epsilon}\right)$ bit operations.

\subsection{Complexity Analysis}
\label{subsec:complexity}

We here summarize the theoretical computational complexities of different graph sampling approaches, and list them  in Table\;\ref{complexity}, assuming that the signal bandwidth $\omega$ is the $K$-th eigenvalue $\theta{_K}$ of $\L$ and the graph is sparse.
We separate the complexity into two parts: \emph{preparation} (for eigen-decomposition or computing initial preparation information) and \emph{sampling} (during greedy sampling step).
The competing algorithms are referred to as  \texttt{E-optimal} \cite{e_optimal2015}, Spectral proxies (\texttt{SP}) \cite{sp_proxy2016}, \texttt{MFN} \cite{MFN2016TSP},  \texttt{MIA} \cite{MIA2018Fen} and Eigendecomposition-free (\texttt{Ed-free}) \cite{akie2018eigenFREE}, respectively.
Complexities of methods \texttt{E-optimal}, \texttt{SP}, \texttt{MFN} and \texttt{MIA} are directly borrowed from paper \cite{MIA2018Fen}, with modified notations. For \texttt{E-optimal} and \texttt{MFN}, $T_1$ is the convergence steps for computing the first $K$ eigen-pairs of Laplacian $\L$ and $R$ is a constant for combining complexities. In strategy \texttt{SP}, $k$ is the degree of spectral proxies and $T_2$ is the convergence steps for obtaining the first eigen-pair of $\L$ via LOBPCG \cite{sp_proxy2016}. The parameters $q$ and $F$ in \texttt{MIA} are the degree of Chebyshev polynomials for a low-pass filter approximation and the truncated parameter of a Neuamnn series, respectively.
As for \texttt{Ed-free} method, paper \cite{akie2018eigenFREE} claimed that its complexity was dominated by the non-zero entries in matrix $\L^{\bar{q}}$, \textit{i.e}, $J$ in paper \cite{akie2018eigenFREE}, where $\bar{q}$ is degree of Chebyshev polynomial approximation. Similarly, the non-zero entries in $\L^{\bar{q}}$ is at most $Nd_{\max}^{\bar{q}}$ in a degree-bounded graph, so we use $\bar{Q}$ to denote the largest number of edges in $\bar{q}$-hop neighborhood, which implies $J=\cO(N\bar{Q})$.
From Table\;\ref{complexity}, we can see that the complexities of \texttt{Ed-free} and our proposed \texttt{BS-GDA} method are both roughly linear in terms of graph size $N$, which is substantially lower than other eigen-decomposition-based sampling methods and implies those two methods are applicable in fairly large graphs. The explicit execution time comparison will be demonstrated in the coming experimental section.

\begin{table*}
\caption{Theoretical Complexity Comparison of Different Graph Sampling Strategies. }
\label{complexity}
\begin{center}
\begin{tabular}{ccccccc}
 \hline
\hline
  \textbf{}&{\texttt{E-optimal}}&\texttt{SP}&\texttt{MFN}&\texttt{MIA}&\texttt{Ed-free}&proposed \texttt{BS-GDA}\\\hline
Preparation
&$\mathcal{O}\left( {\left( { MK + R{K^3}} \right){T_1}} \right)$
&$\mathcal{O}\left({kMK{T_2}}\right)$
&$\mathcal{O}\left( {\left( {MK + R{K^3}} \right){T_1}} \right)$
& $\mathcal{O}(qNM)$
&$\cO(M\bar{q}+N\bar{Q})$
&$\cO\left(N(P+Q)\log_2\frac{1}{\epsilon}\right)$
\\
Sampling
&$\mathcal{O}\left(NK^{4}\right)$
&$\mathcal{O}\left(NK\right)$
&$\mathcal{O}\left(NK^{4}\right)$
&$\mathcal{O}\left(NFK^{3.373}\right)$
&$\cO(NK\bar{Q})$
&$\cO\left(NKP\log_2\frac{1}{\epsilon}\right)$
\\
\hline
 \end{tabular}
 \end{center}
\end{table*}

\section{Experiments}
\label{sec:experiments}
\subsection{Experimental Setup}
\label{subsec:experimental_design}
We demonstrate the efficacy of our proposed fast graph subset sampling algorithm (\texttt{BS-GDA}) via extensive simulations.
Our experimental platform is Windows~10 desktop computer with Intel i5-4670K CPU and 24GB memory.
All algorithms are performed on MATLAB R2016b.

We use four types of graphs in GSPBOX \cite{perraudin2014gspbox} for testing:
\begin{enumerate}
    \item Random sensor graph with $N$ nodes.
    \item Community graph with $N$ nodes and $\lfloor\sqrt{N}/2\rfloor$ random communities.
    \item Barab{\'{a}}si-Albert graph with $N$ nodes.
    Each graph is constructed by adding new nodes, each with $m=1$ edge, that are preferentially attached to existing nodes with high degrees.
    \item Minnesota road graph with fixed $N=2642$ nodes.
\end{enumerate}

Random sensor graphs are weighted sparse graphs generated by GSPBOX, where each node connects to its six nearest neighbours. Edge weights are computed using:
\begin{equation}
    w_{ij}=\exp \left\{-\frac{\|\x_i-\x_j\|^2_2}{\sigma_x^2} \right\}
    \label{eq:weight_function_experiment}
\end{equation}
where $\x_i$ is the 2D coordinate of node $i$, and $\sigma_x$ is set automatically by the toolbox.
Community graphs, Barab{\'{a}}si-Albert graphs and Minnesota road graph are unweighted sparse graphs generated by GSPBOX, where
we manually compute weight for each edge.
For community graphs and Minnesota road graph, we compute edge weights using \eqref{eq:weight_function_experiment}.
We set parameter $\sigma_x=1$ for community graphs and set parameter $\sigma_x=0.1$ for Minnesota road graphs.
Since Barab{\'{a}}si-Albert graphs are scale-free graphs without 2D coordinates, we randomly generate edge weights using a uniform probability distribution with interval $(0,1)$. Illustrations of four types of graphs are shown in Fig.~\ref{fig:four_graphs}.

\begin{figure}[!t]
\centering
\subfloat[]{
\label{fig:random_sensor_graph}
\includegraphics[width=0.45\linewidth]{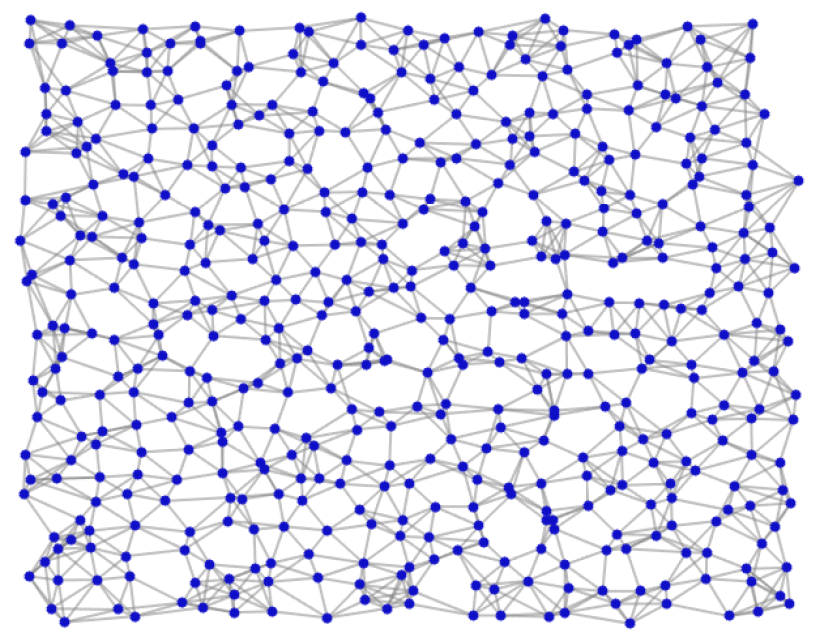}}~~
\subfloat[]{
\label{fig:community_graph}
\includegraphics[width=0.45\linewidth]{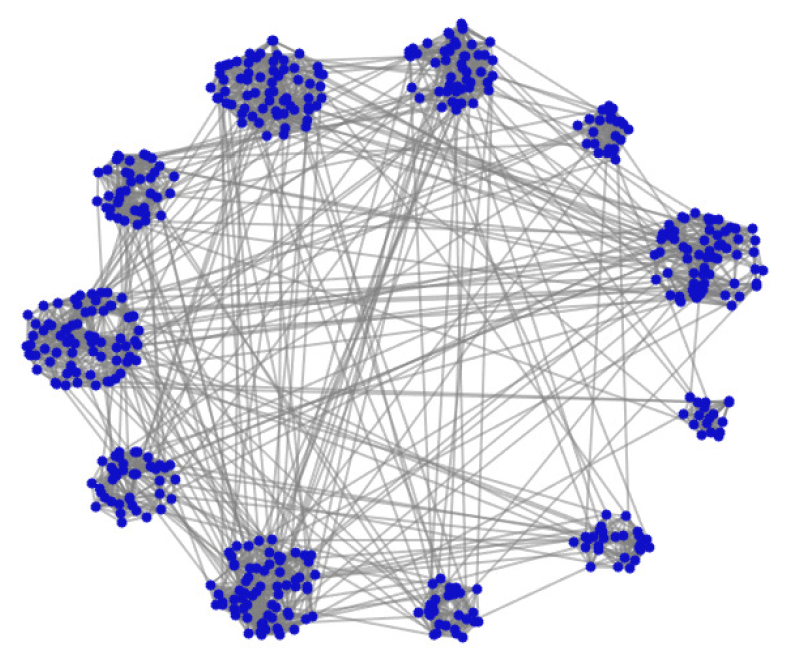}} \\
\subfloat[]{
\label{fig:BA_graph}
\includegraphics[width=0.45\linewidth]{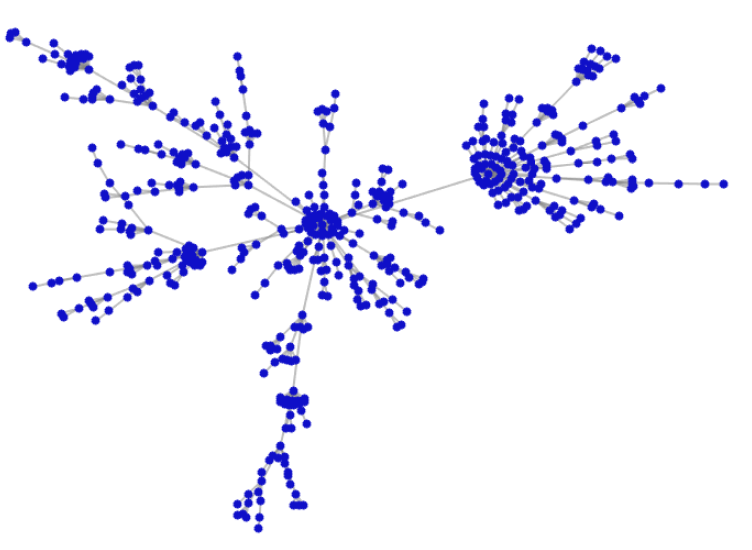}}~~
\subfloat[]{
\label{fig:Minnesota_road_graph}
\includegraphics[width=0.45\linewidth]{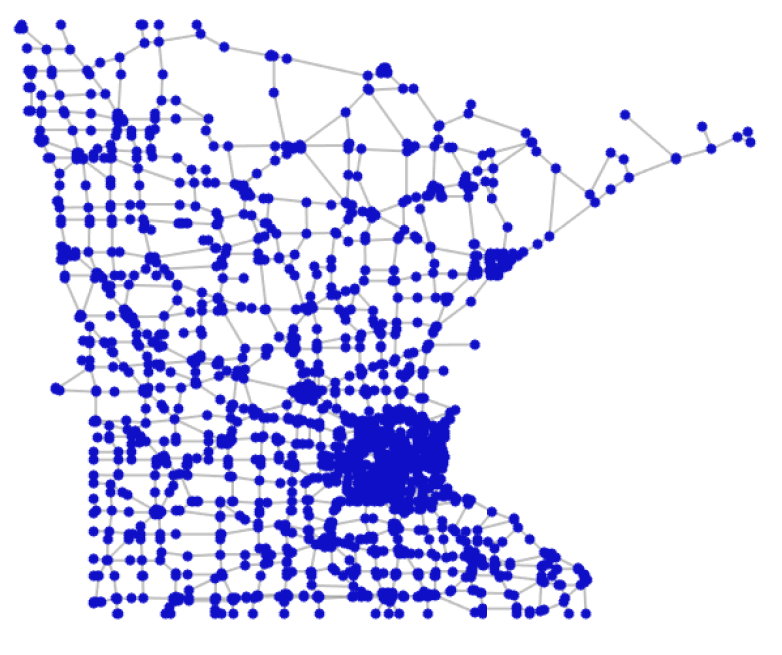}}
\caption{Illustrations of four types of graphs. (a) Random sensor graph. (b) Community graph. (c) Barab{\'{a}}si-Albert graph. (d) Minnesota road graph.}
\label{fig:four_graphs}
\end{figure}

For each graph type, we consider two types of graph signals described as follows:
\begin{enumerate}
    \item \textbf{GS1}: The true signals are exactly $\omega$-bandlimited, where $\omega=\theta_{\lfloor N/10\rfloor}$ is the $\lfloor N/10\rfloor$-th eigenvalue of $\L$.
    For example, $\omega=\theta_{50}$ for $N=500$.
    The non-zero GFT coefficients are randomly generated from {$\cN(0,10)$}.
    We then add i.i.d. Gaussian noise generated from $\cN(0,0.1^2)$ such that the SNR equals to $20$\;dB.
    \item \textbf{GS2}: The true signals are generated from multivariate Gaussian distribution $\cN\left(\mathbf{0},(\L+\delta\I)^{-1}\right)$, where $\delta=10^{-5}$. Because the power of the generated graph signals is inconsistent,
    we normalize the signals using $\x'=\frac{\x-mean(\x)}{std(\x)}$, where $mean(\x)=\sum_i x_i/N$ and $std(\x)=\sqrt{\frac{\sum_i\left(x_i-mean(\x)\right)^2}{N}}$, respectively.
    We then add i.i.d. Gaussian noise generated from $\cN(0,0.1^2)$ such that the SNR equals to $20$\;dB.
\end{enumerate}

We compare the performance of our proposed algorithm with six recent graph sampling set selection methods, of which implementations are available:
\begin{enumerate}
    \item Random graph sampling with non-uniform probability distribution (referred to as \texttt{Random}) \cite{r_sampling2018ACHA}.
    \item Deterministic sampling set selection methods, including  \texttt{E-optimal} \cite{e_optimal2015}, Spectral proxies (\texttt{SP}) \cite{sp_proxy2016}, \texttt{MFN} \cite{MFN2016TSP}, \texttt{MIA} \cite{MIA2018Fen} and Eigendecomposition-free (\texttt{Ed-free}) \cite{akie2018eigenFREE}, respectively.
\end{enumerate}

Our proposed algorithm \texttt{BS-GDA} avoids expensive eigen-decomposition and matrix computations, thus we build MEX functions for the iterative BFS and the greedy set covering, leading to a MATLAB+MEX implementation.
For all experiments, we set the target numerical precision $\epsilon=10^{-5}$ for binary search and $\mu=0.01$.
We empirically set hop constraint $p=12$ to consider no more than $12$ hops for each subset $\Omega_i$ estimation.
It is analogus to the neighborhood considered by Chebyshev polynomial approximation with order $12$ in \cite{akie2018eigenFREE}.

\subsection{Running Time Comparisons}
\begin{figure*}[!t]
\centering
\subfloat[]{
\label{fig:r_time_sensor}
\includegraphics[width=0.4\linewidth]{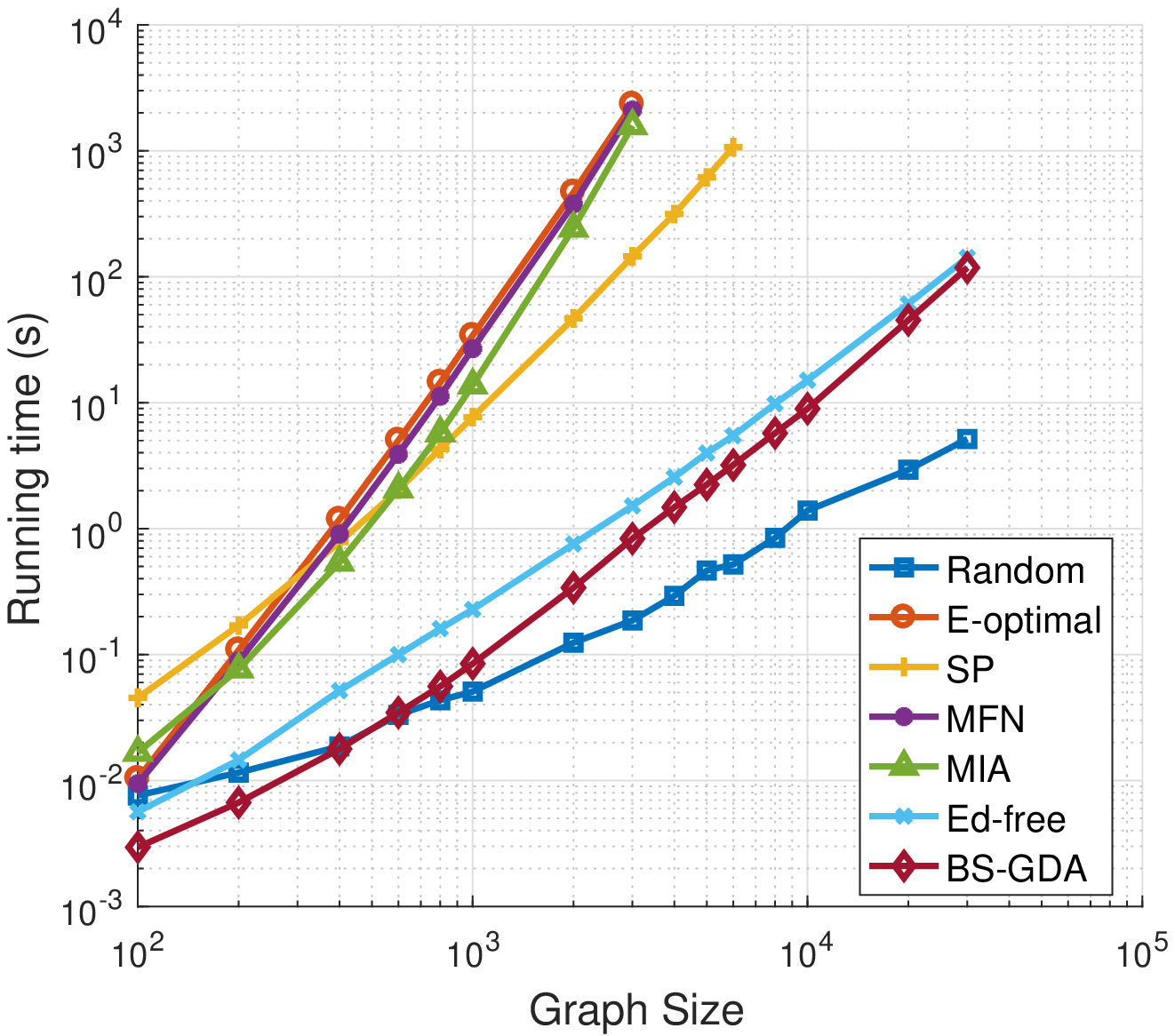}}
\subfloat[]{
\label{fig:r_time_community}
\includegraphics[width=0.4\linewidth]{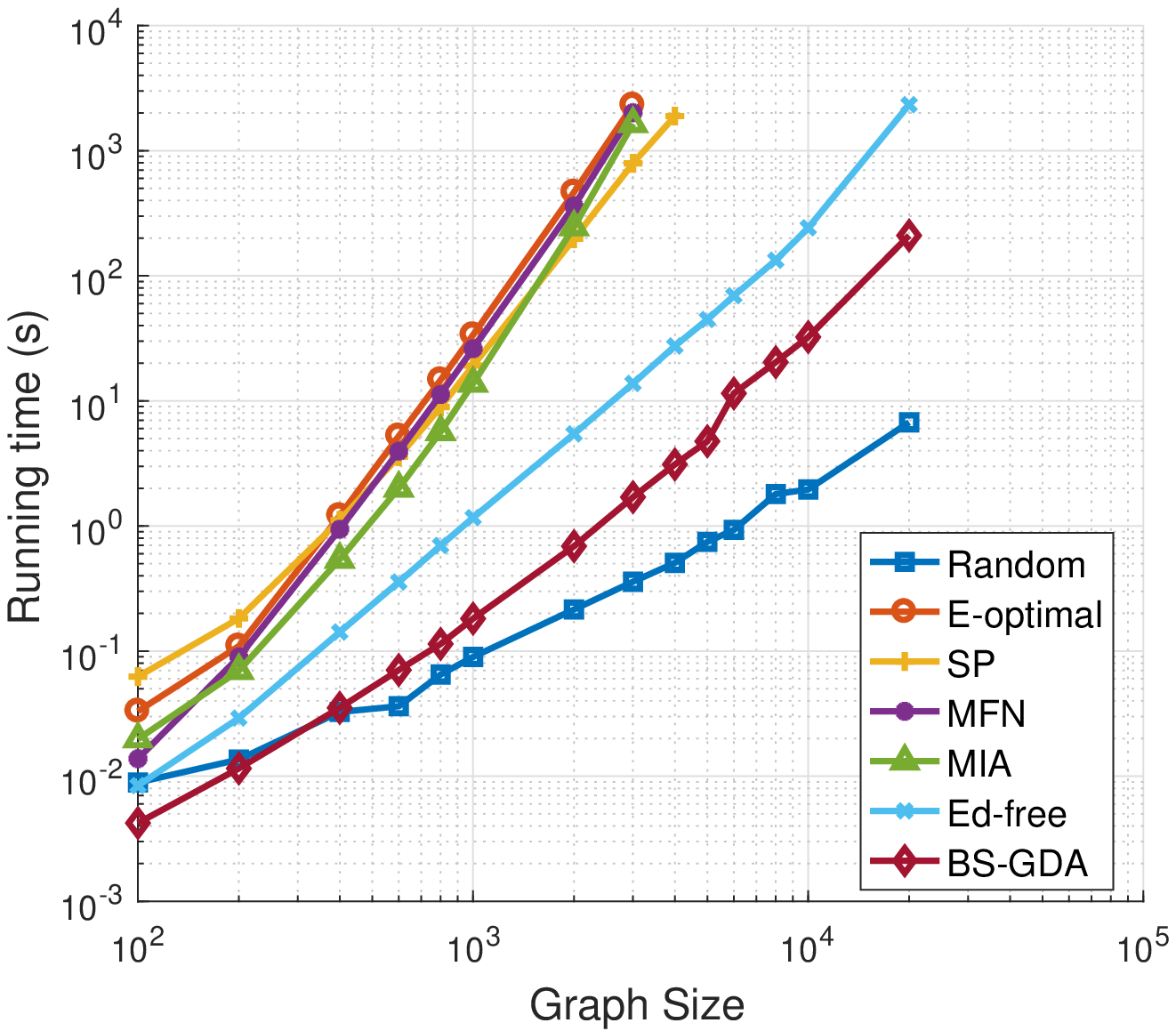}}
\caption{Running time comparisons on two different graphs. (a) Random sensor graph. (b) Community graph. Both axes are represented on logarithmic scales.}
\label{fig:running_time}
\end{figure*}

We compare the running time of our proposed \texttt{BS-GDA} with other competing algorithms on random sensor graphs and community graphs with increasing graph size $N$.
We select sampling budget $K=N/10$ nodes for both graphs. Fig.~\ref{fig:running_time} shows the curves of running time versus graph size $N$.
Specifically, for $N=3000$, we further compute the speedup factor (SF) \cite{akie2018eigenFREE} of our algorithm with respect to the competing algorithms using
\begin{equation}
    SF=\frac{Running\ time\ of\ competing\ method}{Running\ time\ of\ our\ method},
\end{equation}
as summarized in Table.~\ref{tb:speedup}.

\begin{table}[!t]
\caption{Speedup factors of our algorithm with respect to other sampling algorithms for $N=3000$}
\label{tb:speedup}
\centering
\begin{tabular}{C{9em}|C{6em}C{6em}}
\hline
  Sampling Algorithms & Sensor & Community \tabularnewline
\hline
\hline
   \texttt{Random} \cite{r_sampling2018ACHA} & 0.22    & 0.21 \tabularnewline
    \texttt{E-optimal} \cite{e_optimal2015}   & 2812.77 & 1360.76 \tabularnewline
    \texttt{SP} \cite{sp_proxy2016}           & 174.09  & 466.18 \tabularnewline
    \texttt{MFN} \cite{MFN2016TSP}            & 2532.91 & 1184.23 \tabularnewline
    \texttt{MIA} \cite{MIA2018Fen}            & 1896.19  & 964.65 \tabularnewline
    \texttt{Ed-free} \cite{akie2018eigenFREE} & 1.82    & 8.11 \tabularnewline
\hline
\end{tabular}
\end{table}

We observe that \texttt{BS-GDA} is the fastest deterministic sampling set selection algorithm for both graphs.
As mentioned, \texttt{E-optimal}, \texttt{SP} and \texttt{MFN} all require computation of extreme eigen-pairs of the graph Laplacian matrix (or submatrix), which is computationally demanding.
Though there is no explicit computation of eigen-pairs in \texttt{MIA}, the large matrix series computation also results in large computation cost.
Thus, \texttt{BS-GDA} is hundreds to thousands times faster than \texttt{E-optimal}, \texttt{SP}, \texttt{MFN} and \texttt{MIA}.

The recent eigen-decomposition-free method, \texttt{Ed-free} \cite{akie2018eigenFREE}, utilizes the localized operator to design an intuitive sampling strategy, whose complexity is $\cO(NK\bar{Q})$ in Table\;\ref{complexity}.
Our \texttt{BS-GDA} is faster than \texttt{Ed-free} for both graphs, in particular on the community graph.
The reason is that edges within each community are relatively dense, and thus $\bar{Q}$ for community graph is much larger than $\bar{Q}$ for random sensor graph, leading to a slower speed for \texttt{Ed-free}.
In contrast, the computation efficiency of \texttt{BS-GDA} is similar on both graphs.

\texttt{Random} \cite{r_sampling2018ACHA} is faster than our proposed \texttt{BS-GDA}, but its reconstruction error performance is noticeably worse than \texttt{BS-GDA}, as shown in the following section.

\subsection{Reconstruction Error Comparisons}
\label{subsec:recon_mse}
We also compare the resulting reconstruction MSEs between our proposed algorithm and the competing ones.
We choose $N=500$ for random sensor graphs, community graphs and Barab{\'{a}}si-Albert graphs.
Minnesota road graph is of fixed size $N=2642$.
All the graph sampling set selection algorithms are compared on these four graphs with increasing sampling budgets.

For each graph, we consider two types of graph signals, as described in Sec.\;\ref{subsec:experimental_design}.
We randomly generate $50$ graph signals for each type, and for each signal we add $50$ different i.i.d. Gaussian noise, resulting in $2500$ noisy signals for each experiment.
For signal type \textbf{GS2}, because the generated signals are not strictly bandlimited, we approximately set the bandwidth to $\omega=\theta_{\lfloor N/10\rfloor}$ for sampling algorithms that require this parameter.

We use our proposed \texttt{BS-GDA} and the competing algorithms to sample the noisy signals. All the competing algorithms are run with the default or the recommended settings in the referred publications.
We then reconstruct original graph signals by solving GLR-based reconstruction (\ref{eq:linear_equation}) with $\mu=0.01$.

\begin{figure*}[!t]
\centering
\subfloat[]{
\label{fig:sensors_bw}
\includegraphics[width=0.24\linewidth]{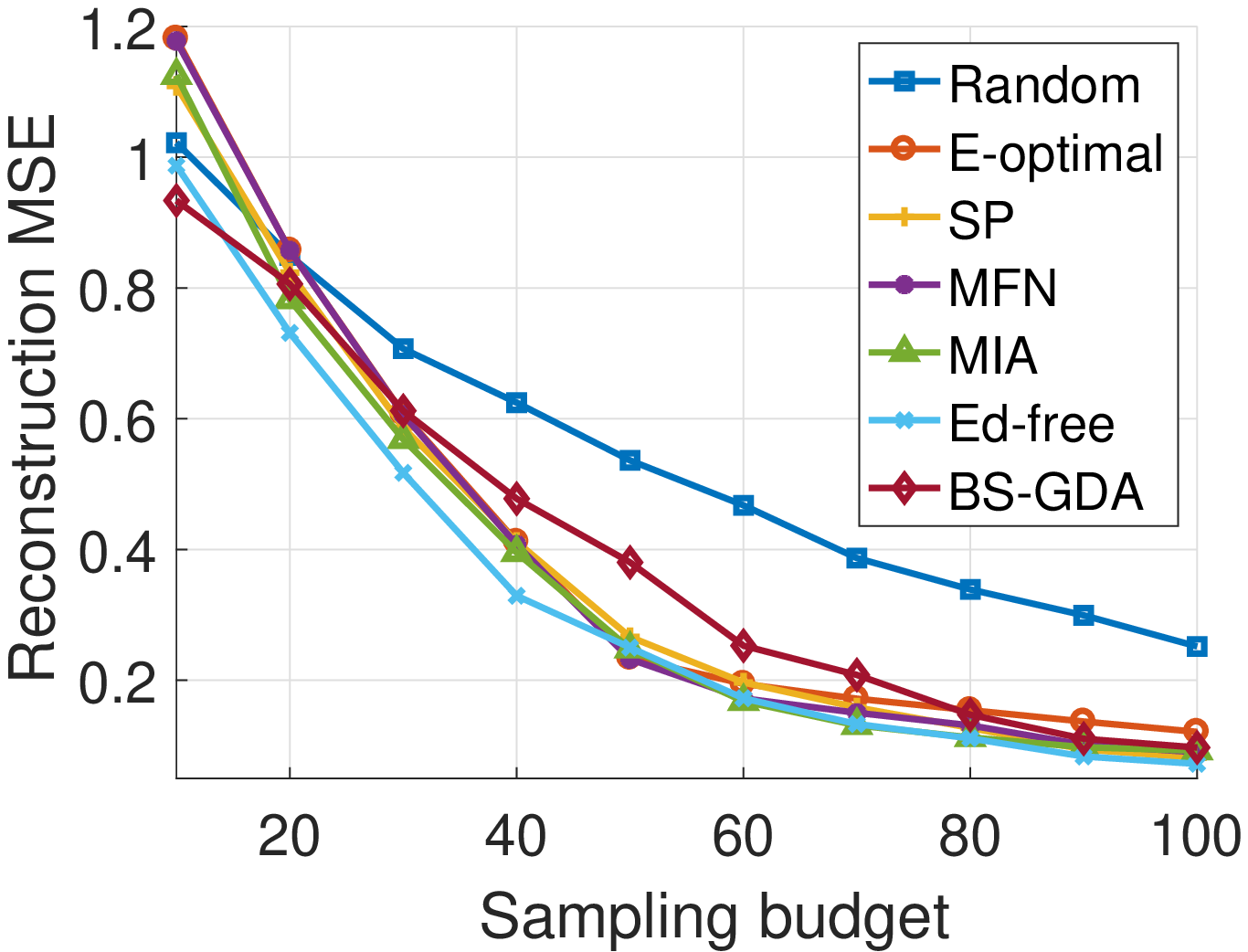}}
\subfloat[]{
\label{fig:community_bw}
\includegraphics[width=0.24\linewidth]{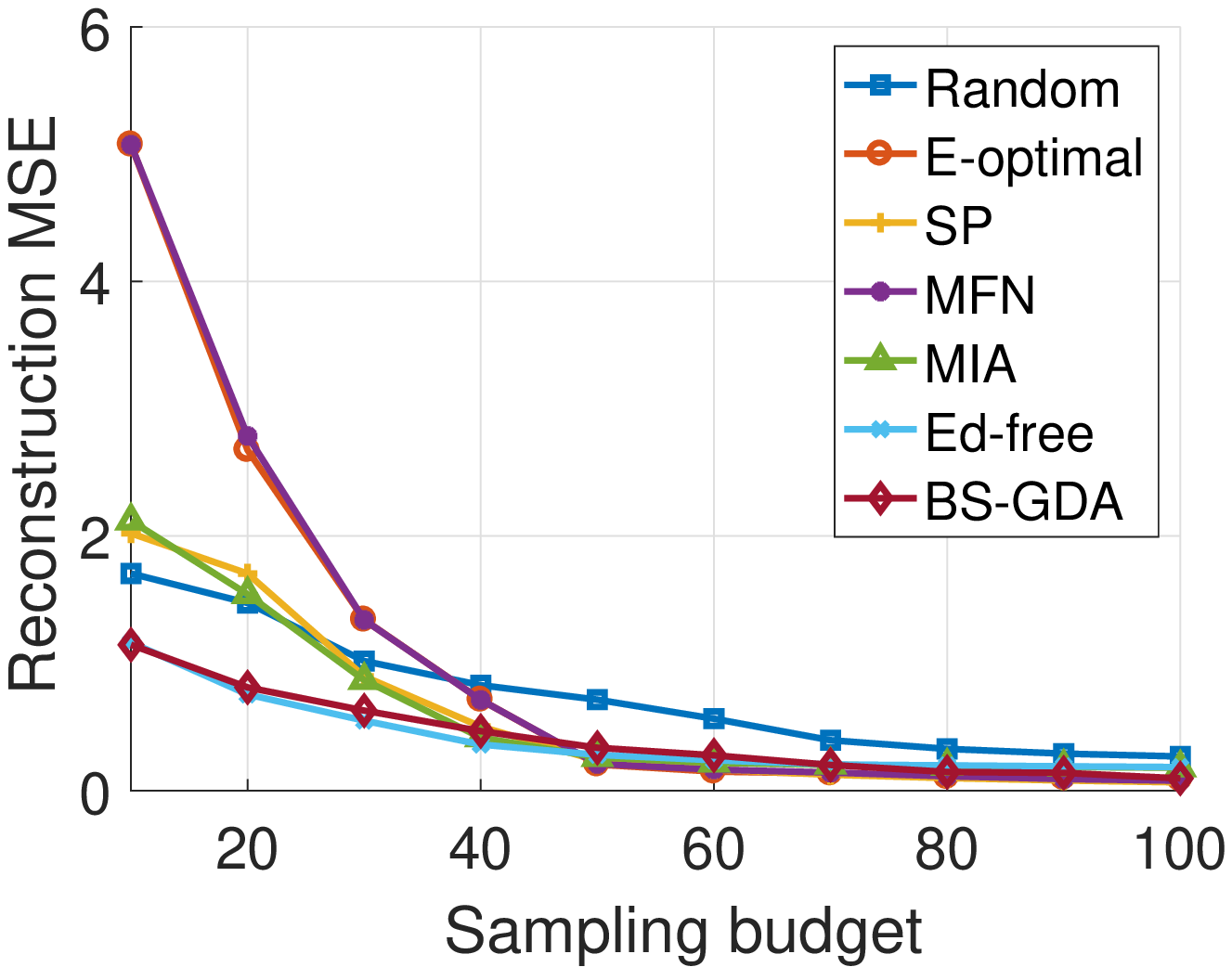}}
\subfloat[]{
\label{fig:ba_bw}
\includegraphics[width=0.24\linewidth]{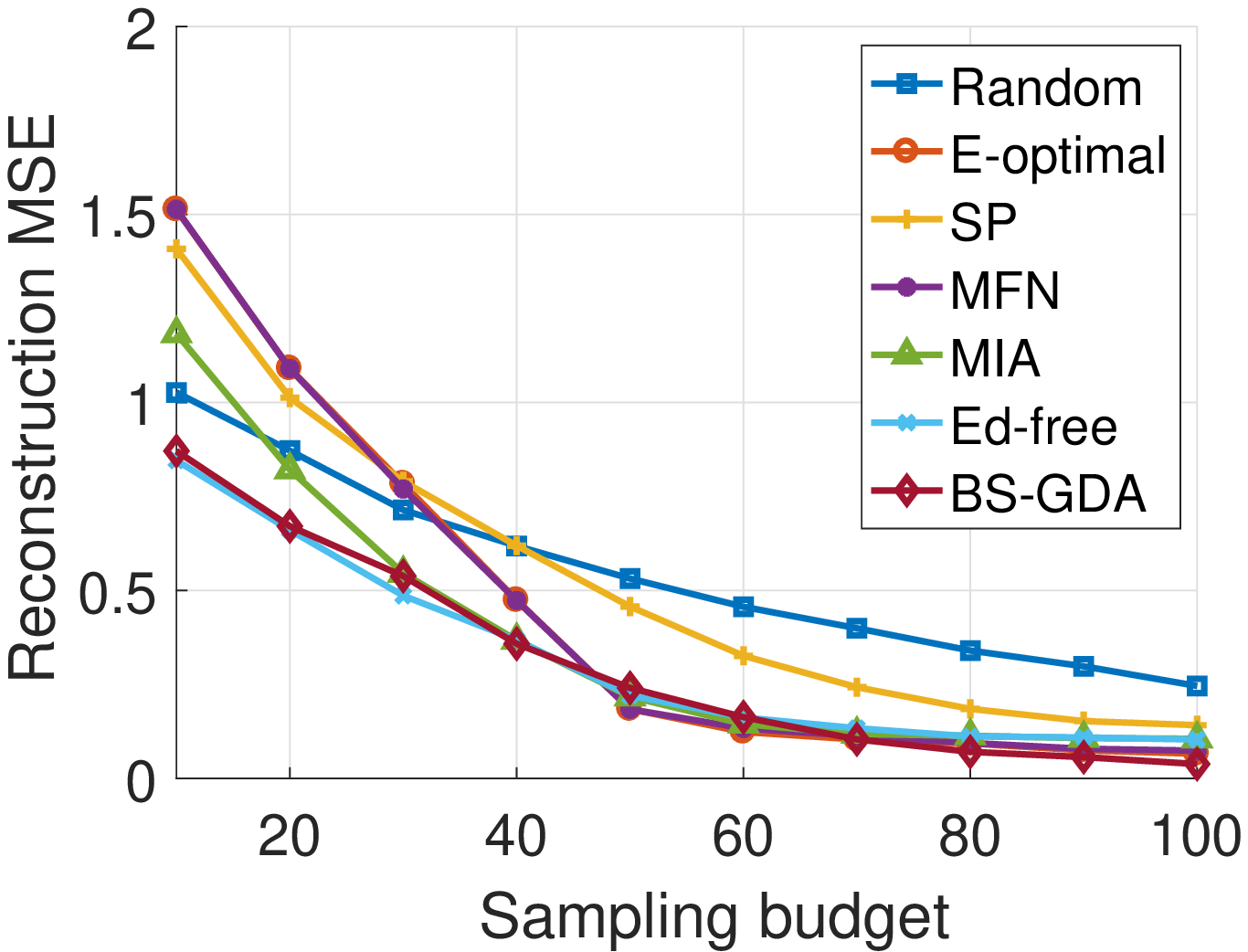}}
\subfloat[]{
\label{fig:minnesota_bw}
\includegraphics[width=0.24\linewidth]{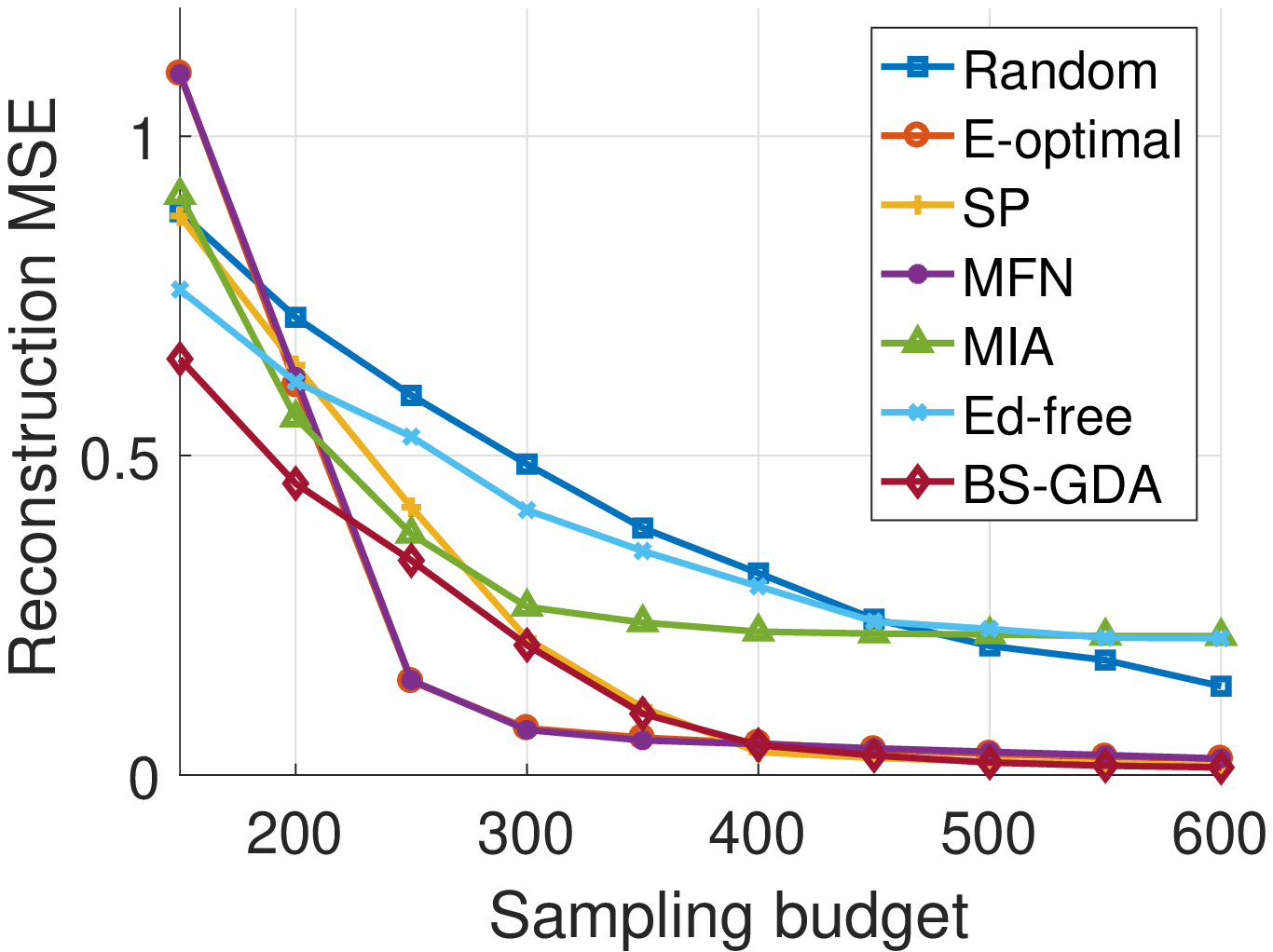}}\\
\subfloat[]{
\label{fig:sensors_smooth}
\includegraphics[width=0.24\linewidth]{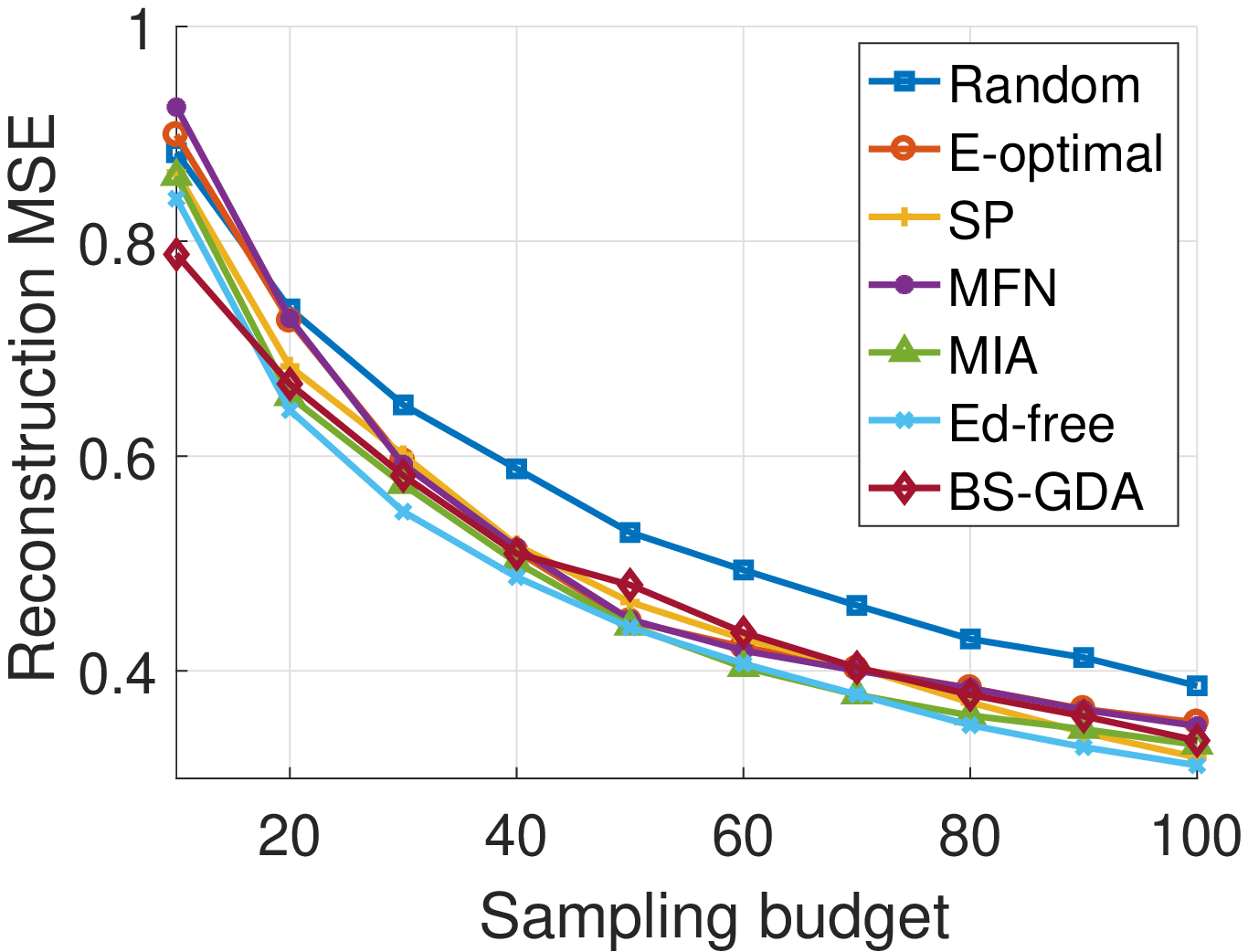}}
\subfloat[]{
\label{fig:community_smooth}
\includegraphics[width=0.24\linewidth]{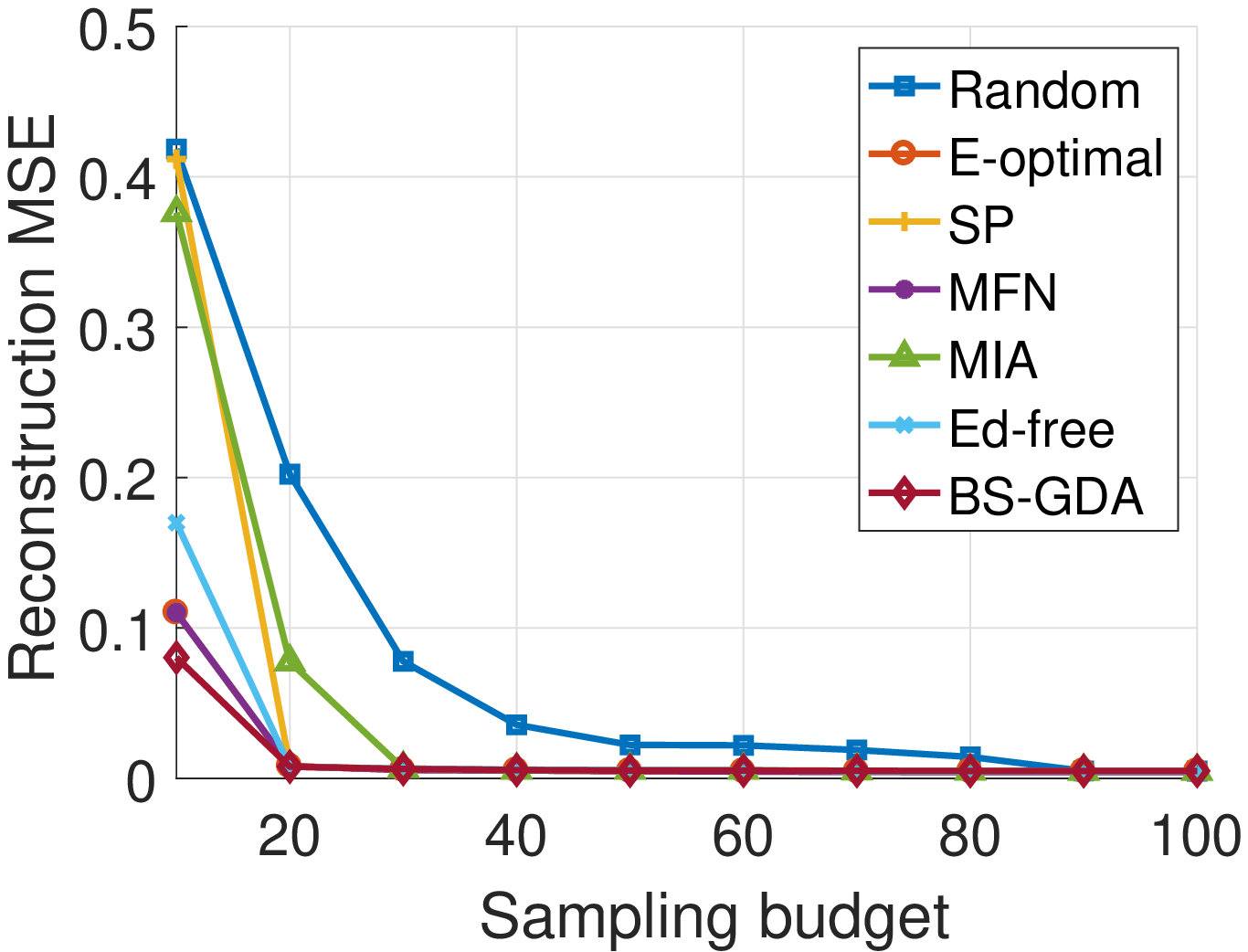}}
\subfloat[]{
\label{fig:ba_smooth}
\includegraphics[width=0.24\linewidth]{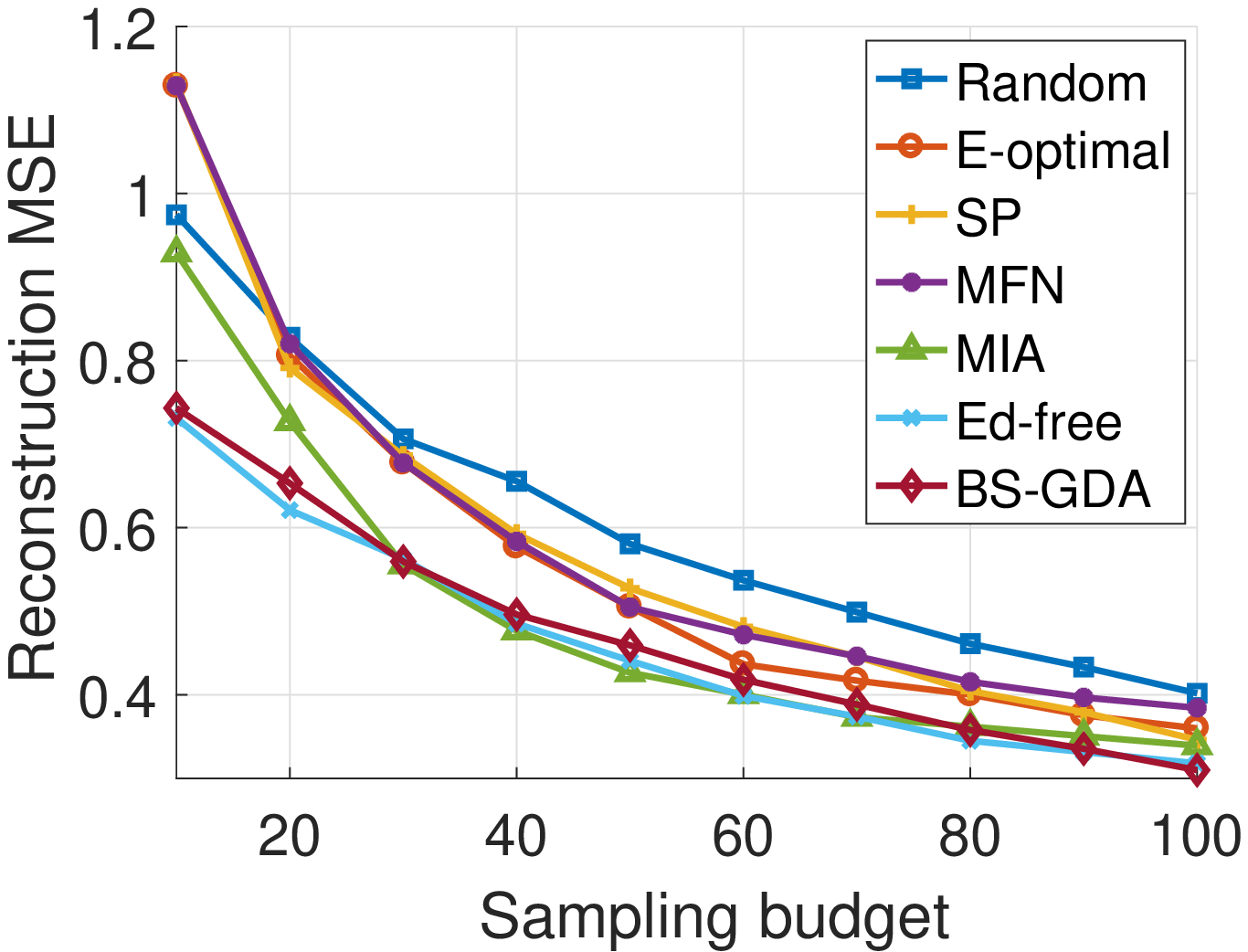}}
\subfloat[]{
\label{fig:minnesota_smooth}
\includegraphics[width=0.24\linewidth]{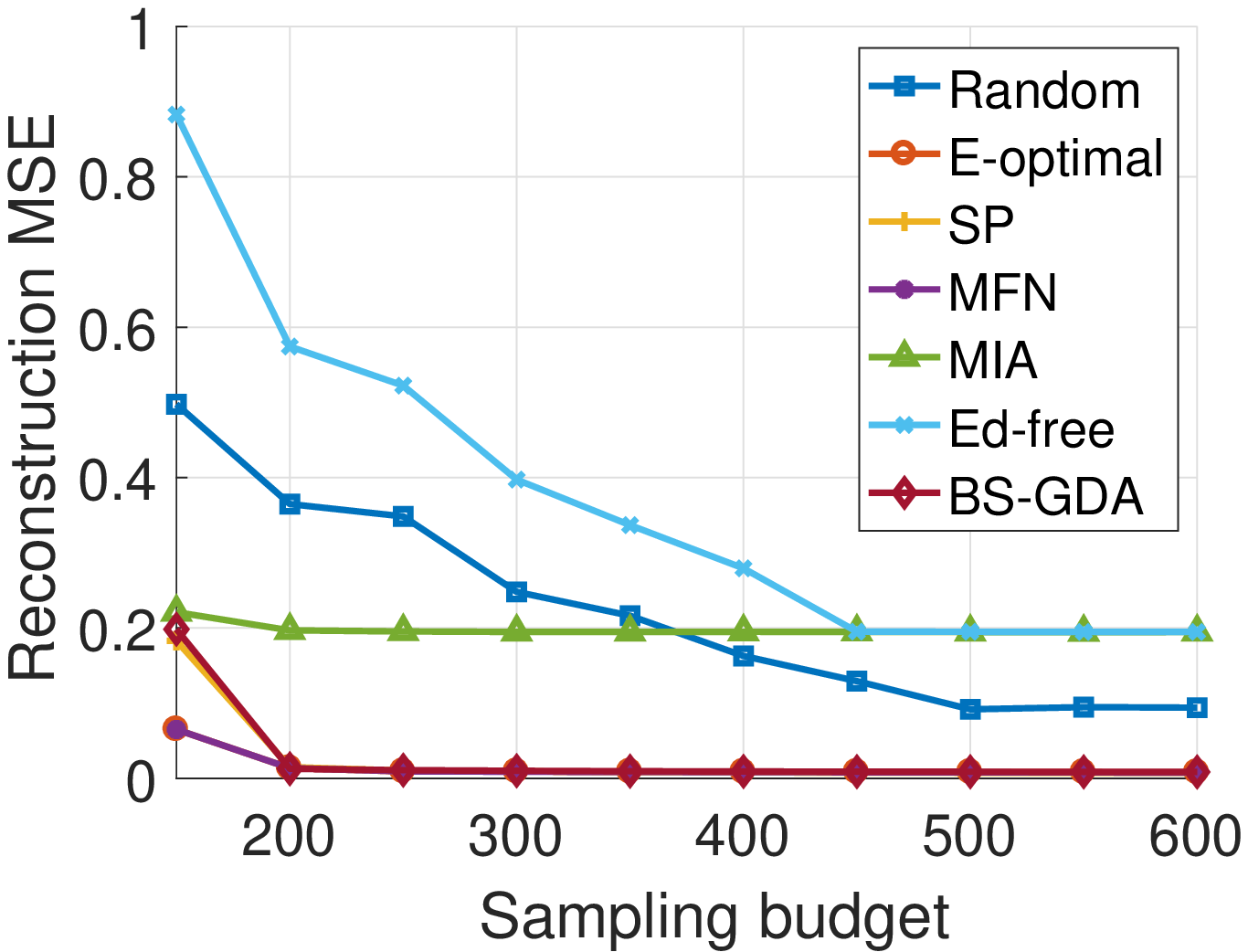}}
\caption{Reconstruction MSEs for different graphs and sampling algorithms using GLR-based reconstruction.
(a) Random sensor graph of size $N=500$ and signal type \textbf{GS1} with bandwidth $\omega=\theta_{50}$.
(b) Community graph of size $N=500$ and $11$ random communities and signal type \textbf{GS1} with bandwidth $\omega=\theta_{50}$.
(c) Barab{\'{a}}si-Albert graph of size $N=500$ and signal type \textbf{GS1} with bandwidth $\omega=\theta_{50}$.
(d) Minnesota road graph of size $N=2642$ and signal type \textbf{GS1} with bandwidth $\omega=\theta_{264}$.
(e) Random sensor graph of size $N=500$ and signal type \textbf{GS2}.
(f) Community graph of size $N=500$ and $11$ random communities and signal type \textbf{GS2}.
(g) Barab{\'{a}}si-Albert graph of size $N=500$ and signal type \textbf{GS2}.
(h) Minnesota road graph of size $N=2642$ and signal type \textbf{GS2}.
}
\label{fig:mse}
\end{figure*}

The average MSEs in terms of sampling budgets are illustrated in Fig.\;\ref{fig:mse}.
As shown, our proposed \texttt{BS-GDA} achieves comparable or smaller MSE values to the competing deterministic schemes on all the four graphs and achieves much better performance than the random sampling scheme \texttt{Random}.

For eigen-decomposition based algorithms, the performances of \texttt{E-optimal}, \texttt{MFN} and \texttt{SP} are unsatisfactory when the sampling budget is much smaller than the bandwidth.
Further, these eigen-decomposition based algorithms are computationally too expensive for large graphs, such as the Minnesota road graph.
For eigen-decomposition-free algorithms, \texttt{MIA} and \texttt{Ed-free} work well on the random sensor graphs, the community graphs and the Barab{\'{a}}si-Albert graphs, but their performance is poor on the large Minnesota road graph.
To avoid explicit eigen-decomposition, symmetrically normalized $\L_n$ is designated for \texttt{MIA} and \texttt{Ed-free}, whose eigenvalues are theoretically bounded in $[0,2]$, to utilize Chebyshev polynomials of $\L_n$ to approximate the low-pass graph filter operator. These two schemes are not robust enough on the graph signals generated on the Minnesota road graph.
It has been observed that, in general, the performance of \texttt{Random} is not comparable to deterministic graph sampling algorithms.

Since the previous deterministic graph sampling algorithms such as \texttt{E-optimal} and \texttt{MFN} are designed assuming a standard least-squares reconstruction, we further test all the sampling algorithms with least-squares reconstruction on random sensor graphs, as shown in Fig.~\ref{fig:ls_mse}.
We set the sampling budget to be no smaller than the bandwidth for \textbf{GS1} or \textbf{GS2}.
As observed in Fig.~\ref{fig:ls_mse}, \texttt{E-optimal}, \texttt{MFN} and \texttt{SP} achieve robust performance.
For \texttt{MIA}, \texttt{Ed-free} and the proposed \texttt{BS-GDA}, their reconstruction MSEs are somewhat larger when the sampling budget is close to the bandwidth.
As the sampling budget increases, all deterministic sampling algorithms have similar reconstruction MSEs.
In contrast, the performance of \texttt{Random} is noticeably worse than the deterministic sampling algorithms.

\begin{figure}[!t]
\centering
\subfloat[]{
\label{fig:ls_sensors_bw}
\includegraphics[width=0.48\linewidth]{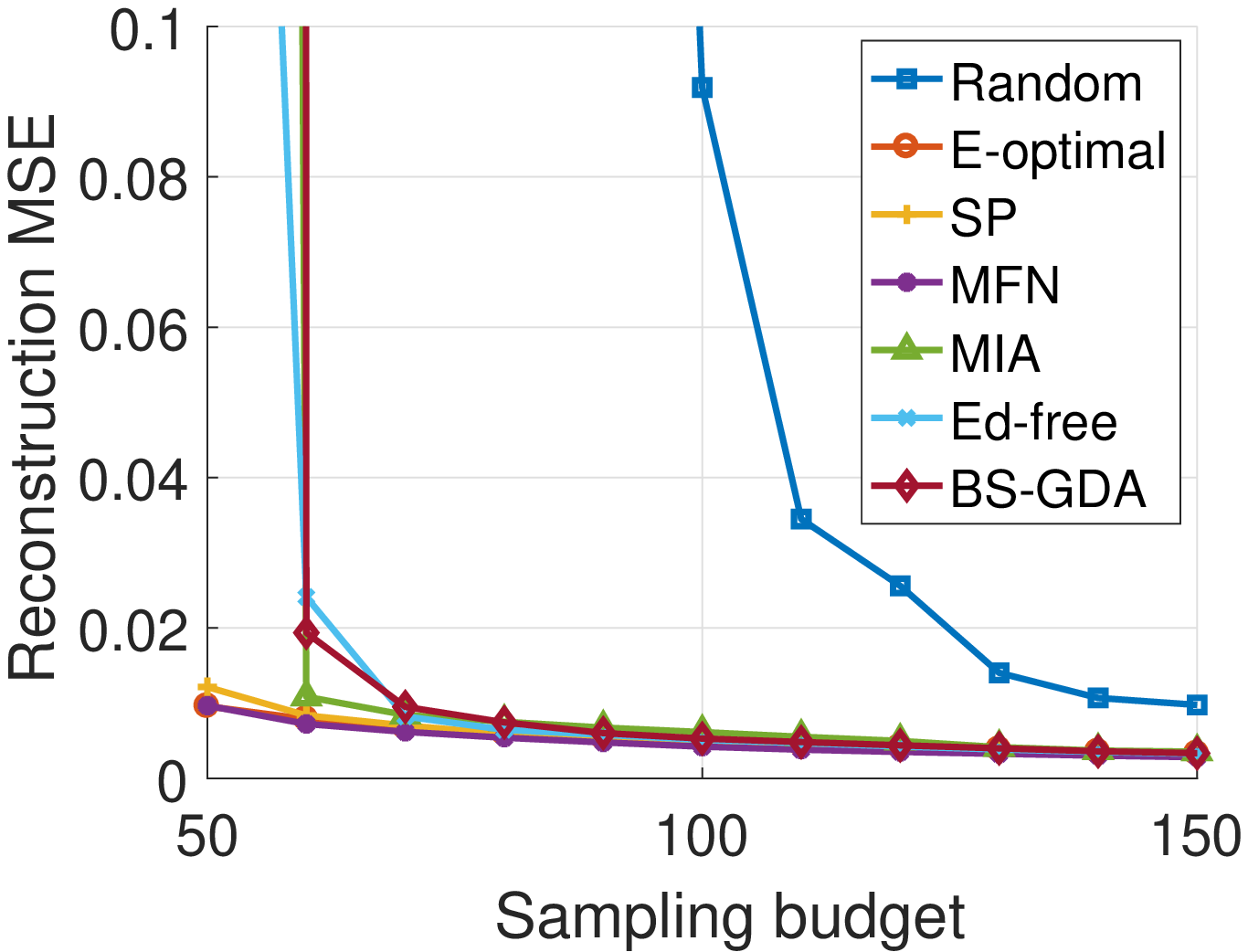}}
\subfloat[]{
\label{fig:ls_sensors_smooth}
\includegraphics[width=0.48\linewidth]{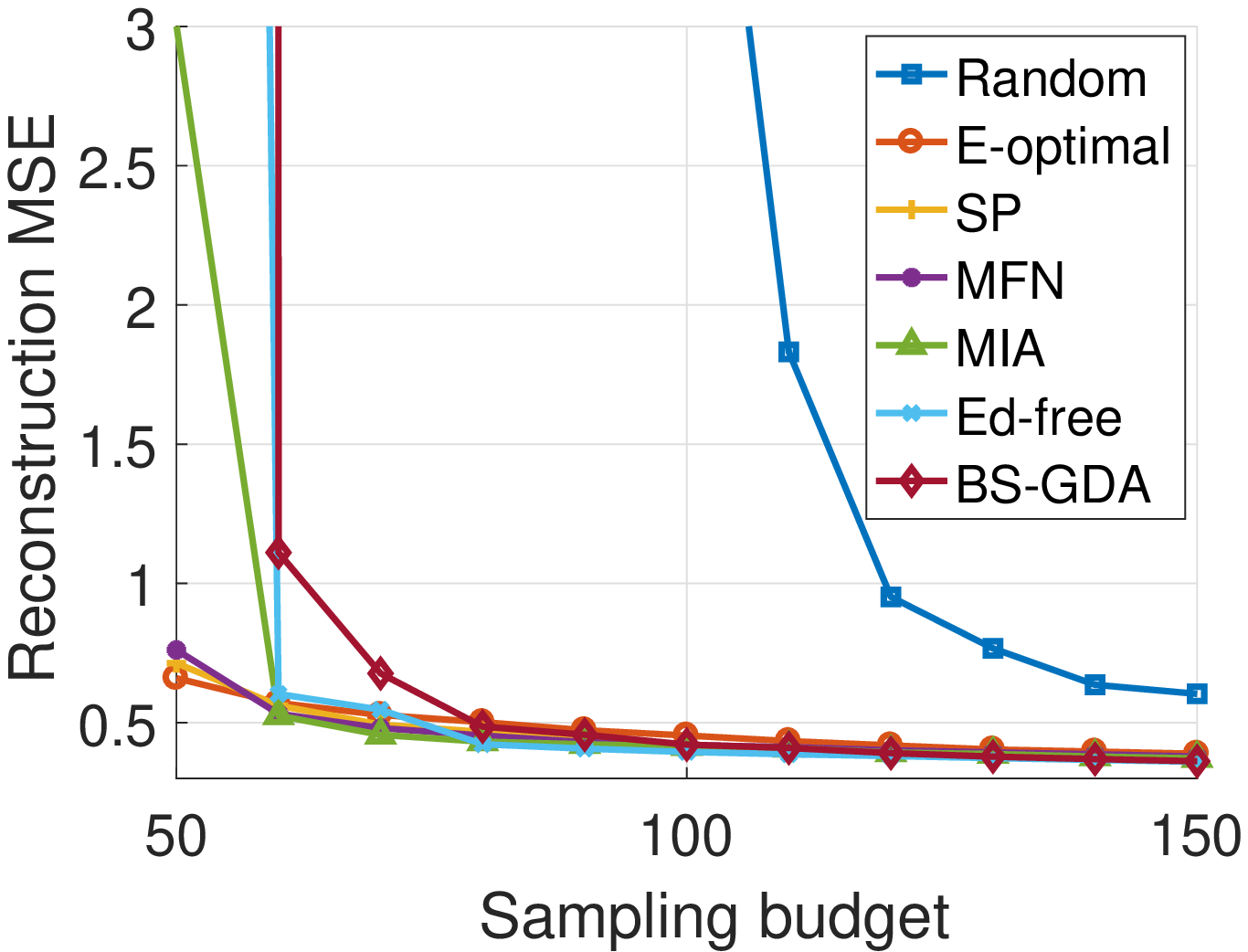}}
\caption{Reconstruction MSEs for different sampling algorithms using least-squares reconstruction on random sensor graph.
(a) Random sensor graph of size $N=500$ and signal type \textbf{GS1} with bandwidth $\omega=\theta_{50}$.
(b) Random sensor graph of size $N=500$ and signal type \textbf{GS2} assuming bandwidth $\omega=\theta_{50}$.
}
\label{fig:ls_mse}
\end{figure}

The proposed \texttt{BS-GDA} solves the dual (\ref{eq:disc_alignment}) via greedy set-covering and performs a binary search to find the optimal lower bound of the primal (\ref{eq:max_lower_bound}), which reduces the global upper bound of the MSE value.
Hence, \texttt{BS-GDA} guarantees robust performance for different graphs in the experiments.
We further show an illustrative example of sampling $K=11$ nodes on the community graph of size $N=500$ with 11 communities in Fig.\;\ref{fig:visualization}.
This illustration clearly shows that \texttt{BS-GDA} captures the graph structure very well and evenly samples one node in each community.
This matches our sampling analogy to the set cover (SC) problem, where each sampled node covers many uncovered nodes.

\begin{figure*}[!t]
\centering
\subfloat[]{
\label{fig:v_graph}
\includegraphics[width=0.23\linewidth]{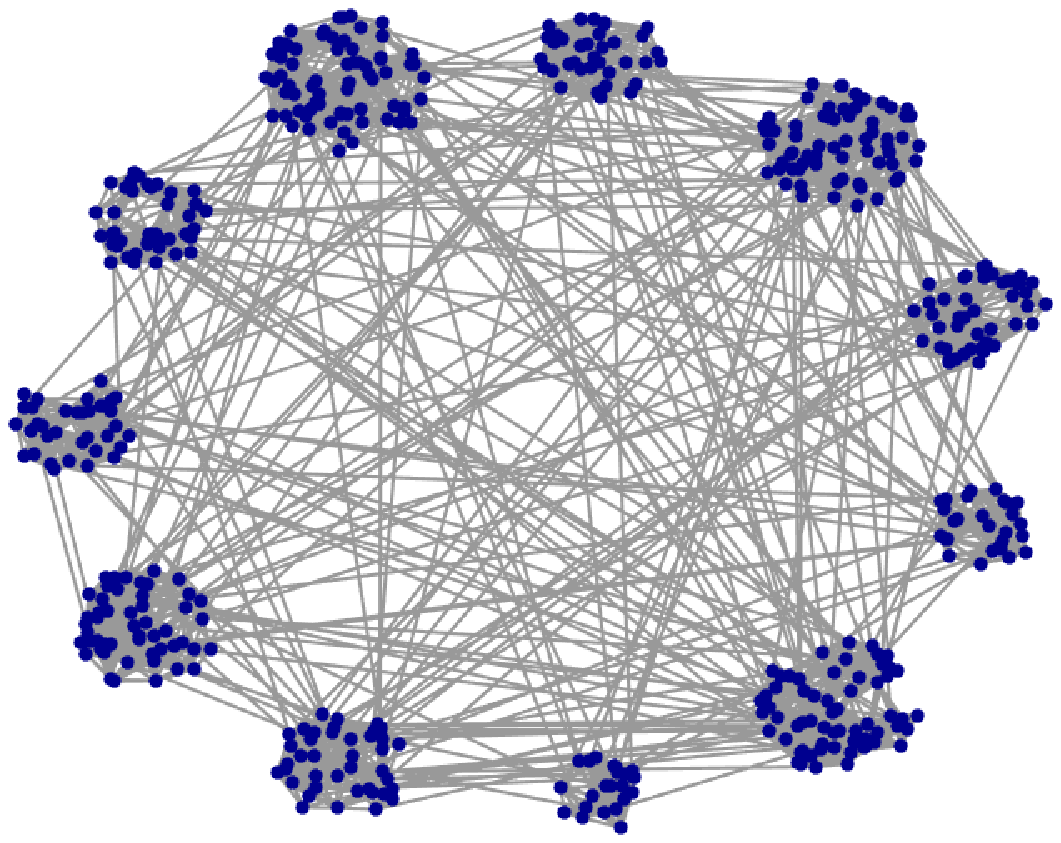}}
\subfloat[]{
\label{fig:v_random}
\includegraphics[width=0.23\linewidth]{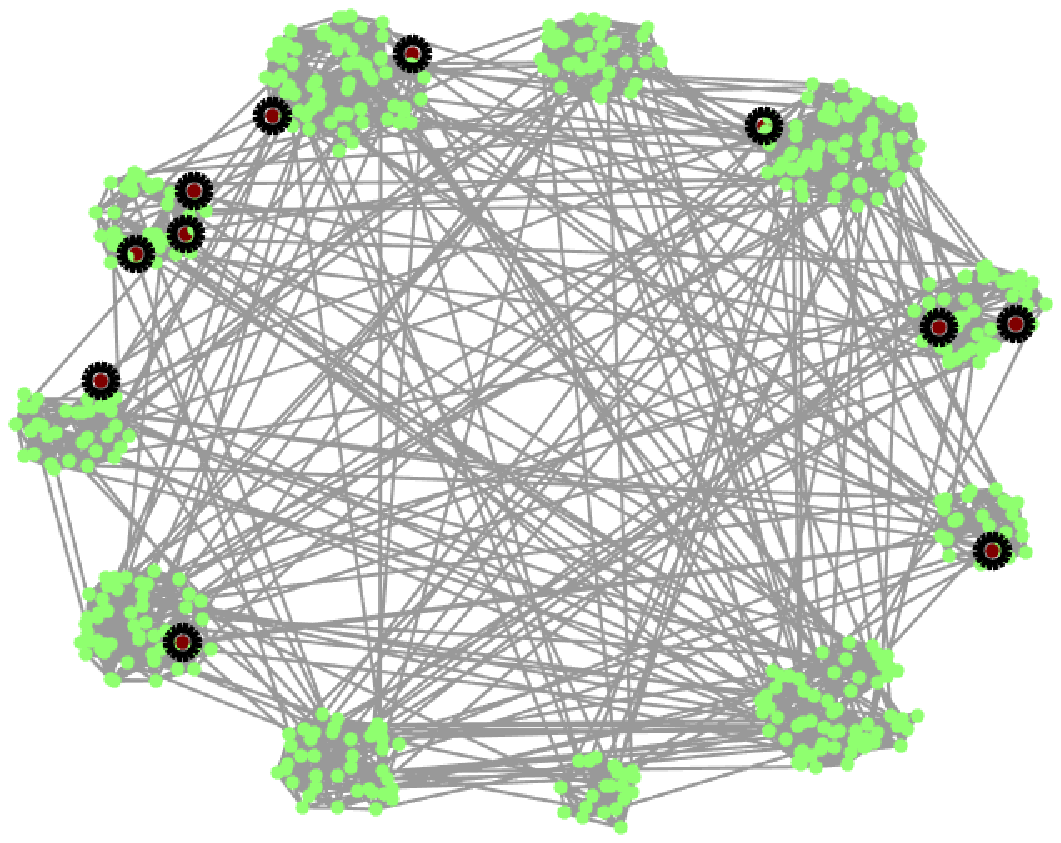}}
\subfloat[]{
\label{fig:v_e_optimal}
\includegraphics[width=0.23\linewidth]{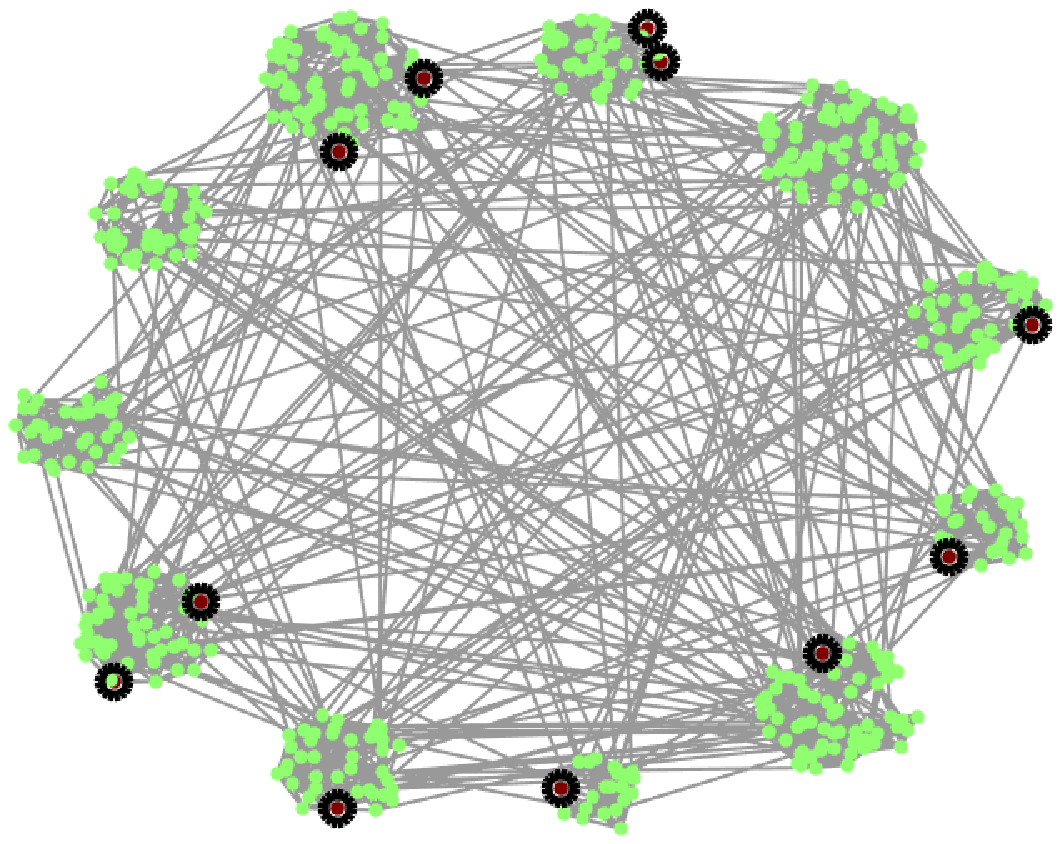}}
\subfloat[]{
\label{fig:v_sp}
\includegraphics[width=0.23\linewidth]{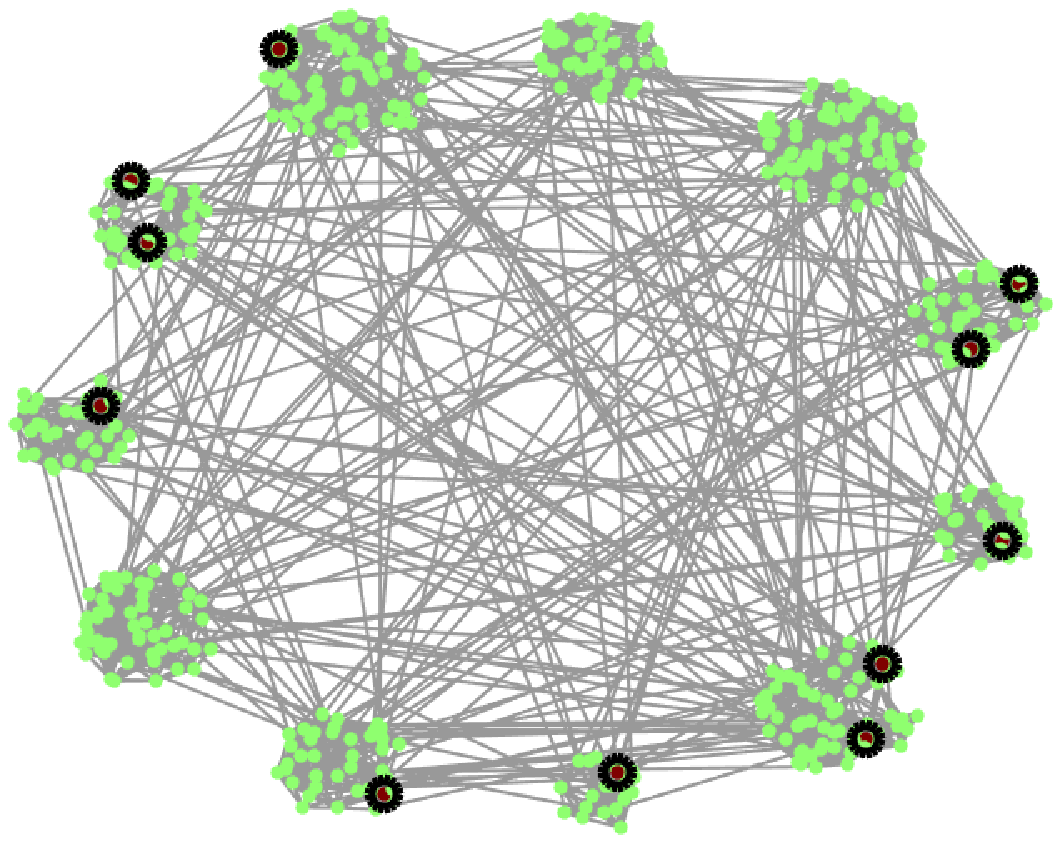}}\\
\subfloat[]{
\label{fig:v_mfn}
\includegraphics[width=0.23\linewidth]{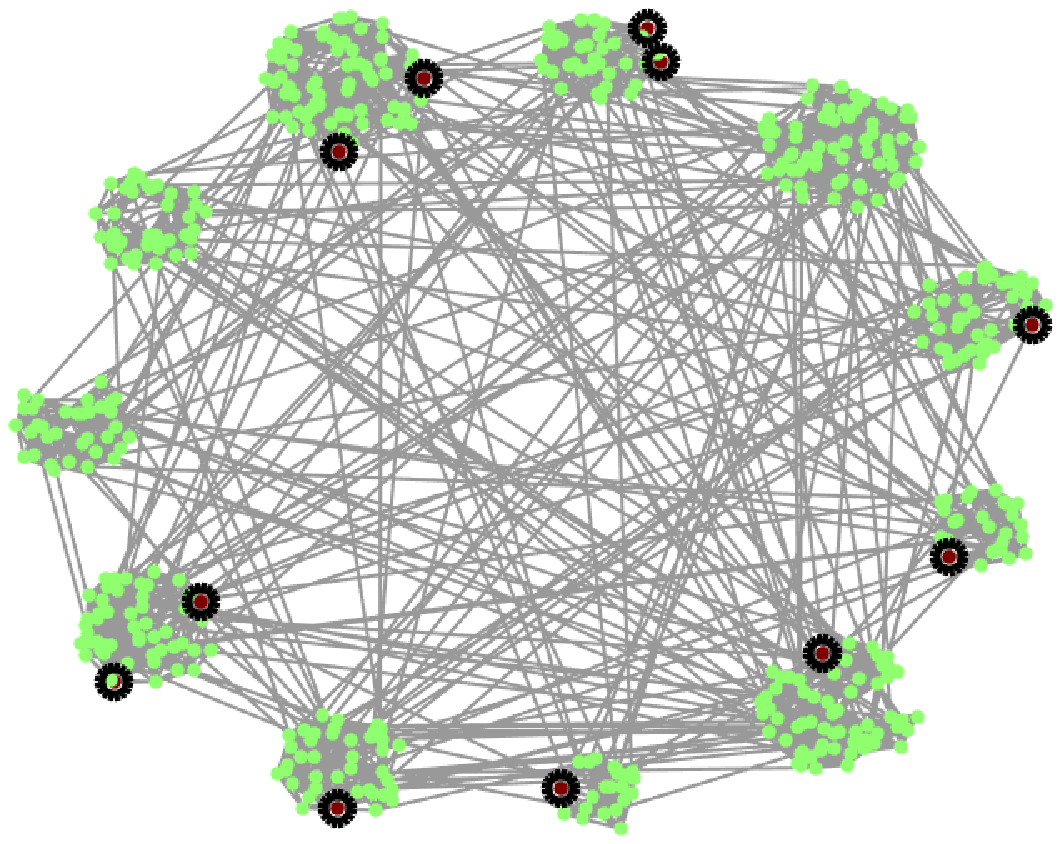}}
\subfloat[]{
\label{fig:v_mia}
\includegraphics[width=0.23\linewidth]{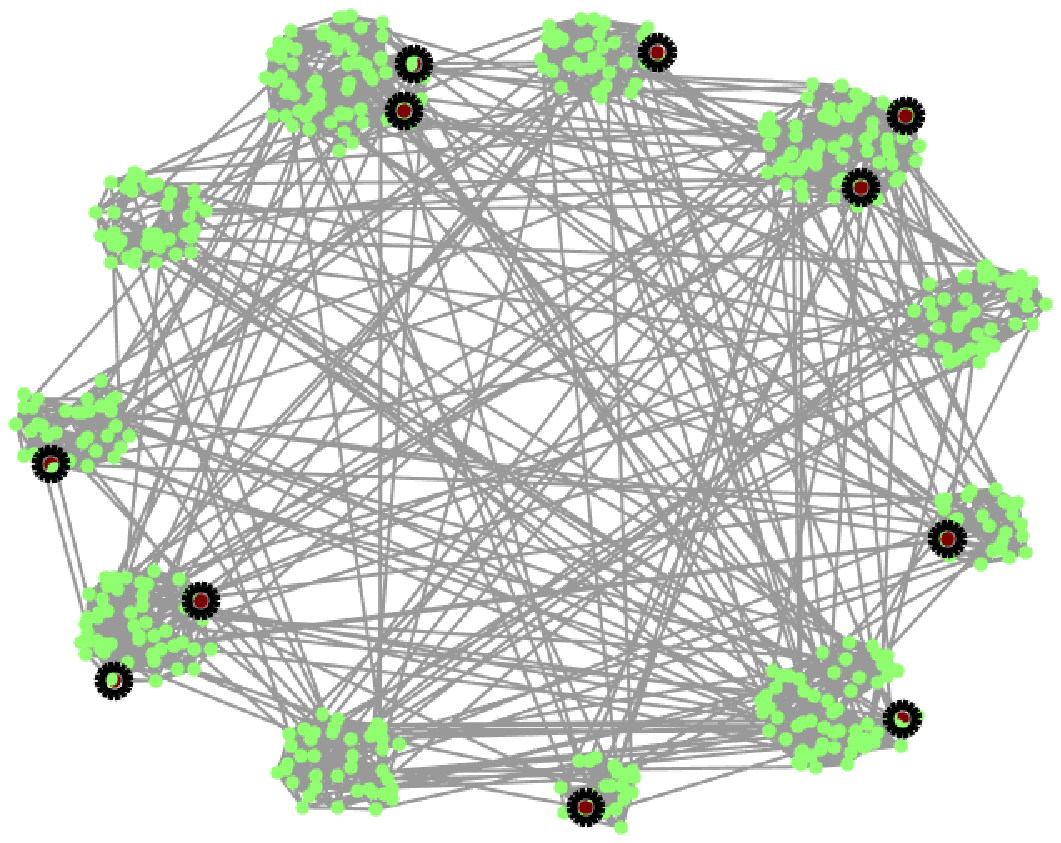}}
\subfloat[]{
\label{fig:v_ed_free}
\includegraphics[width=0.23\linewidth]{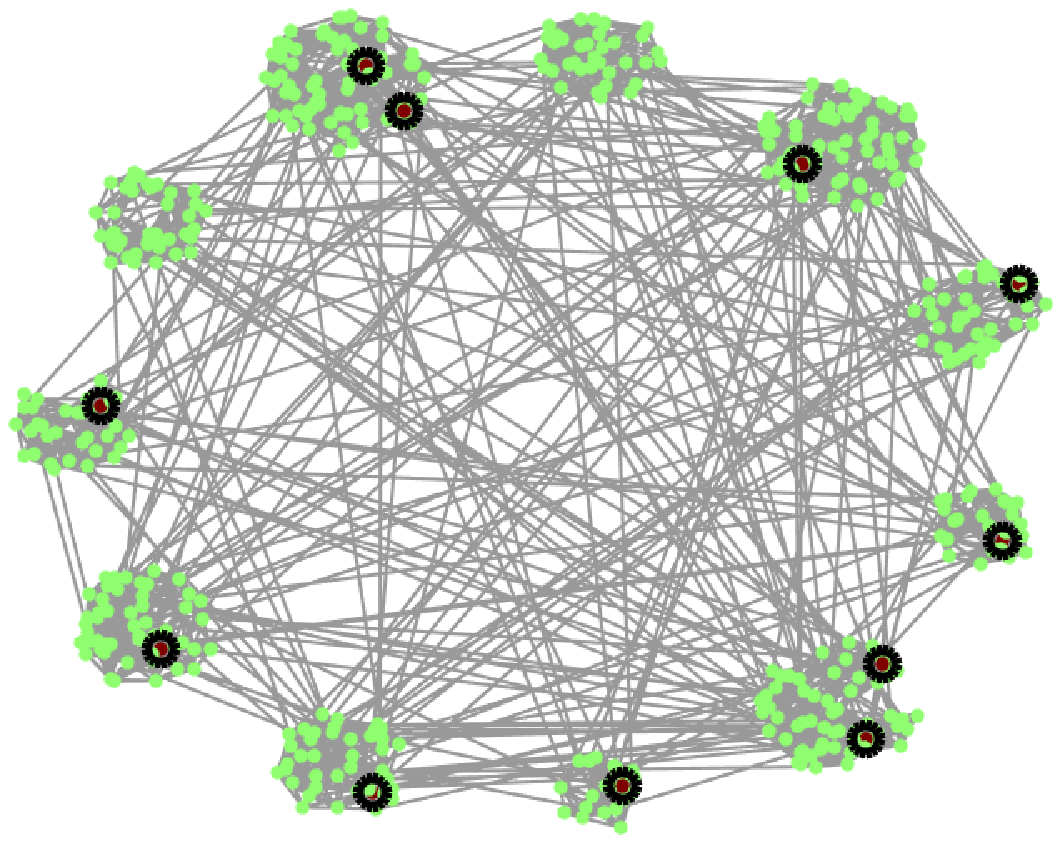}}
\subfloat[]{
\label{fig:v_gda}
\includegraphics[width=0.23\linewidth]{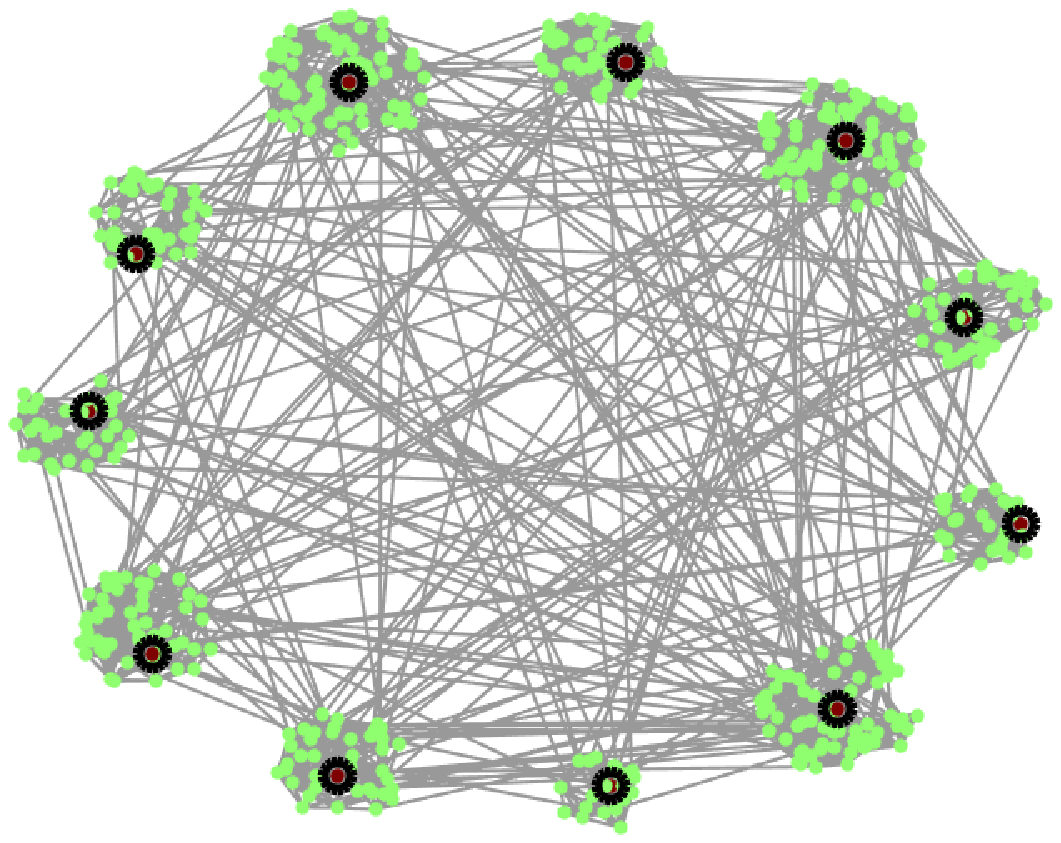}}
    \caption{Visualization of selected nodes on the community graph $(N=500, K=11)$. Black circles denote sampled nodes. (a) Original graph. (b) \texttt{Random} \cite{r_sampling2018ACHA}. (c) \texttt{E-optimal} \cite{e_optimal2015}. (d) \texttt{SP} \cite{sp_proxy2016}. (e) \texttt{MFN} \cite{MFN2016TSP}. (f) \texttt{MIA} \cite{MIA2018Fen}. (g) \texttt{Ed-free} \cite{akie2018eigenFREE}. (h) The proposed \texttt{BS-GDA}.
}
\label{fig:visualization}
\end{figure*}

\section{Conclusion}
\label{sec:conclusion}
To address the ``graph sampling with noise'' problem, in this paper, we first assume a biased graph Laplacian regularization (GLR) based signal reconstruction scheme from samples, and then propose a fast graph sampling set selection algorithm to maximize the stability of the resulting linear system.
In particular, our proposed \texttt{BS-GDA} maximizes the lower-bound of $\lambda_{\min}$ of the coefficient matrix via Gershgorin disc alignment, which is proven to mean also minimizing the upper-bound of reconstruction MSE.
Unlike most previous graph sampling schemes, our proposed \texttt{BS-GDA} does not require computation of any eigen-pairs.
Extensive simulation experiments show that our disc-based sampling algorithm runs substantially faster than existing eigen-decomposition-based schemes (hundreds to thousands times faster) for medium-sized graphs and outperforms other eigen-decomposition-free schemes in reconstruction error for large-sized graphs.

\vspace{20pt}

\appendices

\section{Proof of Proposition~\ref{pr:yuji_bound}}
\label{ap:yuji}
\begin{proof}
    Since matrix $\C$ is similar to $\B$, we have $\lambda_{\min}(\B)=\lambda_{\min}(\C)$. Hence (\ref{eq:eigenvalue_bound}) is equivalent to
    \begin{equation}
        \min_i \ell_i \le\lambda_{\min}(\C)\le \max_i \ell_i
        \label{eq:eigenvalue_bound_2}
    \end{equation}

    From GCT, we see that $\lambda_{\min}(\C)$ is bounded from below by the smallest left-end $\ell_i$, \ie
    \begin{equation}
        \lambda_{\min}(\C)\ge \min_i \ell_i=\min_i \left(c_{ii}-\sum_{j\ne i} |c_{ij}|\right).
    \end{equation}

    We next prove that the largest left-end $\ell_i$ is the upper bound of $\lambda_{\min}(\C)$, \ie,
    \begin{equation}
        ~~~~\lambda_{\min}(\C)\le \max_i \ell_i=\max_i \left(c_{ii}-\sum_{j\ne i} |c_{ij}|\right).
    \end{equation}
    Because $\B=\A+\mu\L$ and $\bS=\diag(\s)$, $\s > \mathbf{0}$, the diagonal elements of $\C$ are non-negative $c_{ii}\ge0$ and the off-diagonal elements of $\C$ are non-positive $c_{ij}\le 0$ for $i\ne j$.
    We consider the matrix $\Ct:=-\C+t\I$, where $t$ is chosen large enough so that the diagonal elements of $\Ct$ are all non-negative, and hence the whole matrix $\Ct$ is non-negative.
    Then by the \textit{Perron-Frobenius Theorem} \cite{horn1990matrix}, $\Ct$ has a positive eigenvalue $\phi_{\max}$ (with a non-negative eigenvector), which is also the eigenvalue of largest absolute value.
    Since the eigenvalues $\lambda$'s of $\C$ and those $\phi$'s of $\Ct$ are related simply by $\phi = -\lambda +t$, it follows that
    \begin{equation}
        \label{eq:lammu}
        \lambda_{\min}(\C) = -\phi_{\max}+t
    \end{equation}
    is the minimum eigenvalue of $\C$.

    Let $\bv=[1,1,\ldots,1]^\top$ and consider the vector (recalling that $\ct_{ij}=-c_{ij}=|c_{ij}|$ for $i\neq j$)
    \begin{equation}
        \Ct \bv
        =
        \begin{bmatrix}
        \ct_{11}+\sum_{j\neq 1}\ct_{1j}\\
        \ct_{22}+\sum_{j\neq 2}\ct_{2j}\\
        \vdots\\
        \ct_{nn}+\sum_{j\neq n}\ct_{nj}\\
        \end{bmatrix}
        =
        \begin{bmatrix}
        -c_{11}+t+\sum_{j\neq 1}|c_{1j}|\\
        -c_{22}+t+\sum_{j\neq 2}|c_{2j}|\\
        \vdots\\
        -c_{nn}+t+\sum_{j\neq n}|c_{nj}|\\
        \end{bmatrix}.
    \end{equation}

    Now by \textit{Collatz-Wielandt' bound} \cite{horn1990matrix},
    we have
    \begin{equation}
        \label{eq:collatz}
        \phi_{\max}\geq \min_i(\Ct \bv )_i = \min_i\left(-c_{ii}+t+\sum_{j\neq i}|c_{ij}|\right).
    \end{equation}
    Together with~\eqref{eq:lammu}, we have
    \begin{equation}
        t-\lambda_{\min}(\C)\geq \min_i(\Ct \bv )_i= \min_i\left(-c_{ii}+t+\sum_{j\neq i}|c_{ij}|\right).
    \end{equation}
    Thus, we conclude that
    \begin{align}
        \lambda_{\min}(\C)&\leq  t-\min_i\left(-c_{ii}+t+\sum_{j\neq i}|c_{ij}|\right) \notag \\
        &=\max_i\left(c_{ii}-\sum_{j\neq i}|c_{ij}|\right) = \max_i\ell_i,
    \end{align}
    as required.
\end{proof}

\section{Proof of NP-Completeness}
\label{ap:np}

We prove that the decision version of the \textit{disc coverage} (DC) problem \eqref{eq:disc_coverage} is NP-complete by transforming the decision version of the \textit{set cover} (SC) problem \cite{cormen2009introduction} to a special case of DC.
The SC decision problem asks the following binary decision question:
given a family of subsets $\cF$ of a finite element set $\cX$ and positive integer $K \leq |\cF|$, does there exist a \textit{cover} $\cC$ of size $K$ or less, \textit{i.e.}, $\cC \subseteq \cF$ and $|\cC| \leq K$, where each element $x \in \cX$ is in at least one member of $\cC$?

For each instance of SC, we perform the following transformation.
First, for each element $j$ in set, $x_j \in \cX$, we construct one \textit{element node} $j$ in a graph $\cG$.
Next, for each subset in collection $c_i \in \cF$, we construct one \textit{subset node} $|\cX| + i$ and connect it to element nodes $j$ that correspond to elements $x_j \in c_i$.
We construct one \textit{super node} $|\cX|+|\cF|+1$ and connect it to $|\cF|$ subset nodes.
Finally, we construct $K+1$ \textit{extra nodes} and connect them to the super node.
Coverage subset $\Omega_i$ for each node $i$ is set to be all 1-hop neighbors of node $i$ in the constructed graph.

As an example, consider the following SC instance:
\begin{align}
\cX &= \{1, 2, 3, 4, 5\}, ~~ \cF = \{c_1, c_2, c_3, c_4\} \nonumber \\
c_1 &= \{1, 2, 4\}, ~~ c_2 = \{1, 4, 5\} \nonumber \\
c_3 &= \{3, 5\}, ~~ c_4 = \{3, 4\}
\end{align}
Fig.\;\ref{fig:dual sampling} illustrates the corresponding constructed graph.
The corresponding DC decision problem is:
does there exist a sampling vector $\a$ with objective of \eqref{eq:disc_coverage} $K+1$ or less, while having each graph node $i$ included in at least one sample $i$'s coverage subset $\Omega_i$?

We prove that the SC decision problem is true iff the corresponding DC decision problem is true.
Suppose first that there exists a cover $\cC$ of size $|\cC| \leq K$.
We observe that a corresponding sample selection of super node $|\cX| + |\cF| + 1$ and subset nodes representing members $c_i \in \cC$ will mean that each graph node is included in at least one sample coverage subset $\Omega_i$, \textit{i.e.}, all nodes are within one hop of selected samples.

\begin{figure}[!t]
    \centering
    \includegraphics[width=0.8\linewidth]{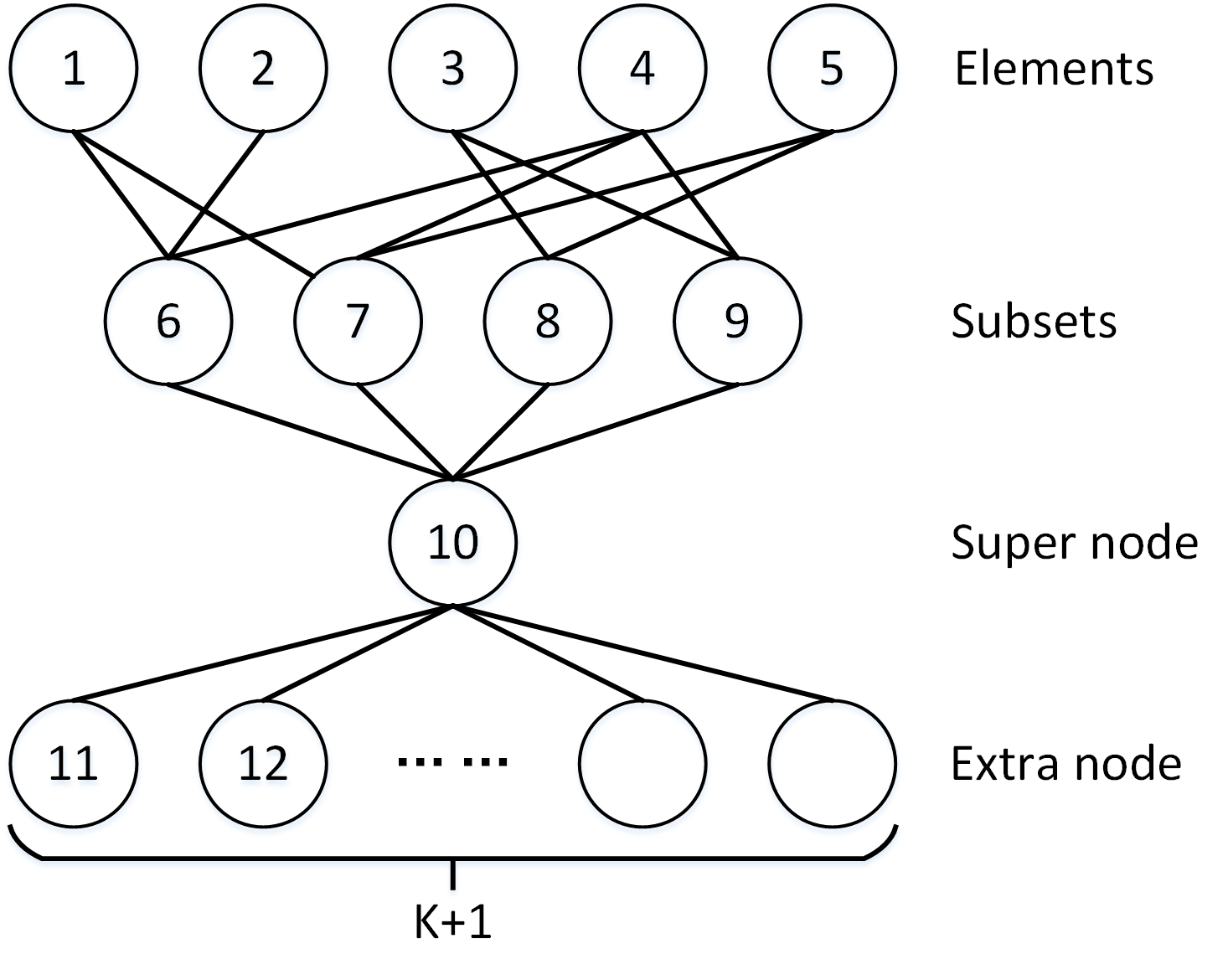}
    \caption{Example of a mapped dual graph sampling problem from a SC problem.}
    \label{fig:dual sampling}
\end{figure}

Now suppose that there exists a sample selection of size $K+1$ such that all nodes are within one hop of selected samples.
We first observe that the super node must be selected.
Otherwise, just covering the $K+1$ extra nodes would require $K+1$ node selections, leaving no leftover budget to select and cover element and subset nodes.

Given that the super node is selected, there remain $K$ sample budget to cover $|\cX|$ element nodes.
We can identify a cover of $K'$ subsets in SC, $K' \leq K$, from the chosen $K$ node samples as follows.
Suppose there are $K_1$, $K_2$ and $K_3$ selected element, subset and extra nodes respectively, where $K_1 + K_2 + K_3 = K$.
For each one of $K_2$ chosen subset node, we select the corresponding subset in SC.
For each one of $K_1$ chosen element node, we randomly select one of its connected subset node and choose the corresponding subset in SC (if not already chosen in previous step).
We ignore all $K_3$ chosen extra nodes.
We see that the number of chosen subsets in SC is no larger than $K_1 + K_2 \leq K$, but covers all elements in set $\cX$.
Thus a feasible sample set for the DC decision problem maps to a feasible cover $\cC$ for the SC decision problem.
We have thus proven that an instance of the SC problem maps to our construction of the DC problem.
Therefore, we can conclude that the DC problem is at least as hard as the SC problem.
Since the SC decision problem is known to be NP-complete, the DC decision problem is also NP-complete.

Given the DC decision problem is NP-complete, the corresponding DC optimization problem in \eqref{eq:disc_coverage} is NP-hard. $\Box$

\ifCLASSOPTIONcaptionsoff
  \newpage
\fi



%
\bibliographystyle{IEEEbib}
\bibliography{ref2,graph_refs}

%








\end{document}